\documentclass[11pt,a4paper]{article}
\pdftrailerid{}
\usepackage[
  margin=2.5cm,
  top=2.5cm,
  bottom=2.5cm
]{geometry}

\usepackage{amsmath,amssymb,amsthm}
\usepackage{mathtools}
\usepackage{bm}

\usepackage{graphicx}
\usepackage{booktabs}
\usepackage{longtable}
\usepackage{tabularx}
\usepackage{array}
\usepackage{multirow}
\usepackage{colortbl}

\usepackage{tikz}
\usetikzlibrary{arrows.meta,positioning,fit,calc}
\usepackage{quantikz}

\usepackage{algorithm}
\usepackage{algpseudocode}

\usepackage{braket}

\usepackage[section]{placeins}
\usepackage{flafter}
\usepackage{needspace}

\usepackage[numbers,sort&compress]{natbib}

\usepackage{xurl}
\usepackage[colorlinks=true,linkcolor=blue!60!black,citecolor=green!50!black,urlcolor=blue!60!black]{hyperref}
\hypersetup{pdfcreator={},pdfproducer={}}
\usepackage{orcidlink}

\usepackage{xcolor}
\newcommand{\TableRowsSetup}{\rowcolors{1}{white}{black!5}}
\newcommand{\TableHideRows}{\hiderowcolors}
\newcommand{\TableBodyRows}{\showrowcolors}
\newcommand{\TableGroupRow}{\rowcolor{black!10}}

\usepackage[nopatch=footnote]{microtype}
\expandafter\newcommand\csname JointSeqNoisePercent0\endcsname{0}
\expandafter\newcommand\csname JointSeqSeqFOne0\endcsname{0.833}
\expandafter\newcommand\csname JointSeqJointFOne0\endcsname{0.833}
\expandafter\newcommand\csname JointSeqDeltaFOne0\endcsname{+0.000}
\expandafter\newcommand\csname JointSeqRelGain0\endcsname{+0\%}
\expandafter\newcommand\csname JointSeqNoisePercent20\endcsname{20}
\expandafter\newcommand\csname JointSeqSeqFOne20\endcsname{0.854}
\expandafter\newcommand\csname JointSeqJointFOne20\endcsname{0.885}
\expandafter\newcommand\csname JointSeqDeltaFOne20\endcsname{+0.031}
\expandafter\newcommand\csname JointSeqRelGain20\endcsname{+4\%}
\expandafter\newcommand\csname JointSeqNoisePercent40\endcsname{40}
\expandafter\newcommand\csname JointSeqSeqFOne40\endcsname{0.821}
\expandafter\newcommand\csname JointSeqJointFOne40\endcsname{0.955}
\expandafter\newcommand\csname JointSeqDeltaFOne40\endcsname{+0.134}
\expandafter\newcommand\csname JointSeqRelGain40\endcsname{+16\%}
\expandafter\newcommand\csname JointSeqNoisePercent50\endcsname{50}
\expandafter\newcommand\csname JointSeqSeqFOne50\endcsname{0.791}
\expandafter\newcommand\csname JointSeqJointFOne50\endcsname{1.000}
\expandafter\newcommand\csname JointSeqDeltaFOne50\endcsname{+0.209}
\expandafter\newcommand\csname JointSeqRelGain50\endcsname{+26\%}
\expandafter\newcommand\csname JointSeqNoisePercent60\endcsname{60}
\expandafter\newcommand\csname JointSeqSeqFOne60\endcsname{0.752}
\expandafter\newcommand\csname JointSeqJointFOne60\endcsname{0.944}
\expandafter\newcommand\csname JointSeqDeltaFOne60\endcsname{+0.192}
\expandafter\newcommand\csname JointSeqRelGain60\endcsname{+26\%}

\newcommand{\JointSeqMinDeltaFOnePct}{3.1}
\newcommand{\JointSeqNSamples}{20}
\newcommand{\JointSeqNFeatures}{10}
\newcommand{\JointSeqKAnomalies}{3}
\newcommand{\JointSeqMFeatures}{5}
\newcommand{\JointSeqNSeeds}{200}

\newcommand{\JointSeqPracticalMaxDeltaFOnePct}{20.9}

\expandafter\newcommand\csname QAOADepthP1MeanAlpha\endcsname{0.688}
\expandafter\newcommand\csname QAOADepthP1PGeqZeroPointSevenFivePct\endcsname{20\%}
\expandafter\newcommand\csname QAOADepthP2MeanAlpha\endcsname{0.722}
\expandafter\newcommand\csname QAOADepthP2PGeqZeroPointSevenFivePct\endcsname{28\%}
\expandafter\newcommand\csname QAOADepthP3MeanAlpha\endcsname{0.754}
\expandafter\newcommand\csname QAOADepthP3PGeqZeroPointSevenFivePct\endcsname{56\%}
\expandafter\newcommand\csname QAOADepthP4MeanAlpha\endcsname{0.774}
\expandafter\newcommand\csname QAOADepthP4PGeqZeroPointSevenFivePct\endcsname{64\%}

\newcommand{\PenaltyXCandidateEvals}{8}
\newcommand{\PenaltyXShotsPerEval}{512}
\newcommand{\PenaltyXFinalShots}{2{,}048}

\newcommand{\CGXYGroupedComparisonCount}{264}
\newcommand{\CGXYGroupedSameDepthWinCount}{264}
\newcommand{\CGXYGroupedMatchedDepthWinCount}{3}
\newcommand{\CGXYGroupedRecoveryMedianRangePct}{30--36\%}
\newcommand{\CGXYGroupingNullCaseCount}{3{,}072}
\newcommand{\CGXYGroupingNullTypePreservingMedianPercentilePct}{97.0\%}
\newcommand{\CGXYGroupingNullTypePreservingAboveQNinetyFivePct}{54.8\%}
\newcommand{\CGXYGroupingNullTypePreservingGainOverMean}{2.94}
\newcommand{\CGXYGroupingNullSameSizeMedianPercentilePct}{30.5\%}
\newcommand{\CGXYGroupingNullSameSizeAboveQNinetyFivePct}{8.1\%}
\newcommand{\CGXYGroupingNullSameSizeGainOverMean}{0.724}
\newcommand{\CreditLargeEndpointSixtyFourQPTwoFeasPct}{25.5\%}
\newcommand{\CreditLargeEndpointSixtyFourQPTwoChanceRatio}{1{,}875$\times$}
\newcommand{\CreditLargeEndpointThirtySixQPThreeFeasPct}{47.4\%}
\newcommand{\CreditLargeEndpointThirtySixQPThreeChanceRatio}{218$\times$}

\newcommand{\StageOneTurnonLogregRecallRangePct}{36.0\%--79.5\%}
\newcommand{\StageOneTurnonLogregEnrichRange}{7.95--14.39}

\expandafter\newcommand\csname SparsWinRateThreshold50Couplingweighted\endcsname{40\%}
\expandafter\newcommand\csname SparsWinRateThreshold50Standard\endcsname{40\%}

\expandafter\newcommand\csname SparsWinRateTopk50Couplingweighted\endcsname{0\%}

\expandafter\newcommand\csname SparsWinRateTopk50Standard\endcsname{50\%}

\newcommand{\HardwareAlignedCompileCellCount}{40}
\newcommand{\HardwareAlignedTwoQCountReductionRangePct}{42.3\%--82.6\%}
\newcommand{\HardwareAlignedTwoQDepthReductionRangePct}{54.5\%--84.8\%}
\newcommand{\CreditWideSparseDenseGainRetainedRangePct}{47.1\%--67.7\%}
\newcommand{\CompilePathDepthReductionRangePct}{35.1\%--54.8\%}
\newcommand{\CompilePathDepthReductionMedianPct}{47.4\%}

\newcommand{\CompilePathNativeTwoQCountReductionRangePct}{1.4\%--3.6\%}

\newcommand{\CompilePathDurationReductionRangePct}{12.5\%--29.2\%}
\newcommand{\CompilePathDurationReductionMedianPct}{20.0\%}

\newcommand{\CompilePathMatchedPointCount}{26}
\newcommand{\CompilePathTranspilerSeedCount}{10}

\newcommand{\HWCompBlockxyQubits}{12}

\newcommand{\HWCompBlockxyFeasComp}{100\%}
\newcommand{\HWCompBlockxyFeasCompRaw}{41\%}

\expandafter\newcommand\csname HWReadyAphroditeMedianDeltaDepthBipartiteDenseP2XXYYBlockXY\endcsname{242}
\expandafter\newcommand\csname HWReadyAphroditeMedianDeltaTwoQBipartiteDenseP2XXYYBlockXY\endcsname{118}
\expandafter\newcommand\csname HWReadyAphroditeMedianDeltaDepthBipartiteDenseP2XXYYRingXY\endcsname{192}
\expandafter\newcommand\csname HWReadyAphroditeMedianDeltaTwoQBipartiteDenseP2XXYYRingXY\endcsname{118}
\expandafter\newcommand\csname HWReadyIBMMiamiSnapshotMedianDeltaDepthBipartiteDenseP2CZNativeBlockXY\endcsname{279}
\expandafter\newcommand\csname HWReadyIBMMiamiSnapshotMedianDeltaTwoQBipartiteDenseP2CZNativeBlockXY\endcsname{130}
\expandafter\newcommand\csname HWReadyIBMMiamiSnapshotMedianDeltaDepthBipartiteDenseP2CZNativeRingXY\endcsname{270}
\expandafter\newcommand\csname HWReadyIBMMiamiSnapshotMedianDeltaTwoQBipartiteDenseP2CZNativeRingXY\endcsname{135}
\expandafter\newcommand\csname HWReadyRepDeltaDepthBipartiteDenseAphroditeXXYYN6D6P2BlockXY\endcsname{342}
\expandafter\newcommand\csname HWReadyRepDeltaTwoQBipartiteDenseAphroditeXXYYN6D6P2BlockXY\endcsname{180}
\expandafter\newcommand\csname HWReadyRepDeltaDepthBipartiteDenseAphroditeXXYYN6D6P2RingXY\endcsname{283}
\expandafter\newcommand\csname HWReadyRepDeltaTwoQBipartiteDenseAphroditeXXYYN6D6P2RingXY\endcsname{181}
\expandafter\newcommand\csname HWReadyRepDeltaDepthBipartiteDenseIBMMiamiSnapshotCZNativeN6D6P2BlockXY\endcsname{355}
\expandafter\newcommand\csname HWReadyRepDeltaTwoQBipartiteDenseIBMMiamiSnapshotCZNativeN6D6P2BlockXY\endcsname{193}
\expandafter\newcommand\csname HWReadyRepDeltaDepthBipartiteDenseIBMMiamiSnapshotCZNativeN6D6P2RingXY\endcsname{369}
\expandafter\newcommand\csname HWReadyRepDeltaTwoQBipartiteDenseIBMMiamiSnapshotCZNativeN6D6P2RingXY\endcsname{196}

\newcommand{\CalibPerturbNSeeds}{64}

\expandafter\newcommand\csname CalibPerturbJointCombinedNoise50EtaOne\endcsname{0.993}
\expandafter\newcommand\csname CalibPerturbSeqCombinedNoise50EtaOne\endcsname{0.710}
\expandafter\newcommand\csname CalibPerturbJointFeatureNoise50EtaOne\endcsname{0.987}
\expandafter\newcommand\csname CalibPerturbSeqFeatureNoise50EtaOne\endcsname{0.686}
\expandafter\newcommand\csname CalibPerturbJointExactNoise50EtaOnePct\endcsname{93.4\%}
\expandafter\newcommand\csname CalibPerturbSeqExactNoise50EtaOnePct\endcsname{5.4\%}
\newcommand{\StructuredCalibModesNSeeds}{64}

\newcommand{\RealCalibCreditNSeeds}{25}

\newcommand{\CreditWideRetainedCrossMassRangePct}{2.1\%--4.6\%}

\newcommand{\CreditTwentyQPFiveTiedBKRatePct}{0.14\%}
\newcommand{\CreditTwentyQPFiveFullyGroupedBKRatePct}{3.7\%}
\newcommand{\CreditTwentyQPFiveFullyGroupedBKRateFold}{25-fold}
\newcommand{\CreditTwentyQPFiveFullyGroupedFeasPct}{64.7\%}
\newcommand{\CreditWideFixedChanceRatioMin}{161}
\newcommand{\CreditWideFixedChanceRatioMax}{5{,}238}
\newcommand{\CreditWarmStartGapClosedRangePct}{78.1\%--93.3\%}
\newcommand{\CGXYFullyGroupedQPUHitInstances}{8/8}
\newcommand{\CGXYFullyGroupedQPUThresholdHits}{823}
\newcommand{\CGXYFullyGroupedQPUTotalShots}{327{,}680}
\newcommand{\CGXYFullyGroupedQPUMeanFeasPct}{67.4\%}

\newcommand{\CalibPerturbJointCombinedNoiseFiftyEtaOne}{\csname CalibPerturbJointCombinedNoise50EtaOne\endcsname}
\newcommand{\CalibPerturbSeqCombinedNoiseFiftyEtaOne}{\csname CalibPerturbSeqCombinedNoise50EtaOne\endcsname}
\newcommand{\CalibPerturbJointExactNoiseFiftyEtaOnePct}{\csname CalibPerturbJointExactNoise50EtaOnePct\endcsname}
\newcommand{\CalibPerturbSeqExactNoiseFiftyEtaOnePct}{\csname CalibPerturbSeqExactNoise50EtaOnePct\endcsname}

\DeclareMathOperator*{\argmax}{arg\,max}

\DeclareMathOperator{\HW}{HW}

\newcommand{\vecs}{\bm{s}}           
\newcommand{\vecf}{\bm{f}}           
\newcommand{\matX}{\bm{X}}           
\newcommand{\matZ}{\bm{Z}}           
\newcommand{\matQ}{\bm{Q}}           
\newcommand{\matW}{\bm{W}}           
\newcommand{\Ham}{H}                  
\newcommand{\Hmix}{H_{\mathrm{mix}}}  
\newcommand{\Hxy}{H_{\mathrm{XY}}}    
\newcommand{\qubo}{QUBO}
\newcommand{\qaoa}{QAOA}
\newcommand{\scaps}[1]{\textnormal{\textsc{#1}}}
\newcommand{\mainref}[1]{\ref{#1}}

\newtheorem{theorem}{Theorem}
\newtheorem{lemma}[theorem]{Lemma}
\newtheorem{proposition}[theorem]{Proposition}
\newtheorem{corollary}[theorem]{Corollary}

\title{Coupling-Grouped XY-\qaoa{} for\\[0.3em]
Joint Anomaly-Feature Selection}

\author{%
  Pauli Taipale\,\orcidlink{0009-0000-0280-3243}\\
  OP Lab, OP Pohjola\\
  \texttt{pauli.taipale@op.fi}
}

\date{}

\begin{document}

\maketitle

\begin{abstract}
Selecting anomalous samples and explanatory features under fixed cardinalities defines
a coupled optimization problem: feature-first selection can overlook features whose
usefulness depends on the selected samples. In the
analyzed model, joint calibration-error sensitivity is sample-count independent,
while the sum-column feature-first rule's sensitivity grows linearly. To optimize this formulation, we
introduce Coupling-Grouped XY Quantum Approximate Optimization Algorithm
(CG-XY-QAOA), with constraint-preserving mixers and block-structured phase angles. On
matched sparse IBM Heron R3 circuits, our implementation reduces circuit depth by
\CompilePathDepthReductionRangePct{} against Qiskit's optimization-level~3
preset pass manager on the CZ-basis target and enables runs at 64 decision qubits (\(p=2\)) and 36
(\(p=3\)). At 20 decision qubits (\(p=5\)), fully grouped angles raise the
fixed low-energy hit probability \CreditTwentyQPFiveFullyGroupedBKRateFold{} to
\CreditTwentyQPFiveFullyGroupedBKRatePct{}, with
\CreditTwentyQPFiveFullyGroupedFeasPct{} of measured samples satisfying both
constraints.
Benchmarks favor joint selection when feature usefulness varies across anomalies.
Noiseless simulations show problem-structured grouping improves over
same-depth XY-\qaoa{} and parameter-matched type-preserving randomizations.
\end{abstract}


\section{Introduction}
\label{sec:intro}

Many applications require exact-cardinality selection from two interacting sets.
An exact budget means choosing exactly \(k\) items on one side and \(m\) items
on the other. Examples include anomaly-feature selection, drug-target screening,
protein-module selection, sensor-event coverage, facility-customer assignment,
and enzyme-substrate pathway selection. We call this structure
\emph{constrained bipartite selection}. Its quadratic objective can be written in
Quadratic Unconstrained Binary Optimization (QUBO) or Ising
form~\cite{hammer1969pseudoboolean,lucas2014ising}.

The primary motivating application is joint anomaly-feature selection. In
high-dimensional data one must identify \emph{both} which samples are anomalous
\emph{and} which features are informative. Feature-first approaches separate
these decisions by selecting features before detecting anomalies.
This can fail when feature utility is case-dependent. A feature may be
uninformative on average yet decisive for a particular anomaly subtype.
Ranking features before selecting samples can therefore discard informative
sample--feature interactions.
The joint formulation is most relevant when multiple anomalies and features
must be selected and feature utility varies across samples. In the tested
boundary cases, its advantage vanishes at zero feature noise and can reverse
when \(k=1\).

We propose a \emph{joint} formulation that optimizes the sample and feature sets
simultaneously. We use $[N]=\{1,\ldots,N\}$ and $[D]=\{1,\ldots,D\}$. Let
$\matW \in \mathbb{R}^{N \times D}$ encode bipartite interaction weights between
samples and features, and let $\{a_i\}_{i\in[N]}$ and $\{b_j\}_{j\in[D]}$ denote
marginal scores. We optimize
\begin{equation}
  \argmax_{\substack{S \subseteq [N],\, \lvert S \rvert = k \\[2pt]
                     F \subseteq [D],\, \lvert F \rvert = m}}
  \rho \left( \sum_{i \in S} a_i + \sum_{j \in F} b_j \right)
  + (1-\rho) \sum_{i \in S} \sum_{j \in F} W_{ij}
  \label{eq:objective}
\end{equation}

The weight \(\rho\) balances marginal scores against the pairwise interaction
term. The bilinear term is the only term whose value depends jointly on the
selected samples and selected features.
The selected-set indicators give a quadratic binary cost with
$N + D$ variables, whose cross-coupling terms encode the interactions. Following
standard \qubo{} terminology, the quadratic form itself is unconstrained. The
hard budgets enter through the indicator-vector domain
$\{(\vecs,\vecf)\in\{0,1\}^{N+D}:\lvert\vecs\rvert=k,\lvert\vecf\rvert=m\}$,
where \(\vecs\) and \(\vecf\) encode \(S\) and \(F\). For anomaly-feature detection,
we start from a standardized residual matrix $\matZ$ and derive cross-term
weights $\matW$ through a feature-wise calibration map
(Section~\ref{sec:formulation}). We choose $a_i$ and $b_j$ to encode sample and
feature marginal signal derived from the same data.
The calibration map also sets the quantum cost coefficients: residuals define
bounded weights, and the resulting QUBO coefficients scale the phase-separator
rotations.

\paragraph{Constraint-preserving QAOA and sparse-connectivity overhead.}
The quantum algorithm implements the same exact-budget restriction inside the
ansatz. To enforce strict $(k,m)$ feasibility without penalty tuning, we use
feasible initializations together with a Block~XY mixer that preserves the
Hamming weight of each register. Tensor Dicke states give symmetric noiseless
feasible support, and exact-feasible basis states support low-depth hardware
execution. The main hardware resource cost on sparse connectivity is then the
dense \(N\times D\) cost layer.
On sparse coupling maps, this layer induces routing and can permute the mapping
between logical variables and measured bits. We therefore study
hardware-aligned sparsification to reduce routing and layout-aware decoding to
preserve the logical interpretation of measured outputs after transpilation.

\paragraph{Application scope.}
We evaluate the formulation in controlled synthetic experiments and on time-split
anomaly benchmarks built from the Credit Card dataset and the IBM IT-AML
HI-Small anti-money-laundering (AML) dataset.
The benchmark protocol has two stages. An earlier-time ranking procedure forms a
candidate pool, and the stage-2 optimizer then performs exact-budget \((k,m)\)
selection over that pool. Protocol-specific scorer, calibration, and trigger
rules are fixed before test evaluation, so test outcomes cannot influence them
(Section~\ref{sec:experiments}).
\paragraph{Contributions.}
\begin{enumerate}
  \item \textbf{Calibration-robust constrained bipartite selection.}
  We formulate exact-budget bipartite selection as a joint \qubo{}/Ising
  objective and map calibrated residual signal into bounded bipartite
  weights. When the calibrated sample--feature coefficient matrix satisfies
  \(\|\widehat{\matW}-\matW^\star\|_\infty \le \eta\), calibration can reduce
  the joint recovery margin by at most \(2\eta\), whereas the analyzed
  sum-column feature-first margin can lose up to \(2N\eta\) by aggregating the
  same entrywise error over all \(N\) candidate samples. The same proof mechanism extends from
  light-tailed sub-Gaussian residuals to sub-exponential Bernstein-tail residuals.
  \item \textbf{Coupling-Grouped XY-QAOA (CG-XY-QAOA).}
  We introduce a problem-structured grouped-angle variant of alternating-operator
  XY-\qaoa{} for the bipartite objective. The construction uses the standard
  feasible-subspace XY mixer~\cite{wang2020xy} and adapts the per-term phase
  parameterization of multi-angle \qaoa{}~\cite{herrman2022multiangle} to the
  sample--feature cost decomposition. We evaluate
  fixed-transport variants that keep the mixer schedule tied and a fully
  grouped cost-plus-mixer variant that tests layerwise transport freedom.
  Parameter-matched controls compare this decomposition with type-preserving
  random regroupings and with tied XY-\qaoa{} at depth \(2p\).
  \item \textbf{Hardware-efficient sparse-surrogate execution and decoding.}
  We make the dense two-register constrained ansatz executable on sparse IBM
  Heron hardware by combining register-preserving variable placement,
  hardware-aligned sparse-surrogate phase separators, Block~XY fusion,
  fractional gates, edge-colored scheduling, and layout-aware decoding. On
  matched sparse IBM Heron R3 circuits, this implementation reduces submitted
  circuit depth by \CompilePathDepthReductionRangePct{} and the target-based
  circuit-duration estimate by \CompilePathDurationReductionRangePct{} relative
  to Qiskit's optimization-level~3 preset pass manager on the standard Heron
  CZ-basis target. The fractional target accounts for a
  \CompilePathNativeTwoQCountReductionRangePct{} reduction in native
  two-qubit instruction count. Block~XY fusion, late parameter binding, and
  edge-colored transport ordering provide additional depth and duration
  reductions without changing that count. At 20 decision qubits, the fully
  grouped \(p=5\) condition applies five alternating phase-separator and mixer
  layers while preserving separate sample and feature cardinalities. It raises
  the BK hit rate \CreditTwentyQPFiveFullyGroupedBKRateFold{} over tied
  XY-\qaoa{} and retains \CreditTwentyQPFiveFullyGroupedFeasPct{} decoded
  exact-budget mass. Sparse-surrogate executions reach 64 decision
  qubits at \(p=2\) and 36 decision qubits at \(p=3\). In these \(p=2\)
and \(p=3\) executions, decoded exact-budget mass remains
\CreditLargeEndpointSixtyFourQPTwoFeasPct{} at 64 qubits and
\CreditLargeEndpointThirtySixQPThreeFeasPct{} at 36 qubits, giving
feasibility-to-chance ratios of
\CreditLargeEndpointSixtyFourQPTwoChanceRatio{} and
\CreditLargeEndpointThirtySixQPThreeChanceRatio{}. Retained
  sample--feature coupling mass quantifies the sparse surrogate, and
  dense-objective rescoring evaluates measured samples against the original
  objective.
\end{enumerate}

\section{Related Work}
\label{sec:related}

\paragraph{QAOA and Constrained Optimization.}
The Quantum Approximate Optimization Algorithm~\cite{farhi2014quantum} alternates
evolution under cost and mixer Hamiltonians. Standard \qaoa{} uses
the transverse-field mixer~$\sum_i X_i$, whose bit flips change cardinality.
The alternating-operator view of constrained \qaoa{} makes feasible-subspace
design explicit~\cite{hadfield2019quantum}. With feasible initialization, XY
mixers preserve Hamming weight and keep the evolution within a fixed-cardinality
sector~\cite{wang2020xy}.
Multi-angle \qaoa{} assigns separate variational parameters to Hamiltonian
terms or term groups and can trade additional classical parameters for reduced
quantum depth~\cite{herrman2022multiangle}. We build on that idea with a
coupling-grouped phase parameterization for the bipartite cost.
Recent work also relates XY-mixer topology and generator choice to trainability
and warm-start performance for shared- and multi-angle \qaoa{}~\cite{kordonowy2026lie}.

\paragraph{Dicke State Preparation.}
Dicke states $\ket{D^n_k}$ are symmetric superpositions of all $n$-qubit states
with Hamming weight~$k$~\cite{dicke1954coherence,bartschi2019deterministic}.
Efficient preparation circuits construct this superposition without enumerating
its $\binom{n}{k}$ basis states and thereby provide a symmetric feasible
initialization for fixed-cardinality XY-\qaoa{}.

\paragraph{Feasible-subspace ansatz choices.}
Constraint-preserving QAOA can vary both the feasible initialization and the
mixer construction. Grover-style mixers shift more burden
to state preparation~\cite{bartschi2019grover}, while warm-start QAOA biases the
search using solutions of classical relaxations~\cite{egger2021warmstart}. We
use Block~XY mixing because it preserves Hamming weight independently in each
register and is generated by Pauli \(XX+YY\) exchange interactions.

\paragraph{Hardware optimization benchmarks.}
Hardware optimization studies show that encoding, native connectivity,
implementation path, and classical reference all affect the interpretation of
quantum-optimization results. On superconducting processors, QAOA performance
depends on whether the problem graph is native to the device or routed through
sparse connectivity~\cite{harrigan2021nonplanar}. Rydberg-array studies take a
different hardware-native route by encoding optimization problems as
maximum-independent-set instances and benchmarking them against classical
heuristics~\cite{ebadi2022rydbergmis,nguyen2023rydbergarbitrary}. Gate-model
constrained-optimization experiments similarly emphasize feasible
initialization matched to the mixer~\cite{he2023alignment}, direct execution of
XY-constrained ansatzes~\cite{niroula2022summarization}, and explicit IBM
hardware resource accounting~\cite{kotil2025multiobjective}.

Earlier fixed-cardinality demonstrations include two trapped-ion studies.
Niroula~et~al. executed \(p=1\) Hamming-weight-preserving XY-\qaoa{}
circuits for extractive summarization up to 20 qubits with two-qubit gate depth
up to 159. He~et~al. applied XY-mixer \qaoa{} to a 32-qubit portfolio problem
with one cardinality constraint on a trapped-ion quantum
processor~\cite{niroula2022summarization,he2023alignment}. Relative to the
closest 20-qubit fixed-cardinality XY-\qaoa{} comparison, our \(p=5\) condition
extends logical depth from one to five alternating layers while adding a second
preserved cardinality constraint.
Register width alone does not capture the comparison because the circuits also
preserve two global cardinality constraints at nontrivial \qaoa{} depth. Broader
gate-model demonstrations use larger unconstrained registers or width--depth
products. Pelofske~et~al. execute \(p=1,2\) \qaoa{} on 127-qubit
superconducting-transmon IBM Quantum hardware with a heavy-hex coupling graph.
Shaydulin and Pistoia report MaxCut \qaoa{} with \(N p\) up to 320 on Quantinuum
trapped-ion processors~\cite{pelofske2023shortdepthqaoa127,shaydulin2023qaoanp}.
Our comparison is restricted to constraint-preserving bipartite-selection
\qaoa{} with reported feasible-sector retention. Within this scope, the
64-qubit \(p=2\) and 36-qubit \(p=3\) runs are, to our knowledge, the largest
reported width--depth configurations.

\paragraph{Sparse and surrogate phase operators.}
A related strategy is to replace a dense phase separator by a cheaper sparse
surrogate. Liu, Shaydulin, and Safro study QAOA with sparsified phase operators
and show that performance depends on preserving the original problem's ground
state or low-energy structure, not merely on coefficient
retention~\cite{liu2022sparsifiedphase}. For sparse surrogates, retained
coupling mass measures how much of the dense bilinear objective remains in the
executed phase separator, and dense-objective rescoring evaluates measured
bitstrings on the original objective.

\paragraph{Industrial applications.}
Prior gate-model studies have applied constrained optimization to portfolio
rebalancing, planning, and airline tail
assignment~\cite{hodson2019portfolio,stollenwerk2020planning,vikstal2020tail}.
These studies include single-register constraints as well as multi-register and
one-hot encodings. Here we study a calibrated constrained
bipartite selection problem over interacting sample and feature sets, together
with a decoding scheme that preserves interpretation after sparse-connectivity
transpilation. Other domains would require their own calibration maps and
validation protocols for the corresponding bipartite weights.

\paragraph{Application setting: anomaly detection.}
Isolation Forest detects anomalies through isolation~\cite{liu2008isolation},
while Local Outlier Factor uses local density~\cite{breunig2000lof}. Broader
taxonomies also cover reconstruction-based and other classical
methods~\cite{chandola2009anomaly}. When feature selection precedes anomaly
scoring, it can discard joint sample--feature information. Post-hoc
anomaly-explanation methods
interpret a detector after it has been
fixed~\cite{sejr2021explainable,li2023surveyexplainableanomalydetection,ijcai2018p341}.
Joint instance-feature co-selection has also been studied for unsupervised
rare-category analysis~\cite{he2010coselection}. Our setting differs by
imposing exact review budgets, deriving calibrated bounded sample--feature
weights, comparing calibration sensitivity against feature-first aggregation,
and mapping the resulting QUBO/Ising objective to a constraint-preserving
\qaoa{} implementation.

\section{Calibrated Constrained Bipartite Selection}
\label{sec:formulation}

The calibrated weight matrix defines both the classical objective and the
CG-XY-\qaoa{} phase separator. The key
robustness result is that joint recovery is less sensitive to calibration error:
entrywise error of radius $\eta$ can reduce the joint recovery margin by at most
$2\eta$, while the analyzed feature-first sufficient condition can lose up to
$2N\eta$ through column aggregation. This turns calibration robustness into a
perturbation statement for a single \qubo{} coefficient matrix.

\subsection{Calibration Map and Objective}

Let $\matX \in \mathbb{R}^{N \times D}$ be a data matrix with $N$ samples and $D$
features. Let $\mathcal{R}$ denote a reference set of sample indices used for
calibration (e.g., an early-time window in a time-split evaluation). Define the
standardized residual matrix \(\matZ\) element-wise by
\begin{equation}
  Z_{ij} = \frac{X_{ij} - \hat{\mu}_j}{\hat{\sigma}_j}\,,
  \label{eq:zscore}
\end{equation}
where $\hat{\mu}_j$ and $\hat{\sigma}_j$ are feature-local location and scale
estimates fit on $\mathcal{R}$.
For heavy-tailed residuals we use robust feature-wise estimates, taking the median
\(\hat{\mu}_j=\mathrm{median}_{i\in\mathcal R} X_{ij}\) and the scaled median
absolute deviation
\[
  \hat{\sigma}_j
  = 1.4826\,\mathrm{median}_{i\in\mathcal R}|X_{ij}-\hat{\mu}_j|.
\]
The factor \(1.4826\) makes the median absolute deviation consistent with
\(\sigma\) under a normal model.

\paragraph{Role of the reference set $\mathcal{R}$.}
In a deployed monitoring setting, calibration is fit on historical data and then
applied prospectively. Our time-split evaluation mirrors this fit-then-apply structure. All feature-local
estimates $(\hat{\mu}_j,\hat{\sigma}_j)$ and calibration maps are fit only on
$\mathcal{R}$ and then applied to later samples without look-ahead leakage. The
resulting weights measure how unusual an observation is relative to recent
reference-period behavior, not relative to test-period information. When labels exist,
restricting \(\mathcal{R}\) to known-normal samples avoids contaminating the
reference distribution with labeled anomalies. When labels are unavailable,
\(\mathcal{R}\) is selected without labels.

\paragraph{Feature-wise calibration to bounded weights.}
We convert standardized residuals to nonnegative cross-term weights
$\matW \in \mathbb{R}_{\ge 0}^{N \times D}$ via a per-feature calibration map.
The calibration map has two roles. It makes tail scores comparable across
heterogeneous features and bounds the QUBO coefficients that scale phase-separator rotations.
It must therefore be monotone in feature-wise extremeness, return finite weights
for finite reference windows, and cap them at a fixed numerical scale.

The default label-free map uses a smoothed empirical cumulative distribution
function (CDF) on the reference window, followed by a two-sided tail
transformation. Let \(n=|\mathcal R|\), and
let
\[
  R_j(z)=\sum_{i\in\mathcal R}\mathbf 1\{Z_{ij}\le z\}
\]
be the number of reference residuals in feature \(j\) no larger than \(z\). The
clipped smoothed empirical CDF estimate \(\widehat F_j(z)\) and the two-sided
tail probability \(p_{ij}\) are
\begin{align}
  \widehat F_j(z)
  &= \mathrm{clip}\!\left(
       \frac{R_j(z)+1/2}{n+1},\,
       \frac{1}{n+1},\,
       \frac{n}{n+1}
     \right), \nonumber\\
  p_{ij}
  &= 2\min\!\left(\widehat F_j(Z_{ij}),\,1-\widehat F_j(Z_{ij})\right).
  \label{eq:tail-calibration}
\end{align}
The corresponding bounded outlier weight \(W_{ij}\) is
\begin{equation}
  W_{ij} = \mathrm{clip}(-\log p_{ij},\,0,\,w_{\max})\,,
  \label{eq:weights}
\end{equation}
where \(w_{\max}\) is the clipping cap. The rank smoothing and CDF clipping
avoid zero tail probabilities. The final weight cap keeps the bilinear
coefficients and gate angles bounded.

The empirical-CDF probability integral transform (ECDF/PIT) validity statement
used in the theory assumes that the score map is fixed independently of the
ECDF calibration sample. A split-reference fit/calibration window satisfies this
condition. Supplementary Section~\ref{supp:calibration_validity} gives the
finite-sample ECDF/PIT result, the coefficient/gate-angle clipping bound, and
the split-conformal super-uniform
alternative~\cite{rosenblatt1952remarks,massart1990tight,vovk2022algorithmic,lei2018distributionfree}.

We also include an absolute-\(z\) scorer, \(u_{ij}=|Z_{ij}|\), as a baseline
before the same clipping and coefficient construction. Protocol-specific
calibration maps are fixed before test evaluation.

\paragraph{Exact-budget objective and coefficient scale.}
We introduce binary decision variables where $s_i \in \{0,1\}$ indicates whether
sample~$i$ is selected, and $f_j \in \{0,1\}$ indicates whether feature~$j$ is
selected. The cost function to minimize, \(\Ham(\vecs,\vecf)\), is
\begin{equation}
  \Ham(\vecs, \vecf) =
    -\sum_{i=1}^{N} a_i \, s_i
    -\sum_{j=1}^{D} b_j \, f_j
    -\lambda \sum_{i=1}^{N} \sum_{j=1}^{D} W_{ij} \, s_i \, f_j\,,
  \label{eq:hamiltonian}
\end{equation}
subject to the cardinality constraints $\sum_i s_i = k$ and $\sum_j f_j = m$.
The exact-budget sector \(\Omega\) is
\[
  \Omega=\{(\vecs,\vecf):|\vecs|=k,\ |\vecf|=m\},
  \qquad |\Omega|=\binom Nk\binom Dm .
\]
The calibrated objective is optimized over \(\Omega\) by both the strict-feasible
classical selectors and the constraint-preserving quantum circuits.
Equation~\eqref{eq:hamiltonian} fixes the minimization convention. Its
coefficients absorb the tradeoff parameter $\rho$ from
Eq.~\eqref{eq:objective}.
The coefficients separate the sample marginal, feature marginal, and
cross-coupling terms as follows.
\begin{itemize}
  \item \textbf{Sample marginal score.} $a_i = \max_{j \in [D]} W_{ij}$ measures the
    maximum deviation of sample~$i$ across all features under the chosen weights.
  \item \textbf{Feature marginal score.} $b_j = \mathrm{Var}_{i\in[N]}(|Z_{ij}|)$
    captures cross-sample variation in feature~$j$ before the bounded tail map
    clips large residuals, so feature ranking is not dominated by the clipping cap.
  \item \textbf{Cross-coupling.} $\lambda > 0$ balances marginal terms against
    the bilinear coupling. In all protocols, \(\lambda\) is fixed
    before optimization or test evaluation together with the coefficient map.
\end{itemize}
The bilinear term $W_{ij}\cdot s_i \cdot f_j$ rewards selecting sample--feature
pairs with high joint outlier weights. Unless a protocol specifies a separate
generator, \(W\) denotes the bounded weight map in Eq.~\eqref{eq:weights}.

Hardware-facing QUBOs use the clipped coefficient map without additional
per-instance normalization before circuit construction. Sample marginals,
feature marginals, and bilinear terms can therefore have different numerical
ranges. Fixed-angle and angle-selected executions use this coefficient scale.
When theorem statements require uniform marginal
bounds, we impose finite-instance bounds on \(a_i\) and \(b_j\), or use
normalization or clipping at the protocol level.

\subsection{Joint Signal and Feature-First Separation}

\paragraph{Separable-coupling check.}
The bilinear contribution is irreducibly joint only when the coupling encodes
\emph{non-separable} interactions between samples and features. For the core
bilinear objective, an exactly separable (rank-one) coupling reduces to
independent top-$k$ and top-$m$ selection.
\begin{proposition}[Rank-one coupling reduces to independent selection]
\label{prop:rank-one}
Let $\matW = \bm{u} \bm{v}^\top$ with $\bm{u}\in\mathbb{R}_{\ge 0}^N$ and
$\bm{v}\in\mathbb{R}_{\ge 0}^D$, and consider the core bilinear objective
$B(S,F)=\sum_{i\in S}\sum_{j\in F} W_{ij}$ under $\lvert S\rvert=k$ and
$\lvert F\rvert=m$.
Then a maximizer is obtained by choosing $S$ as the indices of the $k$ largest
entries of $\bm{u}$ and $F$ as the indices of the $m$ largest entries of $\bm{v}$.
Moreover, if $\matW=\bm{u}\bm{v}^\top+\bm{E}$, then for any $S,F$ we have
$\lvert B_{\matW}(S,F)-B_{\bm{u}\bm{v}^\top}(S,F)\rvert \le km\|\bm{E}\|_\infty$.
Consequently, the perturbed optimum can improve over the rank-one independent
solution by at most $2km\|\bm{E}\|_\infty$ under the perturbed objective.
\end{proposition}
Proof in Supplementary Section~\ref{supp:elementary_formulation_lemmas}.
Proposition~\ref{prop:rank-one} shows that when sample and feature effects
factor, the bilinear objective contains no irreducibly joint signal and reduces
to independent marginal ranking.

\paragraph{Sufficient-condition separation under feature noise.}
The theorem and corollary identify parameter ranges in a planted model where
exact joint recovery has a weaker sufficient condition than the analyzed
sum-column feature-first rule. The centered score-field model isolates
sample--feature signal from
sub-Gaussian fluctuations before the nonnegative clipping used by empirical
calibration. It adapts elevated-mean submatrix localization and
biclustering~\cite{ariascastro2011anomalouscluster,kolar2011minimax,
butucea2015sharp,cai2017submatrix,hajek2018submatrix} to exact budgets, then adds
the feature-first comparison and deterministic calibration error. The
experiments test whether this separation persists after finite-sample
calibration and whether feasible-subspace circuits can sample the resulting
objective.

We assume \(\mathcal M>1\), \(m(D-m)>0\), and \(0<\delta<1\), so the
logarithms in the theorem and corollary are defined.

\begin{theorem}[Joint recovery under a planted bicluster with sub-Gaussian noise]
\label{thm:joint-noise-superiority}
Assume there exist planted sets $S^\star\subseteq[N]$, $F^\star\subseteq[D]$ with
$|S^\star|=k$, $|F^\star|=m$, and observed weights
\[
W_{ij}=\mu\,\mathbf{1}\{i\in S^\star,\;j\in F^\star\}+\varepsilon_{ij},
\]
where $\{\varepsilon_{ij}\}$ are independent, mean-zero, $\sigma$-sub-Gaussian.
Define
\[
B(S,F)=\sum_{i\in S}\sum_{j\in F}W_{ij},\qquad
\mathcal{M}=\binom{N}{k}\binom{D}{m}.
\]
If
\[
\mu \;\ge\; 2\sigma\sqrt{\frac{\log\!\left((\mathcal{M}-1)/\delta\right)}{\min(k,m)}},
\]
then with probability at least $1-\delta$, $(S^\star,F^\star)$ is the unique
maximizer of $B(S,F)$ over $|S|=k$, $|F|=m$.
\end{theorem}
Proof in Supplementary Section~\ref{supp:noise-proof}.
The theorem is stated for the bilinear core \(B(S,F)\). For the full
maximization score
\[
  Q(S,F)=\sum_{i\in S}a_i+\sum_{j\in F}b_j+\lambda B(S,F),
\]
the same proof gives a conservative bounded-marginal condition. If
\(|a_i|\le A_{\max}\), \(|b_j|\le B_{\max}\), and
\[
  \lambda\mu-\frac{2A_{\max}}{m}-\frac{2B_{\max}}{k}
  \;\ge\;
  2\lambda\sigma
  \sqrt{\frac{\log((\mathcal M-1)/\delta)}{\min(k,m)}},
\]
then the planted pair is also the unique maximizer of \(Q\) with probability at
least \(1-\delta\)
(Supplementary Corollary~\ref{cor:bounded-marginal-full-score}). Thus the
full-objective guarantee requires the bilinear signal to dominate both noise
and any adversarial marginal advantage. Empirical marginal scores fall outside
the centered bilinear theorem unless this margin condition is satisfied.

\begin{corollary}[Feature-first sufficient condition under column aggregation]
\label{cor:sequential-snr}
Under the same model, let $C_j=\sum_{i=1}^N W_{ij}$ and define a sequential
feature-first step that chooses the top-$m$ columns by $C_j$.
A sufficient condition for exact recovery of $F^\star$ with probability at least
$1-\delta$ is
\[
\mu k \;\ge\; \sigma\sqrt{8N\log\!\left(\frac{m(D-m)}{\delta}\right)}.
\]
Thus the analyzed feature-first condition accumulates \(N^{1/2}\) sample-wise noise,
whereas Theorem~\ref{thm:joint-noise-superiority} scales with
\(\min(k,m)^{-1/2}\) in the required $\mu$.
\end{corollary}
Proof in Supplementary Section~\ref{supp:noise-proof}.
This corollary analyzes the sum-column feature-first rule. The empirical
feature-first baselines separately benchmark other calibrated aggregators.

Together, Theorem~\ref{thm:joint-noise-superiority} and
Corollary~\ref{cor:sequential-snr} identify a parameter range where the joint
condition is weaker than the analyzed feature-first condition. Suppressing
logarithmic factors, the joint threshold scales as
\(\sigma \min(k,m)^{-1/2}\). The analyzed feature-first threshold scales as
\(\sigma N^{1/2}/k\). For fixed \(k,m\) and growing \(N\), some signal/noise
windows satisfy the joint recovery condition without satisfying the analyzed
feature-first condition. The planted model supplies a mechanism-level recovery
certificate rather than a fitted model of the public datasets. The synthetic
experiments test its
recovery and perturbation implications under known support. The time-split
benchmarks test whether joint gains and feature-aligned calibration drift
persist without assuming block-constant signal or independent residuals.

The independent-noise assumption can be relaxed using a proxy covariance
matrix \(\Gamma\). For each planted-versus-competitor contrast vector \(c\),
Supplementary Proposition~\ref{prop:covariance-aware-joint-recovery} replaces
the independent-noise scale by \(\sqrt{c^\top\Gamma c}\). Positive
within-feature correlation can therefore weaken the averaging gain in
feature-swap contrasts. The deterministic \(2\eta\) calibration perturbation
from Proposition~\ref{prop:calibration-perturbation} does not depend on the
noise covariance.

The remaining recovery variants keep the same joint-versus-feature-first
comparison and change only the residual-tail assumption or the clipping-bias
term.
Supplementary Proposition~\ref{prop:bernstein-tail-joint-recovery} gives a
Bernstein-tail extension for residuals that can have heavier, sub-exponential
tails than the sub-Gaussian model. If the centered entrywise residuals satisfy a
Bernstein moment-generating-function (MGF) condition with variance proxy $v^2$
and scale $b$, then joint
recovery is certified when
\[
  \mu \ge
  \max\!\left\{
    2v\sqrt{\frac{\log((\mathcal M-1)/\delta)}{\min(k,m)}},
    \frac{2b\log((\mathcal M-1)/\delta)}{\min(k,m)}
  \right\}.
\]
The sequential feature-first sufficient condition becomes
\[
  \mu k \ge
  \max\!\left\{
    2v\sqrt{N\log\!\left(\frac{m(D-m)}{\delta}\right)},
    2b\log\!\left(\frac{m(D-m)}{\delta}\right)
  \right\}.
\]
Thus the Bernstein-tail correction adds a linear logarithmic term while
preserving the structural difference between joint sample--feature recovery
and feature-first aggregation.
Supplementary Corollary~\ref{cor:clipped-residual-recovery} gives the
bounded-noise version for clipped residual calibration. Clipping at level
\(\tau\) yields the joint condition
\[
  \mu-2(\rho_\tau+\eta)
  \ge
  2\tau\sqrt{\frac{\log((\mathcal M-1)/\delta)}{\min(k,m)}},
\]
where \(\rho_\tau\) is the clipping-bias budget. The corresponding sequential
condition replaces the left-hand side by \(\mu k-2N(\rho_\tau+\eta)\) and the
right-hand side by
\(2\tau\sqrt{N\log(m(D-m)/\delta)}\), so clipping makes the bias--variance
tradeoff explicit without changing the calibration-scaling separation.

\subsection{Calibration Perturbations}

The calibration-aware extension retains the centered planted bicluster model
and separates clean signal from coefficient error. Let \(\matW^\star\) denote
the ideal calibrated weight field and \(\widehat{\matW}\) the weight field
actually produced by the reference window and clipping map. The quantity
\[
  \|\widehat{\matW}-\matW^\star\|_\infty\le\eta
\]
means that every stage-2 bilinear coefficient is perturbed by at most \(\eta\)
after calibration and clipping.

\begin{proposition}[Calibration-perturbation stability of joint recovery]
\label{prop:calibration-perturbation}
Assume the planted model of
Theorem~\ref{thm:joint-noise-superiority} for an ideal weight field $\matW^\star$,
and let $\widehat{\matW}$ be a calibrated estimate with
\[
  \|\widehat{\matW} - \matW^\star\|_\infty \le \eta.
\]
Define $\widehat B(S,F)=\sum_{i\in S}\sum_{j\in F}\widehat{\matW}_{ij}$.
If
\[
  \mu - 2\eta \;\ge\; 2\sigma\sqrt{\frac{\log\!\left((\mathcal{M}-1)/\delta\right)}{\min(k,m)}},
\]
then with probability at least $1-\delta$, $(S^\star,F^\star)$ is the unique
maximizer of $\widehat B(S,F)$ over $|S|=k$, $|F|=m$.
\end{proposition}
Proof in Supplementary Section~\ref{supp:noise-proof}.
Entrywise coefficient error can reduce the joint recovery margin by at most
\(2\eta\). The analyzed sum-column feature-first rule instead
aggregates the same coefficient error over all \(N\) samples, giving up to a
\(2N\eta\) margin loss, or \(2N\eta/k\) when written as a threshold on \(\mu\)
(Corollary~\ref{cor:sequential-calibration-sensitivity}). Supplementary
Corollary~\ref{cor:calibration-transfer} gives the high-probability
calibration-transfer result.

The empirical calibration analyses vary reference-window size and age, decompose
the resulting coefficient error, and test whether alternative calibration maps
reduce its feature-aligned component
(Supplementary Figures~\ref{fig:synthetic_calibration_robustness},
\ref{fig:real_creditcard_calibration_data_quality}, and
\ref{fig:real_creditcard_calibration_mechanism}).

\section{Coupling-Grouped XY-QAOA Method}
\label{sec:mixer}

CG-XY-\qaoa{} implements the calibrated bipartite objective while preserving
both cardinality budgets. Its grouped parameterizations distinguish sample
marginals, feature marginals, and sample--feature coupling, while Block~XY
mixing confines evolution to the exact-budget sector.

\subsection{Cost Hamiltonian and Phase Separators}
\label{sec:cost-hamiltonian}

The QAOA cost layer is generated by the Ising form of the \qubo{} objective.
The cost Hamiltonian is diagonal in the computational basis. The
feasible-subspace ansatz determines which exact-budget states are explored and
measured. We use the computational-basis convention
$x=(I-Z)/2$, so $x=1$ corresponds to Pauli eigenvalue $Z=-1$.
For the binary minimization objective
\[
  E(\vecs,\vecf)=
  -\sum_i a_i s_i-\sum_j b_j f_j-\lambda\sum_{i,j}W_{ij}s_if_j,
\]
the Ising Hamiltonian \(\Ham\) is, up to a constant,
\begin{equation}
  \Ham = \sum_{i=1}^{N} h_i^{(s)} Z_i
       + \sum_{j=1}^{D} h_j^{(f)} Z_{N+j}
       + \sum_{i=1}^{N} \sum_{j=1}^{D} J_{ij}\, Z_i Z_{N+j},
  \label{eq:ising}
\end{equation}
with
\begin{equation}
  h_i^{(s)}=\frac{a_i}{2}+\frac{\lambda}{4}\sum_j W_{ij},\qquad
  h_j^{(f)}=\frac{b_j}{2}+\frac{\lambda}{4}\sum_i W_{ij},\qquad
  J_{ij}=-\frac{\lambda}{4}W_{ij}.
  \label{eq:ising_coefficients}
\end{equation}
The QAOA cost layer uses \(\Ham\) as \(H_{\mathrm{C}}\). Positive \(a_i\),
\(b_j\), and \(W_{ij}\) lower the binary objective when the corresponding
sample, feature, or sample--feature pair is selected. The linear coefficients
$h_i^{(s)}$, $h_j^{(f)}$ include the marginal terms and the linear shifts
induced by the bilinear binary-to-Pauli conversion. Equivalently, for a
symmetric QUBO matrix \(Q\) satisfying \(E(x)=x^\top Qx\), the nonconstant
Ising coefficients are \(h_i=-(Q_{ii}+\sum_{j\ne i}Q_{ij})/2\) and
\(J_{ij}=Q_{ij}/2\). The corresponding Qiskit rotations are
\(R_Z(2\gamma h_i)\) and \(R_{ZZ}(2\gamma J_{ij})\).

The same decomposition defines a hierarchy of phase-separator
parameterizations used in Section~\ref{sec:quality}. We separate these terms
because marginal evidence and sample--feature coupling can have different phase
sensitivity at a fixed mixer depth. Write
\begin{equation}
\begin{aligned}
  \Ham &= H_{\mathcal S}+H_{\mathcal F}+H_{\mathcal S\mathcal F},\\
  H_{\mathcal S} &= \frac{1}{2}\sum_i a_iZ_i,\qquad
  H_{\mathcal F} = \frac{1}{2}\sum_j b_jZ_{N+j},\\
  H_{\mathcal S\mathcal F}
  &=\frac{\lambda}{4}\sum_{i,j}W_{ij}
  \left(Z_i+Z_{N+j}-Z_iZ_{N+j}\right).
\end{aligned}
  \label{eq:grouped_cost_decomposition}
\end{equation}

This semantic decomposition maps the three original binary objective
components into Ising form separately. In particular,
\(H_{\mathcal S\mathcal F}\) includes the one-local shifts induced by the
sample--feature products as well as their \(ZZ\) interactions. The three
components sum to Eq.~\eqref{eq:ising} up to the omitted constant.

The standard tied phase separator is
$U_{\mathrm{C}}(\gamma)=\exp[-i\gamma \Ham]$. The lower-dimensional
fixed-transport coupling-grouped variant, bilinear-CG-XY-QAOA, uses the same
set of diagonal gates. It keeps the marginal components tied while assigning a
separate phase to the sample--feature component.
\begin{equation}
  U_{\mathrm{C}}^{\mathrm{bilCG}}(\gamma_{\mathrm{marg}},
  \gamma_{\mathcal S\mathcal F})
  =
  \exp[-i(\gamma_{\mathrm{marg}}(H_{\mathcal S}+H_{\mathcal F})
       +\gamma_{\mathcal S\mathcal F}H_{\mathcal S\mathcal F})].
  \label{eq:bilinear_grouped_phase_separator}
\end{equation}
Setting
$\gamma_{\mathrm{marg}}=\gamma_{\mathcal S\mathcal F}=\gamma$
recovers the standard tied separator exactly. We also evaluate a
fixed-transport grouped separator with separate sample-field, feature-field, and
sample--feature phases.
\begin{equation}
  U_{\mathrm{C}}^{\mathrm{3grp}}(\gamma_{\mathcal S},\gamma_{\mathcal F},
  \gamma_{\mathcal S\mathcal F})
  =
  \exp[-i(\gamma_{\mathcal S}H_{\mathcal S}
       +\gamma_{\mathcal F}H_{\mathcal F}
       +\gamma_{\mathcal S\mathcal F}H_{\mathcal S\mathcal F})].
  \label{eq:grouped_phase_separator}
\end{equation}
Because these components are diagonal and commute, grouping changes only the
phase-angle parameterization of a fixed cost layer. Freeing mixer angles changes
transport on the feasible swap graph in addition to the diagonal phase
separator. Here ``transport'' means the probability-mass evolution induced by
the Block~XY mixer on that graph.

The decomposition yields the three CG parameterizations summarized in
Table~\ref{tab:qaoa_parameterizations}. Bilinear-CG frees only the
sample--feature coupling phase. Fixed-transport grouped CG frees the three
diagonal cost-component phases. Fully grouped CG also frees the mixer angle
layer by layer. The fully grouped variant is the most expressive same-depth
schedule tested in noiseless simulation. Bilinear-CG keeps mixer transport fixed
and is the lower-dimensional hardware variant used for the 32--64-qubit fixed-angle
sparse circuits.

Supplementary Section~\ref{supp:qaoa_grouped_multiangle} gives the
deterministic variational-containment and local-response results used for the
CG-XY-QAOA comparison. It also gives a
criterion for when grouped depth \(p\) remains preferable to tied depth \(2p\)
if hardware preserves only part of the additional-depth noiseless gain.

The variational-containment statements compare schedules at fixed depth and
fixed implementation choices. Fixed-transport grouped CG contains the tied-cost
schedule when the mixer transport remains shared. Fully grouped CG contains the
layerwise XY-\qaoa{} parameterization at the same depth. A small-angle response
analysis sharpens the fixed-mixer phase comparison to a strict improvement
whenever the per-component responses are not collinear with the tied direction.
Let \(r=(r_{\mathcal S},r_{\mathcal F},r_{\mathcal S\mathcal F})\) denote the
three-component small-angle mixed response of the expected cost in the grouped
phase angles at the Dicke-initialized origin
(Supplementary Eq.~\ref{eq:grouped_component_response}).

\begin{corollary}[Equal-norm grouped descent direction]
\label{cor:grouped-equal-norm-descent}
Assume $r\neq 0$ and fix a small mixer angle $\beta=\epsilon$. Among phase
displacements at fixed Euclidean norm equal to that of a tied displacement
$(\pm\epsilon,\pm\epsilon,\pm\epsilon)$, the predicted descent is maximized by
the grouped vector antiparallel to $r$. Relative to the better signed tied
direction, the predicted descent improves by
\[
  \epsilon^2\sqrt{3}\,\|r\|_2
  \left(1-|\cos\angle(r,\mathbf{1})|\right),
\]
which is zero exactly when $r$ is collinear with $\mathbf{1}=(1,1,1)$ and
positive otherwise.
\end{corollary}
Proof in Supplementary Section~\ref{supp:qaoa_grouped_multiangle}.

\subsection{Cost-Term Sparsification}
\label{sec:cost-term-sparsification}

Implementing all \(ND\) diagonal sample--feature couplings on sparse hardware
requires routing. For the submitted circuits, register-preserving variable
placement first permutes sample variables only within the sample register and
feature variables only within the feature register. If
\(E_{\times}\) is the set of executable cross-register position pairs, the
placement solves
\[
  (\pi_{\mathcal S},\pi_{\mathcal F})
  \in \arg\max_{\pi_{\mathcal S},\pi_{\mathcal F}}
  \sum_{(u,v)\in E_{\times}}
  \left|W_{\pi_{\mathcal S}(u),\pi_{\mathcal F}(v)}\right|,
\]
where both maps are permutations within their respective registers. The dense
QUBO is permuted consistently, leaving the optimization problem and both
cardinality constraints unchanged. Writing \(W\) in the placed logical order, a
binary mask \(M\) then selects which sample--feature couplings remain in the
phase separator,
\[
  \widetilde W=M\odot W,\qquad M_{ij}\in\{0,1\},
\]
where \(M\) is either a threshold/top-\(K\) term mask used for method
controls or the hardware-aligned mask induced by executable cross-register
edges in the selected layout. The sparse objective
\(\widetilde E(\vecs,\vecf)\) is
\[
  \widetilde E(\vecs,\vecf)=
  -\sum_i a_i s_i-\sum_j b_j f_j
  -\lambda\sum_{i,j}\widetilde W_{ij}s_if_j .
\]
The Ising coefficients are recomputed from \(\widetilde W\), so
sparsification changes which diagonal cost terms appear in the Hamiltonian. The
hardware experiments retain hardware-edge coefficients without refitting or
adding routed couplings, so
the executed Hamiltonian is a hardware-compatible surrogate whose relation to
the dense objective is assessed by retained mass, dense rescoring, and sampling
metrics. Patch selection uses backend calibration data, while variable placement
is determined from the dense objective and fixed hardware support before
simulator or hardware outcomes are available. After layout-aware decoding, the inverse
within-register permutation restores the original sample and feature labels
before scoring. Exact-budget
initialization and Block~XY mixing are otherwise unchanged.
Retained bilinear mass \(R_{\mathcal S\mathcal F}\) and retained-mass ratio
\(\rho_{\mathcal S\mathcal F}\) quantify the absolute and relative
cross-coupling mass retained by the mask:
\[
  R_{\mathcal S\mathcal F}=\sum_{i,j}M_{ij}|W_{ij}|,\qquad
  \rho_{\mathcal S\mathcal F}=
  R_{\mathcal S\mathcal F}/\sum_{i,j}|W_{ij}|.
\]
The retained-mass ratio records how much of the dense bilinear objective remains
in the sparse phase separator. When \(\rho_{\mathcal S\mathcal F}\) is small,
the submitted circuit mainly optimizes the sparse surrogate. Dense rescoring
evaluates its samples under the original objective, whereas direct validation
of the dense bilinear coupling requires instances with more retained
sample--feature terms.

For threshold masks, the objective distortion also has a deterministic bound.
If every omitted coupling satisfies \(|W_{ij}|<\tau\), then, with \(\lambda\)
denoting the bilinear coefficient in Eq.~\eqref{eq:hamiltonian},
\begin{equation}
  |E(\vecs,\vecf)-\widetilde E(\vecs,\vecf)| \le \lambda km\tau,
  \label{eq:sparsification_bound}
\end{equation}
for every feasible \((\vecs,\vecf)\). When \(\lambda>0\), the
dense-objective optimality gap transfers within \(2\lambda km\tau\).
Hardware-aligned and top-\(K\) masks are evaluated empirically because they do
not generally satisfy a useful common threshold bound.

\subsection{Feasible Initialization}
\label{sec:init-options}

The Dicke state $\ket{D^n_k}$ is the uniform superposition over all $n$-qubit
computational basis states with Hamming weight
exactly~$k$~\cite{dicke1954coherence,bartschi2019deterministic}.
\begin{equation}
  \ket{D^n_k} = \binom{n}{k}^{-1/2} \sum_{\substack{x \in \{0,1\}^n \\[1pt] \HW(x) = k}}
  \ket{x}\,.
  \label{eq:dicke}
\end{equation}

For bipartite constraints $\lvert\vecs\rvert = k$ and
$\lvert\vecf\rvert = m$, we use the tensor product initialization
\(\ket{\psi_0}\).
\begin{equation}
  \ket{\psi_0} = \ket{D^N_k}_{\mathcal{S}} \otimes \ket{D^D_m}_{\mathcal{F}}\,,
  \label{eq:init_state}
\end{equation}
where subscripts~$\mathcal{S}$ and~$\mathcal{F}$ denote the sample and feature
registers. For the quantum ansatz, the feasible subspace is the Hilbert-space
span of computational-basis states in \(\Omega\). This state lies entirely in
that subspace.

Initialization trades off symmetric coverage of the feasible sector against
depth and hardware robustness. Any feasible $(k,m)$ initialization remains
feasible under Block~XY, so we select the lowest-depth preparation compatible
with the experiment.

\paragraph{Tensor Dicke initialization.}
The tensor Dicke state is uniform over all feasible selections. Its two
register preparations contribute additively to the gate count, although the
registers can be prepared in parallel. The entangling preparation overhead
limits its use here to noiseless statevector simulations and circuit-resource
analyses.

\paragraph{Feasible computational-basis initialization.}
A single feasible bitstring is prepared using only $X$ gates. This minimizes
preparation depth. The mixer must then spread the resulting initialization bias.
The fixed-threshold simulation studies and hardware executions use this
exact-feasible basis initialization unless explicitly stated otherwise.

\paragraph{Classical warm-start initialization.}
CG-XY-\qaoa{} can start from any feasible classical bitstring prepared with
\(X\) gates. The 36--64-qubit warm-start and matched-control protocols are
defined in Section~\ref{sec:implementation}.

\subsection{Block~XY Mixer}

The pairwise XY interaction \(\Hxy^{(i,j)}\) preserves total Hamming weight and
is given by
\begin{equation}
  \Hxy^{(i,j)} = \tfrac{1}{2}\bigl(X_i X_j + Y_i Y_j\bigr)
               = \sigma_i^+ \sigma_j^- + \sigma_i^- \sigma_j^+\,.
  \label{eq:xy_pair}
\end{equation}
With $\sigma^\pm=(X\pm iY)/2$, the products
$\sigma_i^+\sigma_j^-$ and $\sigma_i^-\sigma_j^+$ transfer one excitation
between qubits~$i$ and~$j$ without changing the total count.

We construct a \emph{block} XY mixer \(\Hmix\) that acts independently on each
register.
\begin{equation}
  \Hmix = \underbrace{\sum_{(i,i') \in E_{\mathcal{S}}} \Hxy^{(i,i')}}_{\text{sample register}}
        + \underbrace{\sum_{(j,j') \in E_{\mathcal{F}}} \Hxy^{(j,j')}}_{\text{feature register}}\,,
  \label{eq:block_mixer}
\end{equation}
where $E_{\mathcal{S}} \subseteq \{(i,i'): i,i' \in [N], i \neq i'\}$ and
$E_{\mathcal{F}} \subseteq \{(j,j'): j,j' \in [D], j \neq j'\}$ define the mixer
connectivity graphs within each register. Any such graph preserves Hamming
weight. Connectedness is needed for the mixer graph to be irreducible on the
fixed-weight sector of that register. Common choices include ring graphs (for
logical simplicity) or hardware-native edges (to minimize routing overhead).
The product mixer \(U_M(\beta)\) uses an ordered product of two-qubit XY gates
over edge-color classes,
\begin{equation}
  U_M(\beta)=
  \prod_c\prod_{(i,i')\in E_{\mathcal S}^{(c)}}
    e^{-i\beta H_{\mathrm{XY}}^{(i,i')}}
  \prod_c\prod_{(j,j')\in E_{\mathcal F}^{(c)}}
    e^{-i\beta H_{\mathrm{XY}}^{(j,j')}},
  \label{eq:block_mixer_product}
\end{equation}
where edges within a color class are disjoint. References to
$e^{-i\beta H_{\mathrm{mix}}}$ denote the ideal continuous-time mixer model.
The hardware circuits use the product mixer in
Eq.~\eqref{eq:block_mixer_product}. Because factors from different color
classes need not commute, their order defines the hardware transport ansatz.

\begin{proposition}[Block~XY preserves both cardinality constraints]
\label{prop:feasibility}
The Block~XY mixer~\eqref{eq:block_mixer} preserves both constraints
$\lvert\vecs\rvert = k$ and $\lvert\vecf\rvert = m$ independently. Starting from
any state supported on the $(k,m)$ feasible subspace, including the tensor
Dicke state~\eqref{eq:init_state} or an exact-feasible basis state, all
\qaoa{} iterates remain in the feasible subspace.
\end{proposition}
Proof in Supplementary Section~\ref{supp:elementary_formulation_lemmas}.
Proposition~\ref{prop:feasibility} gives the ideal constraint-preservation
reference for the decoded hardware-yield metrics defined in
Section~\ref{sec:experiments}.
The trainability considerations for this symmetry-restricted ansatz are
discussed in Supplementary Section~\ref{supp:trainability}.

\subsection{Algorithm}

Algorithm~\ref{alg:joint_qaoa} summarizes the joint CG-XY-QAOA
procedure. The depth parameter \(p \geq 1\) controls the number of cost--mixer
alternations. Table~\ref{tab:qaoa_parameterizations} lists the CG
parameterizations used in the experiments. For hardware runs, the edge-colored
Block~XY transport order, parameterization, and angles are fixed before
execution.
The XY-\qaoa{} baselines use the same Block~XY mixer with the ungrouped cost
phase \(U_C^{(\ell)}=\exp[-i\gamma_\ell(H_{\mathcal S}+H_{\mathcal F}
+H_{\mathcal S\mathcal F})]\). The layer-shared hardware baseline uses the
same cost form with a single shared \(\gamma\) and \(\beta\).
Figure~\ref{fig:circuit}
illustrates the logical circuit structure for a single $p=1$ layer.
Within a fixed implementation path, depth, and hardware patch, these variants use
the same submitted two-qubit structure. They differ only in whether the diagonal phase
and mixer angle parameters are tied or freed. Hardware comparisons state the
implementation path separately.

\begin{table}[!htbp]
\centering
\footnotesize
\setlength{\tabcolsep}{3pt}
\renewcommand{\arraystretch}{1.02}
\caption{CG-XY-QAOA angle schedules. Here \(\mathcal S\), \(\mathcal F\), and
\(\mathcal S\mathcal F\) denote the sample, feature, and sample--feature cost
terms. Tied XY-QAOA is recovered by setting all cost angles equal and sharing
the mixer angle.}
\label{tab:qaoa_parameterizations}
\TableRowsSetup
\begin{tabularx}{\linewidth}{lXXr}
\TableHideRows
\toprule
Variant & Cost-angle schedule & Mixer schedule & Count \\
\midrule
\TableBodyRows
Bilinear-CG & shared \(\gamma_{\mathrm{marg}}\) with layerwise
\(\gamma_{\mathcal S\mathcal F,\ell}\) & shared \(\beta\) & \(p+2\) \\
Fixed-transport grouped & layerwise \(\gamma_{\mathcal S,\ell}\),
\(\gamma_{\mathcal F,\ell}\), \(\gamma_{\mathcal S\mathcal F,\ell}\) &
shared \(\beta\) & \(3p+1\) \\
Fully grouped & layerwise \(\gamma_{\mathcal S,\ell}\),
\(\gamma_{\mathcal F,\ell}\), \(\gamma_{\mathcal S\mathcal F,\ell}\) &
layerwise \(\beta_\ell\) & \(4p\) \\
\bottomrule
\end{tabularx}
\end{table}

\begin{algorithm}[!htbp]
\normalsize
\caption{CG-XY-QAOA construction and sampling}
\label{alg:joint_qaoa}
\begin{algorithmic}[1]
\Require Data matrix $\matX \in \mathbb{R}^{N \times D}$, calibration reference
set $\mathcal R$, budgets $k,m$, depth $p$, CG parameterization $\mathcal P$
from Table~\ref{tab:qaoa_parameterizations}, selected angles
$\boldsymbol{\theta}_{\mathcal P}$ compatible with $\mathcal P$
\Ensure Best observed feasible sample and feature sets $\widehat S,\widehat F$
\State Fit feature-wise location/scale on $\mathcal{R}$ and compute residuals \(Z_{ij} \gets (X_{ij} - \hat{\mu}_j) / \hat{\sigma}_j\)
\State Calibrate cross-term weights $W_{ij}$ from $Z_{ij}$ (Eq.~\eqref{eq:weights})
\State Build \qubo{}/Ising cost coefficients from $W_{ij}$, including the chosen signs and scales
\State Decompose the diagonal cost as $H_{\mathcal S}+H_{\mathcal F}+H_{\mathcal S\mathcal F}$ (Eq.~\eqref{eq:grouped_cost_decomposition})
\State Initialize $\ket{\psi_0}$ using tensor Dicke, a feasible basis state, or a classical warm start from Section~\ref{sec:init-options}
\State Instantiate the CG phase separator specified by $\mathcal P$ using
$\boldsymbol{\theta}_{\mathcal P}$
\Statex \quad Bilinear-CG: \(U_C^{(\ell)}=\exp[-i\{\gamma_{\rm marg}(H_{\mathcal S}+H_{\mathcal F})+\gamma_{{\mathcal S\mathcal F},\ell}H_{\mathcal S\mathcal F}\}]\)
\Statex \quad Grouped CG variants: \(U_C^{(\ell)}=\exp[-i(\gamma_{{\mathcal S},\ell}H_{\mathcal S}+\gamma_{{\mathcal F},\ell}H_{\mathcal F}+\gamma_{{\mathcal S\mathcal F},\ell}H_{\mathcal S\mathcal F})]\)
\For{$\ell = 1, \ldots, p$}
  \State Apply the selected diagonal phase separator \(U_C^{(\ell)}(\mathcal P)\)
  \State Apply the product Block~XY mixer \(U_M(\beta_\ell)\)
  \Statex \quad using a shared mixer angle for bilinear-CG and fixed-transport grouped CG,
  \Statex \quad and layerwise \(\beta_\ell\) for fully grouped CG
\EndFor
	\State Measure \(N_{\rm shots}\) shots and apply the recorded layout decoder when routing is present
		\State Score decoded exact-budget samples under the selected dense or sparse objective
		\State \Return best observed feasible pair $(\widehat S,\widehat F)$
\end{algorithmic}
\end{algorithm}

\paragraph{Classical outer loop.}
For a chosen parameterization \(\mathcal P\), a classical search proposes an
angle vector \(\boldsymbol{\theta}_{\mathcal P}\), evaluates a scalar objective
from the resulting circuit distribution, and uses that value to propose the
next vector. The search fixes the selected angles for final sampling when it
reaches its evaluation budget or stopping criterion. The reported
angle-selection studies evaluate candidates by statevector, shot-based, or
tensor-network simulation, and all hardware angles are fixed before
submission. Algorithm~\ref{alg:joint_qaoa} describes circuit construction and
sampling for a supplied angle vector. Section~\ref{sec:implementation} and
Supplementary Section~\ref{supp:qaoa_angle_protocols} specify the objective,
optimizer, and budget used by each experiment.

\begin{figure}[!htbp]
  \centering
  \begingroup
  \resizebox{0.85\linewidth}{!}{\input{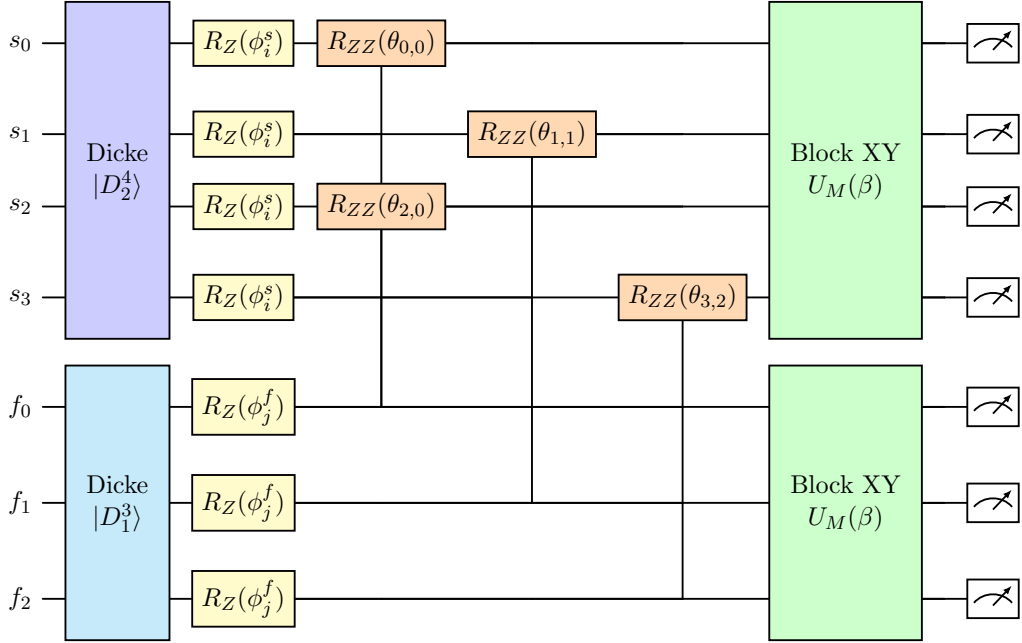}}
  \endgroup
  \caption{Logical CG-XY-\qaoa{} circuit for one layer of exact-budget
  bipartite selection ($p=1$). The example uses $N = 4$, $k = 2$, $D = 3$,
  and $m = 1$ with tensor Dicke
  initialization. Hardware runs use the feasible-basis and classical warm-start
  initializations in Section~\ref{sec:init-options}. The diagonal cost layer
  applies sample and feature fields together with retained sample--feature
  $R_{ZZ}$ rotations. Independent Block~XY mixers then preserve each register's
  Hamming weight. For $p>1$, the cost and mixer layers repeat with the
  parameterization in
  Table~\ref{tab:qaoa_parameterizations}.}
  \label{fig:circuit}
\end{figure}

\FloatBarrier

\subsection{Complexity Analysis}

\emph{Problem complexity.}
The feasible search space has size
$\binom{N}{k}\binom{D}{m}$ under the exact cardinality constraints. The decision
form is \textsc{NP}-complete by reduction from Balanced Complete Bipartite
Subgraph~\cite{garey1979computers}, so the optimization form is \textsc{NP}-hard. Supplementary
Proposition~\ref{prop:cbs_complexity} gives the short reduction in this
notation. The reduction sets \(W_{ij}\) to the input graph's edge indicator and
sets \(k=m=K\), asking whether a feasible \(K\)-by-\(K\) subgraph contains all
\(K^2\) cross edges. The exact-budget classical baselines and
constrained-\qaoa{} method both face this combinatorial search space.

\emph{Circuit resources.}
For the quantum realization, the Block~XY mixer requires $\mathcal{O}(N + D)$
two-qubit gates per layer on ring connectivity, while the dense bipartite cost
layer contributes $\mathcal{O}(ND)$ commuting $ZZ$ terms before routing. On a
fixed-degree hardware graph, hardware-aligned sparsification reduces the
retained cost layer to $\mathcal{O}(|E_{\mathrm{hw}}|)$ terms, making the
submitted resource burden depend on the chosen retained hardware edges rather
than on the full logical bipartite complete graph.

The deterministic preparation circuit of B{\"a}rtschi and
Eidenbenz~\cite{bartschi2019deterministic} uses $\mathcal{O}(Nk)$ and
$\mathcal{O}(Dm)$ two-qubit gates for the sample and feature registers,
respectively, giving $\mathcal{O}(Nk+Dm)$ gates in total. Parallel preparation
on suitable nearest-neighbor chains gives depth
$\mathcal{O}(\max\{N,D\})$.
This preparation cost is incurred once per circuit execution, while the cost
and mixer layers repeat $p$ times. For fixed $k,m$ and comparable register
sizes, preparation has linear gate-count scaling, whereas a dense cost layer
has quadratic scaling. When $k$ and $m$ grow proportionally with their
register sizes, both preparation and the dense cost layer have quadratic
scaling. On a fixed-degree hardware patch of $\mathcal{O}(N+D)$ qubits,
the sparse cost and ring mixer instead use $\mathcal{O}(N+D)$ two-qubit gates
per layer before additional routing. When these interactions are directly
executable, edge-colored scheduling permits bounded depth per
layer~\cite{hadfield2019quantum,wang2020xy,west2001introduction}.
The linear-depth Dicke construction can therefore dominate shallow sparse
circuits.

Short-depth preparation circuits~\cite{bartschi2022shortdepth} achieve depth
$\mathcal{O}(k\log(n/k))$ on all-to-all connectivity and
$\mathcal{O}(\sqrt{nk})$ on suitable grid-like connectivity for
$\ket{D^n_k}$, while retaining $\mathcal{O}(nk)$ gates. These bounds use
$k\le n/2$, with bit complementation for larger weights. The grid depth bound
remains linear when the weight grows proportionally with register size, and
the entangling gate-count scaling is unchanged. Feasible basis-state
initialization requires only $k+m$
single-qubit $X$ gates, executable in one parallel layer, and no entangling
preparation gates.

\section{Experimental Results}
\label{sec:experiments}

Synthetic and public-benchmark tests measure when calibrated joint exact-budget
selection improves over feature-first sequential baselines under the same
budgets. Statevector simulations, fixed-threshold simulations, and fixed-angle
hardware runs evaluate whether grouped phase schedules improve quantum
sampling metrics for the same objective, including the lower-dimensional
bilinear-CG variant with fixed mixer transport.

Transpilation studies quantify submitted circuit resources, while width-scaling
and classical warm-start runs measure exact-budget shot mass and the low-energy
tail of feasible samples on IBM Heron R3 hardware.

\subsection{Experimental Setup}
\label{sec:implementation}

Supplementary Section~\ref{supp:notation} collects recurring notation and
sampling metrics, while Supplementary Section~\ref{supp:hardware_detail}
defines submitted-circuit resource metrics.
Supplementary Section~\ref{supp:repro} gives the run
protocols and hyperparameters. Supplementary Sections~\ref{supp:creditcard_checks}
and~\ref{supp:hardware_detail} give the public-benchmark protocol checks and
hardware/transpiler settings. All Qiskit-based statevector and
transpilation-resource analyses use Qiskit~2.1.2~\cite{qiskit}. Finite-shot
readout tables sample from optimized state distributions, while expected-energy
summaries compute expectations directly from those distributions.

\paragraph{Evaluation settings and metrics.}
Synthetic experiments have full sample and
feature ground truth. Labeled public-benchmark experiments provide anomaly
labels without feature ground truth. Unlabeled analyses use objective values
because recovery scores are unavailable. For synthetic experiments we report
combined F1,
\[
\tfrac{1}{2}(\mathrm{F1@k}_{\text{samples}}+\mathrm{F1@m}_{\text{features}}),
\]
which reduces to the hit fractions $\lvert \hat S \cap S^\star\rvert/k$ and
$\lvert \hat F \cap F^\star\rvert/m$ because predicted and planted sets have
matching fixed sizes. On labeled benchmark data we report anomaly F1@k as the
set-F1 score
\[
  \mathrm{F1@k}=\frac{2|\hat S\cap S^+|}{|\hat S|+|S^+|},
\]
where $S^+$ is the positive set in the candidate pool. In the balanced
public-benchmark evaluations, $|\hat S|=|S^+|=k$, so anomaly F1@k equals precision@k and
$1-\mathrm{F1@k}$ is the false-review rate. Throughout, $\Delta$F1 denotes
joint minus sequential.

\paragraph{Quantum sampling and postprocessing.}
Noiseless depth-frontier simulations use the logical constraint-preserving
ansatz with tensor-Dicke initialization and ring Block~XY mixing. Angles are
optimized with the depth-ladder optimizer and bet-and-run budget. The
primary simulator metric is the optimized feasible-sector score
\(\alpha=(\bar C_\Omega-\mathbb E[C])/(\bar C_\Omega-C^\star)\), where
\(C^\star\) is the exact feasible optimum and
\(\bar C_\Omega=|\Omega|^{-1}\sum_{x\in\Omega}C(x)\) is the uniform-feasible
reference.
Finite-shot readout quality is computed from exact-budget samples of the
optimized distribution.

The fixed-threshold CG-XY-QAOA simulations select by best-known-threshold (BK)
hit rate \(p_{\rm BK}\). For a
threshold \(E_{\rm BK}\) fixed before the corresponding hardware run,
\[
  p_{\rm BK}=P\{x\in\Omega,\ E(x)\le E_{\rm BK}\}.
\]
The threshold is taken from the pre-run threshold manifest, which records the
best feasible energy available from the corresponding classical or simulation
evidence pool when the threshold was fixed.
This quantity is the per-shot probability of an exact-feasible sample at or below
that energy.
Mean energy and the conditional value-at-risk metric \(\mathrm{CVaR}_5(E)\),
defined here as the mean energy of the lowest-energy 5\% of decoded-feasible
shots, are secondary distribution metrics. These
simulations use exact-feasible basis initialization to match the hardware
protocol.

The noiseless depth studies report every optimized trajectory from the
prespecified \(p=1,\ldots,8\) grids. None is excluded for nonmonotonicity. Mean
\(\alpha\) increases at each successive depth in every displayed
perturbation-mode and signal-strength combination. Exact-budget XY-\qaoa{}
evolves inside a constrained feasible sector, so decoded exact-budget mass and
approximation-score scaling are shown alongside the depth budget.

Hardware experiments use the same logical cost and edge-colored Block~XY
transport structure with protocol-specific angle selection. Fixed-angle tied
hardware runs measure exact-budget sampling after final-layout decoding. Fully grouped
cost-plus-mixer experiments test transport freedom inside the CG-XY-QAOA
ansatz. Fixed-threshold and 32--64-qubit fixed-angle bilinear-CG runs keep tied transport
fixed. The fixed-threshold runs use preselected component phases. The
32--64-qubit constant-angle runs lie on the tied-phase subspace and test
sparse-circuit execution scale rather than a grouping benefit.
Supplementary Table~\ref{tab:penalty_x_vs_strict_feasible_qaoa}
reports a penalty-X control, which uses an unconstrained X mixer with
cardinality penalties.

\paragraph{Initialization protocols and tensor-network angle selection.}
The hardware scaling analysis uses three initialization protocols.
Table~\ref{tab:hardware_run_summary} reports the fixed-angle and warm-start
summaries, while Supplementary
Table~\ref{tab:cgxy_fixed_angle_condition_details} gives the matched
random-feasible controls.
The \emph{fixed-angle} protocol prepares
the canonical exact-budget basis state
\(\ket{1^k0^{N-k}}\otimes\ket{1^m0^{D-m}}\) and runs the circuit at
QUBO-independent constant tied angles
\(\gamma_\ell=\beta_\ell=\pi/4\). The \emph{classical warm-start} protocol prepares
the best feasible bitstring from the dense-\(W\) heuristic pool as a basis state
and applies the selected bilinear-CG circuit. The \emph{random-feasible}
protocol uses an independently sampled feasible basis state with the same angle
schedule, isolating the effect of the warm-start bitstring.

For the 36-qubit \(p=3\), 52-qubit \(p=2\), and 64-qubit \(p=2\)
warm-start executions and matched random-feasible controls, angles are selected
before hardware execution by a tensor-network angle-selection study implemented
with cuQuantum~\cite{bayraktar2023cuquantum}. The study evaluates candidate
angle vectors using the submitted bilinear-CG sparse-cost circuit. It decodes
sampled bitstrings, retains the exact \((k,m)\) sector, rescores the retained
samples on the original dense QUBO, and ranks candidates by dense-objective
\(\mathrm{CVaR}_5(E)\). Consequently, warm-start executions use dense-QUBO
information before hardware execution, while fixed-angle executions use the
canonical feasible basis state and QUBO-independent constant tied angles. Exact run
settings are specified in Supplementary Section~\ref{supp:repro}.

\paragraph{Public-benchmark protocol and baselines.}
The public-benchmark evaluation uses time-split candidate pools. The split
matches the deployment pattern in which calibration and baseline choices are
fixed before later review decisions are evaluated. All methods solve
the same fixed-$(k,m)$ stage-2 problem on later-time slices after calibration is
fit on earlier windows. The feature-first sequential aggregation rule is
selected on the held-out calibration window and then fixed for test-window
evaluation. The same frozen stage-2 scoring map defines the objective
coefficients for both feature-first and joint selectors. Only the selection
procedure differs.
Let \(A_{ij}\ge 0\) denote the frozen stage-2 score matrix, which plays the
role of \(W_{ij}\) in the slice objective. The
feature-first sequential rule first aggregates each feature column by maximum,
sum, or median and keeps the top \(m\) features. It then applies the same
aggregation over the selected feature set to score samples and keeps the top
\(k\) samples. This rule is the matched stage-2 baseline for the fixed-budget
selection problem. Full anomaly-detection pipelines such as Isolation Forest or
Local Outlier Factor address a different sample-scoring task.

The primary public-benchmark comparisons use balanced candidate pools for both
Credit Card and IBM. The Credit Card setting uses label-free feature-local
stage-2 calibration, while the IBM setting uses a supervised stage-2 scorer.
IBM trigger-based label-free controls are in
Supplementary Table~\ref{tab:joint_vs_sequential_real_ibm_aml_hi_small_trigger_menu_stage1_baseline}.
Supplementary Table~\ref{tab:solver_matched_compute} compares the methods under
matched compute. Sequential feature-first rules are cheaper per call.
The joint heuristics perform additional local-search work to optimize the
coupled objective under the same fixed stage-2 budgets. Reported confidence
intervals are conditional on the selected split, calibration map, solver set,
and protocol fields specified for each experiment. Protocol-selection
uncertainty is outside those intervals.

\subsection{Classical Baselines}
\label{sec:classical-baselines}

The classical comparisons use constraint-respecting optimizers with different
roles. Swap-SA denotes the full-neighborhood feasible-subspace simulated
annealing used in the tables. It is a problem-specific implementation of
simulated annealing on fixed-cardinality bipartite subsets. The search state is
a pair \(S\subseteq[N]\), \(F\subseteq[D]\) with \(|S|=k\) and \(|F|=m\). Each
proposal is a sample, feature, or coupled swap that remains in this feasible
sector. Non-improving proposals are accepted with a Metropolis-style
temperature rule under fixed or annealed schedules
\cite{metropolis1953equation,kirkpatrick1983optimization}. Greedy local search,
coordinate ascent, and tabu search provide deterministic or memory-based joint
references and sometimes tie or exceed Swap-SA in these comparisons.

The ring-local Kawasaki Metropolis reference isolates the effect of the
feasible swap graph. It uses
nearest-neighbor token exchanges on the same ring-local adjacency graph as the
Block~XY mixer. These exchanges are the standard local moves in Kawasaki and
exclusion-process dynamics~\cite{kawasaki1972kinetics,spitzer1970interaction}.
It serves as a neighborhood-matched classical walk on the feasible graph.
The full-neighborhood classical comparisons use Swap-SA, greedy local search,
coordinate ascent, and tabu search. Coherent Block~XY \qaoa{} uses the same
locality constraint together with superposition, cost-phase interference, and
variational angle optimization.
Supplementary Section~\ref{supp:swap_sa} gives the Swap-SA procedure and
specifies the ring-local reference. Supplementary
Table~\ref{tab:solver_matched_compute} compares the matched wall-clock performance of
the classical heuristics on the synthetic stage-2 setting.

\FloatBarrier
\subsection{Formulation-Level Selection Results}

The formulation-level experiments test the calibrated objective before adding
quantum execution effects. Synthetic instances test recovery of planted
sample--feature structure under controlled ground truth. The Credit Card and IBM
IT-AML slices test the same constrained formulation on candidate pools held
fixed across methods.
Table~\ref{tab:formulation_evidence_summary} summarizes representative settings and
names the joint heuristic attaining each displayed value.

\begin{table}[!htbp]
\centering
\caption[Formulation-level selection results]{Formulation-level selection results. Joint and sequential selection tie at zero synthetic noise. The best prespecified joint heuristic improves every nonzero-noise synthetic row and every displayed public-benchmark row. Entries are means comparing the reported feature-first sequential baseline with the best result among the prespecified full-neighborhood joint heuristics for the same fixed-budget objective. Synthetic rows average 200 seeds, and public-benchmark rows average 25 time-split slices. \(\Delta\) is joint minus sequential in percentage points.}
\label{tab:formulation_evidence_summary}
\small
\setlength{\tabcolsep}{1.95pt}
\renewcommand{\arraystretch}{1.02}
\TableRowsSetup
\begin{tabular}{ l l l r r r r l }
\TableHideRows
\toprule
Benchmark & Parameters & Metric & Seeds / slices & Sequential & Joint & \(\Delta\) pp & Joint heuristic \\
\midrule
\TableBodyRows
Synthetic & $\nu=0\%$ & combined F1 & 200 & 0.833 & 0.833 & +0.0 & All joint methods \\
Synthetic & $\nu=20\%$ & combined F1 & 200 & 0.854 & 0.885 & +3.1 & All joint methods \\
Synthetic & $\nu=40\%$ & combined F1 & 200 & 0.821 & 0.955 & +13.4 & Swap-SA/Coord.-ascent \\
Synthetic & $\nu=50\%$ & combined F1 & 200 & 0.791 & 1.000 & +20.9 & Swap-SA/Coord.-ascent \\
Synthetic & $\nu=60\%$ & combined F1 & 200 & 0.752 & 0.944 & +19.2 & Swap-SA/Coord.-ascent \\
Credit Card & $(20,12,3,4)$ & F1@k & 25 & 0.800 & 0.827 & +2.7 & Swap-SA \\
Credit Card & $(50,12,5,4)$ & F1@k & 25 & 0.664 & 0.760 & +9.6 & Coord.-ascent \\
Credit Card & $(80,12,8,4)$ & F1@k & 25 & 0.745 & 0.795 & +5.0 & Greedy+LS \\
IBM IT-AML & $(20,38,3,4)$ & F1@k & 25 & 0.467 & 0.720 & +25.3 & Tabu \\
IBM IT-AML & $(50,38,5,4)$ & F1@k & 25 & 0.120 & 0.400 & +28.0 & Greedy+LS \\
\bottomrule
\end{tabular}
\end{table}

\paragraph{Synthetic benchmarks.}
We generate synthetic instances by sampling $X\in\mathbb{R}^{N\times D}$ from a
base distribution, sampling $k$ anomaly indices, and injecting mean shifts on an
informative feature subset. A noise fraction $\nu\in[0,1]$ controls what
proportion of features are uninformative, and for these experiments we use the
isolation-style weights $W_{ij}=\lvert Z_{ij}\rvert$ to isolate the effect of
joint selection from calibration choices. After standardization, the injected
shifts raise the expected absolute residuals on the planted sample--feature
block, providing the empirical analogue of the centered score-field model. The
sequential reference is feature-first. It ranks features by an aggregation score over samples, keeps the
top-$m$, then ranks samples using only those features and keeps the top-$k$.

On synthetic data, the sequential reference is the maximum, sum, or median
aggregation rule with the highest planted combined F1. The synthetic noise
sweep uses
\((N,D,k,m)=(\JointSeqNSamples,\JointSeqNFeatures,\JointSeqKAnomalies,\JointSeqMFeatures)\),
averaged over \JointSeqNSeeds{} random seeds. At this size the exact joint
optimum is computable by enumerating all feasible $(S,F)$ pairs, so
Supplementary Table~\ref{tab:joint_vs_sequential_practical} compares the
sequential baseline, the prespecified joint heuristics, and the exact optimum
where enumeration is possible.
The wider synthetic parameter sweep appears in
Figure~\ref{fig:joint_vs_sequential_param_sweep_heatmaps}, with bootstrap
intervals in Table~\ref{tab:joint_vs_sequential_ci}.

\paragraph{Time-split financial benchmarks.}
On Credit Card fraud~\cite{dalpozzolo2015calibrating,kaggle_creditcardfraud},
feature-local calibration maps are fit without labels on the reference window
and evaluated on later-time balanced candidate pools. Each candidate scorer
receives two held-out calibration-window scores: its best prespecified joint F1
and its best feature-first F1. Their average determines the scorer fixed for
test evaluation. The same window selects the feature-first
aggregation rule before test evaluation. The selected rule is sum
aggregation in all balanced settings. In every balanced Credit Card setting, the best reported member of
the prespecified joint heuristic set improves mean anomaly F1@k over this
calibration-selected sequential baseline. Supplementary
Table~\ref{tab:joint_vs_sequential_real_creditcard} gives the full balanced
panel and the fixed-\(k\) review-loss interpretation.

On IBM IT-AML HI-Small~\cite{altman2023realistic,kaggle_ibm_itaml}, the balanced
stage-2 setting selects from the full \(D=38\) engineered feature width. Weights
use a supervised logistic-regression scorer fit before the test period, and the
feature-first sequential rule uses \(\max\) aggregation. In both balanced IBM
settings, the best reported member of the prespecified joint heuristic set
improves mean anomaly F1@k over the feature-first baseline. The 95\%
percentile-bootstrap intervals do not overlap at either candidate-pool size.
The largest lift occurs in the full-width balanced setting shown in
Supplementary Table~\ref{tab:joint_vs_sequential_real_ibm_aml_hi_small}. In that
setting, feature-first ranking must choose features before the final transaction
subset is known. Sensitivity analyses appear in
Tables~\ref{tab:joint_vs_sequential_real_creditcard_protocol_stress}
and~\ref{tab:joint_vs_sequential_real_ibm_aml_hi_small_trigger_menu_stage1_baseline}.
\FloatBarrier

These results quantify the formulation-level differences
between joint and feature-first selection independently of quantum
implementation. Across the synthetic noise sweep, the best practical joint
methods match the feature-first baseline at zero noise. At 20--60\% noise, the
mean gains range from \JointSeqMinDeltaFOnePct{} to
\JointSeqPracticalMaxDeltaFOnePct{} percentage points, and every paired 95\%
bootstrap interval lies strictly above zero. At this small size,
Swap-SA and coordinate ascent reach the exact joint optimum at every tested
noise level. On the financial benchmarks,
the sign and size of the joint-minus-sequential gap depend on the stage-2
evaluation protocol. The balanced Credit Card and IBM comparisons yield
positive joint-minus-sequential gaps under their calibration protocols.
Sensitivity analyses show that the margin can narrow or disappear under
unbalanced pools or broader trigger-defined queues.

\subsection{CG-XY-QAOA Sampling Results}
\label{sec:quality}

On exact planted joint-\qubo{} instances, Block~XY mixing preserves strict
\((k,m)\) feasibility, so depth and readout metrics can be evaluated inside the
feasible sector. Three metrics distinguish noiseless optimization quality,
fixed-threshold sampling, and hardware feasible-sector yield. The optimized
feasible-sector score \(\alpha\) measures noiseless depth-frontier quality,
best-known-threshold hit rate measures fixed-threshold sampling, and decoded
exact-budget mass measures hardware yield. Exact-optimum probability is included only
where the feasible sector is enumerable. Mean energy and
\(\mathrm{CVaR}_5(E)\) are secondary metrics for the sampled energy
distribution~\cite{barkoutsos2020cvar}.

The CG-XY-QAOA variants in Table~\ref{tab:qaoa_parameterizations} separate
diagonal phase grouping from mixer-transport freedom:
Corollary~\ref{cor:grouped-equal-norm-descent} covers fixed-mixer component
phases, and Supplementary
Proposition~\ref{prop:fully-grouped-endpoint-containment} covers the fully
grouped cost-plus-mixer variant. Noiseless statevector simulations identify
fully grouped cost-plus-mixer CG-XY-QAOA as a same-depth refinement. It improves
\(\alpha\) in \CGXYGroupedSameDepthWinCount{} of
\CGXYGroupedComparisonCount{} same-depth comparisons and
recovers a median \CGXYGroupedRecoveryMedianRangePct{} of the optimized
\(p\to2p\) gain without adding logical layers. The grouped depth-\(p\) schedule
exceeds the parameter-matched layerwise XY-\qaoa{} depth-\(2p\) reference in
\CGXYGroupedMatchedDepthWinCount{} of \CGXYGroupedComparisonCount{} comparisons
(Supplementary Table~\ref{tab:qaoa_calibration_grouped_multiangle}). The planted
depth-frontier sweep in Figure~\ref{fig:cgxy_noiseless_evidence} shows positive
same-depth lift through \(p=8\) for both residual and mixed perturbations.
Supplementary Table~\ref{tab:cgxy_grouping_randomization_control} reports
parameter-matched randomization controls. Across
\CGXYGroupingNullCaseCount{} instances, the sample-field, feature-field, and
bilinear split has median percentile
\CGXYGroupingNullTypePreservingMedianPercentilePct{} under type-preserving
regrouping and \CGXYGroupingNullSameSizeMedianPercentilePct{} under unrestricted
same-size partitioning. The semantic split is therefore distinguished within
the sample, feature, and coupling taxonomy, not among arbitrary same-size
partitions.

The fixed-threshold simulations evaluate the hardware sampling metric \(p_{\rm BK}\)
defined in Section~\ref{sec:implementation}, with \(E_{\mathrm{BK}}\) fixed
before the associated hardware run. On the matched eight-instance planted
panel, fully grouped CG-XY-QAOA improves BK hit rate over tied XY-QAOA on all
eight instances and has the largest mean \(p_{\rm BK}\). Its paired 95\%
bootstrap lift interval is strictly above zero.
Table~\ref{tab:cgxy_best_known_threshold_summary_p3}
reports the aggregate variant-level lifts. The selected fully grouped schedule
is evaluated on hardware in Section~\ref{sec:transpilation}.

\begin{table}[!htbp]
\centering
\small
\setlength{\tabcolsep}{5pt}
\renewcommand{\arraystretch}{1.1}
\caption[CG-XY-QAOA BK-threshold selection at p=3]{\emph{CG-XY-QAOA BK-threshold selection at \(p=3\).} Fully grouped has the largest mean \(p_{\mathrm{BK}}\) and improves all eight instances relative to XY-QAOA. The noiseless simulation uses an eight-instance planted 20-qubit sparse-cost panel. All variants use the same exact-feasible basis state, greedy edge-colored Block~XY order, and exact first-phase elision. They report the mean \(p_{\mathrm{BK}}=P(E\le E_{\mathrm{BK}})\) with fixed per-instance thresholds. Lift entries are paired mean differences [bootstrap 95\% confidence interval] against XY-QAOA.}
\label{tab:cgxy_best_known_threshold_summary_p3}
\TableRowsSetup
\begin{tabular}{lrrr}
\TableHideRows
\toprule
Variant & Mean \(10^4p_{\mathrm{BK}}\) & \(10^4\Delta p_{\mathrm{BK}}\) vs XY-QAOA & Instances improved \\
\midrule
\TableBodyRows
XY-QAOA & 7.98 & reference & -- \\
Bilinear-CG & 8.60 & $+0.62$ [+0.15, +1.15] & 5/8 \\
Fixed-transport grouped & 28.8 & $+20.77$ [+2.47, +50.51] & 7/8 \\
Fully grouped & 39.6 & $+31.60$ [+11.43, +58.86] & 8/8 \\
\bottomrule
\end{tabular}
\end{table}

The enumerable \((N,D,k,m)=(14,10,3,3)\) depth series provides a
depth-resource reference for planted instances. The matched classical baseline
reaches the exact optimum on every instance in the series, while exact-optimum
readout mass is nonmonotone over \(p=1,\ldots,8\). The readout probabilities are in
Supplementary Table~\ref{tab:qaoa_sparse_exact_14_10_readout}.

\begin{figure}[!htb]
\centering
\includegraphics[width=\linewidth]{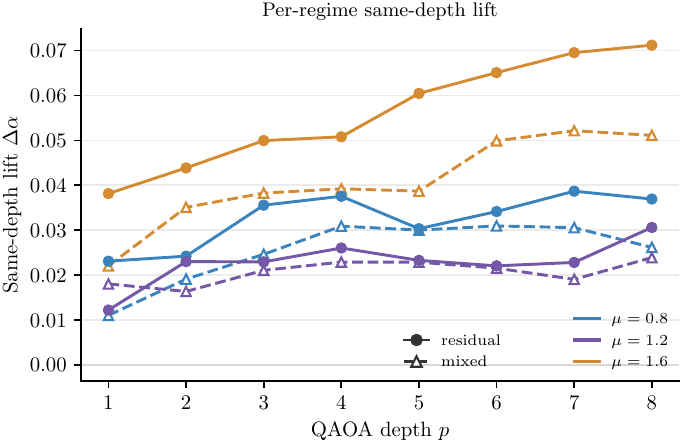}
\caption{Fully grouped cost-plus-mixer CG-XY-QAOA gives positive same-depth
noiseless score lift in every plotted perturbation cell. The panel reports
\(\Delta\alpha = \alpha_{\text{fully grouped}}-\alpha_{\text{tied}}\) over the
  planted-bicluster depth sweep. Solid circles denote residual perturbations,
  dashed triangles denote mixed perturbations, and color denotes signal strength
  \(\mu\). Here
\(\alpha=(\bar C_\Omega-\mathbb E[C])/(\bar C_\Omega-C^\star)\) equals one at
the exact feasible optimum and zero at the uniform-feasible reference. Full pooled
\(\alpha\) values are in Supplementary
Section~\ref{supp:cgxy_depth_resource_evidence}.}
\label{fig:cgxy_noiseless_evidence}
\end{figure}

\subsection{Cost Sparsification and Retained Coupling}
\label{sec:sparsification}

Register-preserving variable placement first permutes variables within their
respective registers. The dense QUBO is unchanged under the recorded
permutation. Hardware-aligned sparsification then removes nonhardware-edge
couplings from the phase separator and retains the remaining coefficients
without reweighting. The resulting sparse phase separator leaves the
exact-budget feasible sector and Block~XY mixer unchanged, so decoded
exact-budget mass remains a hardware-yield quantity while dense-objective energy
is a rescoring metric. Section~\ref{sec:cost-term-sparsification} defines the retained-mass
quantities used to compare the sparse surrogate with the dense objective. When
retained mass is small, the 32--64-qubit fixed-angle runs test
constraint-preserving execution of sparse surrogates dominated by marginal
fields, with limited direct evidence about the dense bilinear term. Dense
joint-selection quality is evaluated separately in the formulation-level
tables.

The hardware-aligned masks also reduce routing resources across every displayed
condition on both targets. Figure~\ref{fig:cross_platform_compilation} compares
dense and sparse costs under matched Qiskit optimization-level~3 configurations
on each target's native basis.

\begin{figure}[!htbp]
\centering
\includegraphics[width=0.94\linewidth]{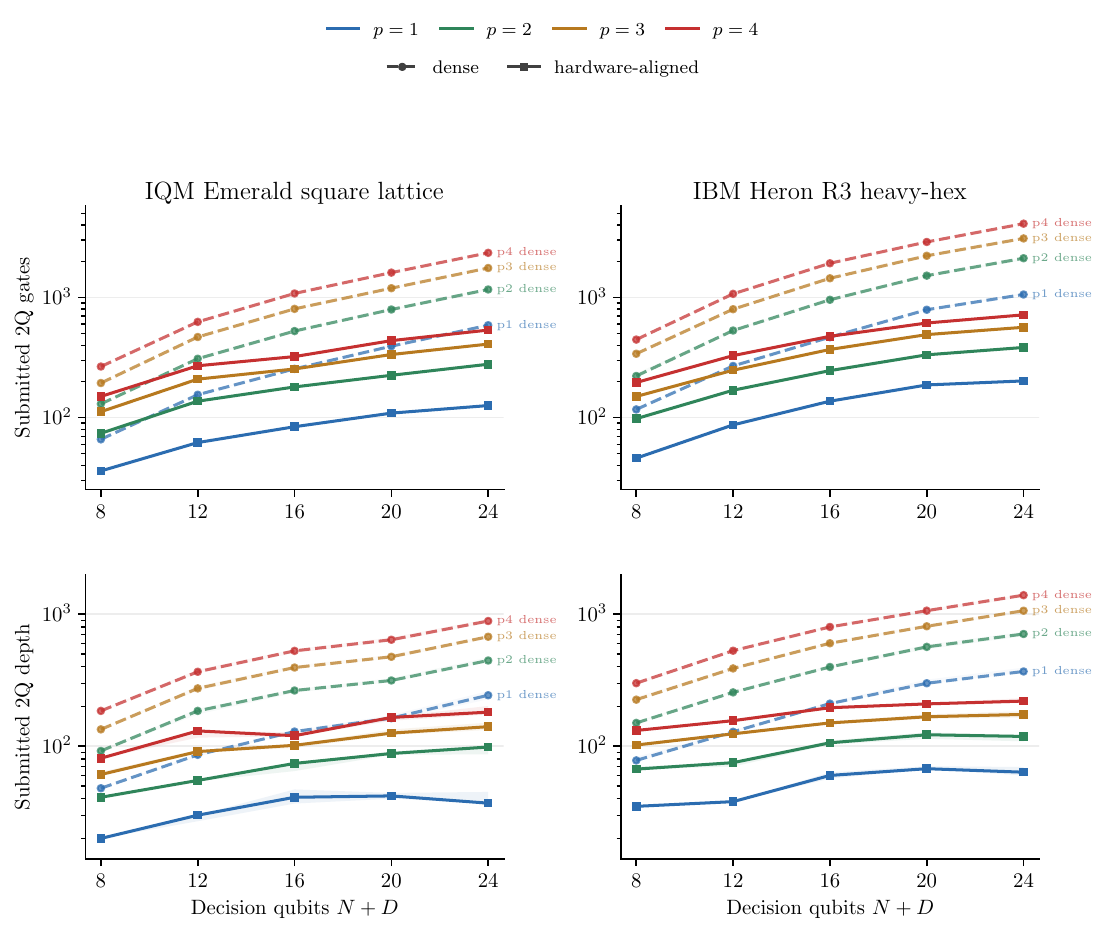}
\caption{\emph{Sparse-cost scaling across hardware targets.}
Hardware alignment lowers both resource measures in all
\HardwareAlignedCompileCellCount{} displayed target--width--depth cells. Across
the IQM Emerald square-lattice and IBM Heron R3 heavy-hex targets, the reduction
in median submitted two-qubit count has an observed range of
\HardwareAlignedTwoQCountReductionRangePct{}, and the reduction in median
submitted two-qubit depth has an observed range of
\HardwareAlignedTwoQDepthReductionRangePct{}. Each point is the median over
20 transpiler seeds, and bands show interquartile ranges. Solid lines denote
hardware-aligned costs, while dashed lines denote dense costs.}
\label{fig:cross_platform_compilation}
\end{figure}

The objective-level cost of sparsification is evaluated before hardware noise.
For each submitted 32--64-qubit condition, let \(x^\star_{\rm sp}\) be the exact sparse
optimum, \(E_d\) the dense objective, \(E_{\rm class,d}\) the strict-feasible
dense classical reference, and \(\bar E_{\rm RF}\) the matched random-feasible
mean. We report the dense-reference gain retained,
\begin{equation}
  g_{\rm dense}=
  \frac{\bar E_{\rm RF}-E_d(x^\star_{\rm sp})}
       {\bar E_{\rm RF}-E_{\rm class,d}}.
  \label{eq:dense_reference_gain_retained}
\end{equation}
This normalization equals one at the dense classical reference and zero at the
random-feasible mean. All five 32--64-qubit conditions lie strictly between these
anchors after dense rescoring. The sparse solutions therefore preserve part of
the dense-reference improvement. The remaining gap quantifies objective
distortion separately from hardware noise and angle selection.

Figure~\ref{fig:submitted_coupling_diagnostic} quantifies the cross-register
structure retained by the submitted Heron circuits from 20 to 64 decision
qubits. In every 32--64-qubit condition, the retained terms change the
sparse-objective optimum relative to diagonal-only exact-budget selection.
Supplementary Table~\ref{tab:cgxy_fixed_angle_condition_details} lists the
retained fraction for every Credit Card condition and compares strict and
band-1-repaired hardware samples with the corresponding classical
references.

\begin{figure}[!htbp]
\centering
\includegraphics[width=0.94\linewidth]{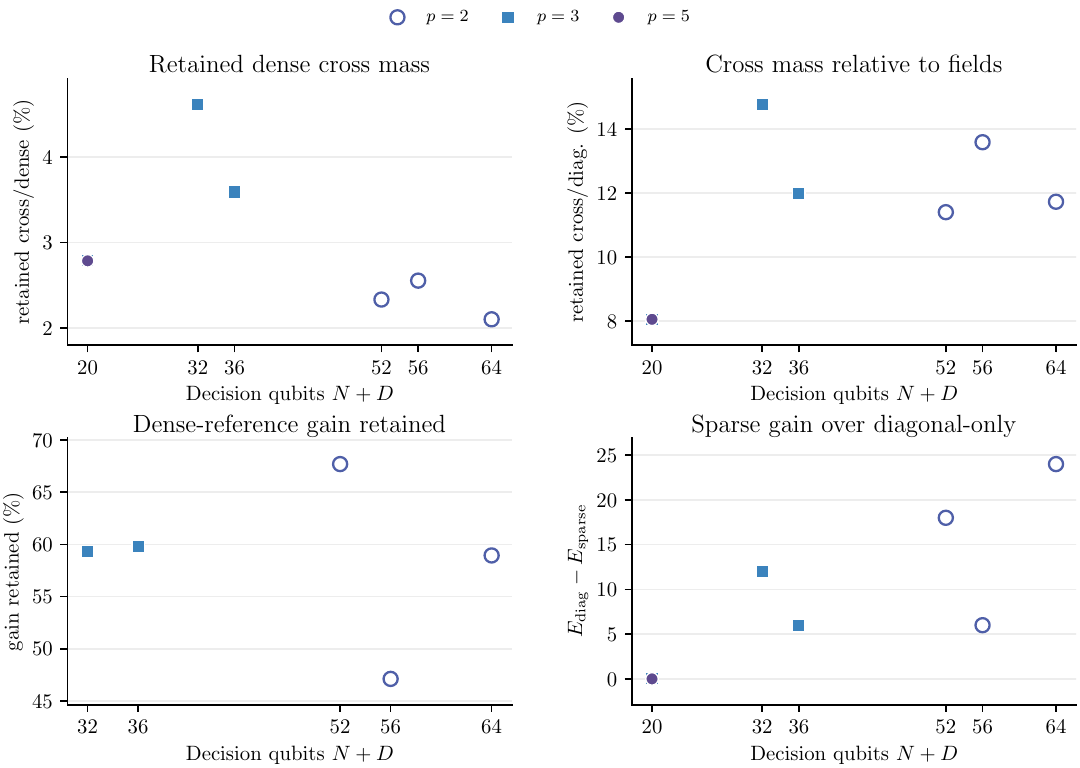}
\caption[Submitted cross-register coupling summary]{\emph{Exact sparse
solutions preserve part of the dense-reference improvement.} Across 32--64
decision qubits, the dense-reference gain retained has an observed range of
\CreditWideSparseDenseGainRetainedRangePct{}, and retained cross-register
\(L_1\) mass has an observed range of
\CreditWideRetainedCrossMassRangePct{}. Every 32--64-qubit sparse optimum
improves on the diagonal-only exact-budget choice. Panels report retained
cross-register mass relative to dense cross-register mass, retained
cross-register mass relative to diagonal-field mass, dense-reference gain
retained after exact sparse optimization, and sparse-objective gain over the
  diagonal-only choice. The 20-qubit grouped conditions have zero
  sparse-objective gain over that choice and lack the matched references needed
  to compute \(g_{\rm dense}\), so they are omitted from the dense-reference
  panel.}
\label{fig:submitted_coupling_diagnostic}
\end{figure}

\FloatBarrier
\needspace{6\baselineskip}
\subsection{Hardware Transpilation, Decoding, and Sampling}
\label{sec:transpilation}

The hardware analysis asks which logical properties survive backend
transpilation and noisy measurement: variable labels, exact-budget mass, and
low-energy feasible samples. It separates layout-aware decoding,
implementation-path resources, fixed-threshold CG-XY-\qaoa{} sampling, and
32--64-qubit sparse-cost sampling. All reported hardware energies use decoded
samples. The 32--64-qubit fixed-angle and 36--64-qubit warm-start energies are
dense-objective rescorings of samples produced by sparse surrogate phase
separators.

The hardware results compare decoded exact-budget mass with chance feasibility
\(p_{\rm chance}=\binom{N}{k}\binom{D}{m}/2^{N+D}\), the uniform-bitstring
probability of satisfying both budgets. The ratio
\(p_{\rm feas}/p_{\rm chance}\) reports the multiple over that level.
Dense-rescored energy evaluates measured
samples against the original dense objective, while sparse-objective energy
evaluates the same samples against the executed surrogate Hamiltonian. The
best-energy gaps summarize the observed low-energy tail against strict-feasible
classical references. Primary hardware metrics are computed on strict decoded
exact-budget samples before classical repair or local improvement. Band-1
repair accepts samples within one count of each budget and projects them to
\(\Omega\). It is evaluated separately.

\paragraph{Layout-aware decoding.}

Transpilation can preserve the logical unitary while permuting the measurement
bits associated with logical variables. In the
\HWCompBlockxyQubits-qubit routing check, physical-order parsing yields an
apparent exact-budget mass of \HWCompBlockxyFeasCompRaw{}. Decoding through the
recorded final layout restores \HWCompBlockxyFeasComp{} decoded exact-budget mass
(Lemma~\ref{lem:routing-feasibility}). Layout-aware decoding adds no gates and
is not included in resource attribution.

The implementation-path study uses a \(2\times2\times2\times2\)
factorial design over backend target, Block~XY fusion, transport order, and
parameter-binding order. Opt-3 CZ and Opt-3 fractional apply Qiskit's
optimization-level~3 preset pass manager to the same unfused circuit with
numerical angles bound before transpilation. They differ only in whether the
standard CZ-basis or fractional-gate target is exposed. The complete CZ and
fractional Block~XY paths bind angles after transpilation, fuse adjacent
equal-angle \(R_{XX}R_{YY}\) mixer pairs into \(\mathrm{XXPlusYY}\), and use
edge-colored transport order. Every fractional path includes exact native-range
\(R_{ZZ}\) folding and its local corrections. These factors interact, so the
study reports paired contrasts rather than additive shares of the total
reduction.

The target-based circuit-duration estimate is the longest dependency path in
the submitted circuit after weighting each instruction by its backend-target
duration. The estimate includes final measurement and excludes reset,
inter-shot delay, Runtime service overhead, and queue time.

\begin{figure}[!htb]
\centering
\includegraphics[width=0.90\linewidth]{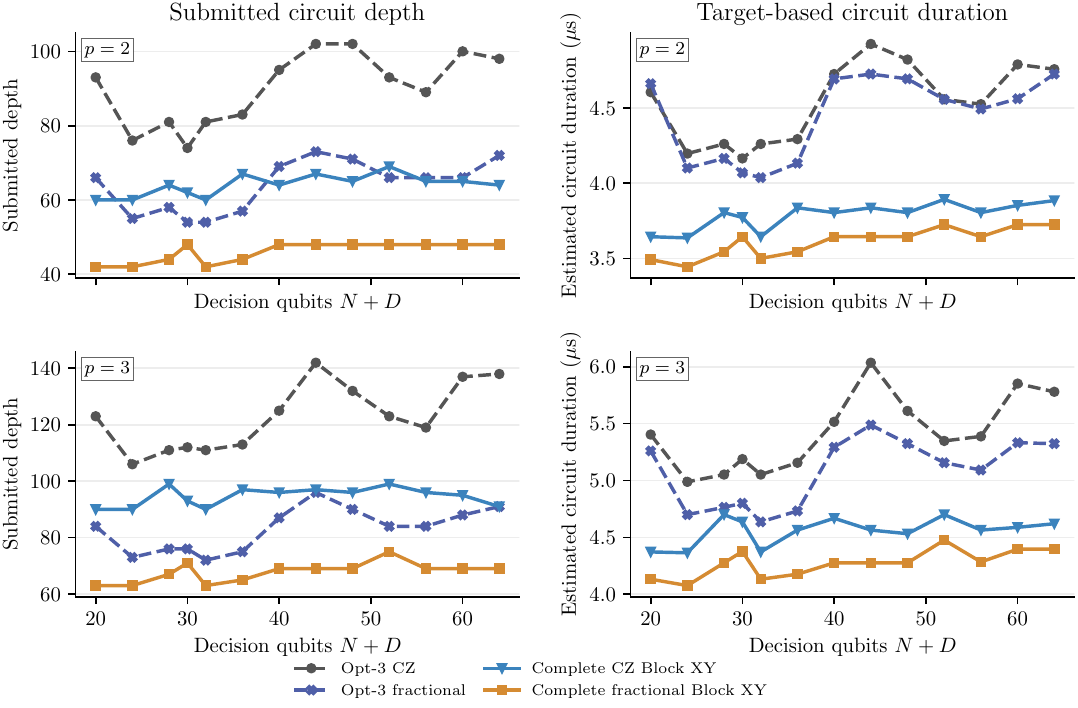}
\caption{\emph{The complete Block~XY paths reduce submitted depth on both the
CZ and fractional targets.} All paths use the same retained cost terms,
physical patch, initial layout, basis initialization, bilinear-CG
parameterization, deterministic numerical angles, backend snapshot, and
transpiler seed. Across \CompilePathMatchedPointCount{} matched width--depth
points, the complete fractional path reduces submitted circuit depth by
\CompilePathDepthReductionMedianPct{} at the median, with an observed range of
\CompilePathDepthReductionRangePct{}, relative to Opt-3 CZ. The target-based
circuit-duration reduction has median
\CompilePathDurationReductionMedianPct{} and observed range
\CompilePathDurationReductionRangePct{}. Native
two-qubit instruction count falls by
\CompilePathNativeTwoQCountReductionRangePct{}, entirely from exposing the
fractional target. The target alone does not reduce the duration estimate at
every point, whereas the complete path does. The
\CompilePathTranspilerSeedCount{} declared transpiler seeds give identical
resources within each condition. Figure~\ref{fig:cross_platform_compilation}
reports the separate dense-to-sparse support comparison, and Supplementary
Table~\ref{tab:implementation_path_factorial} reports the controlled contrasts.}
\label{fig:compiler_path_scaling_baseline}
\end{figure}

\paragraph{IBM Heron R3 hardware experiments.}

The hardware experiments run on an IBM Heron R3 heavy-hex processor using
calibration-selected backend patches. The fixed-threshold runs use \(p_{\rm BK}\), while the
32--64-qubit sparse-cost comparisons use decoded exact-budget mass,
dense-rescored energy, and sparse-objective energy.

The 20-qubit \((12,8,3,5)\) experiments comprise an eight-instance planted
fixed-threshold panel at \(p=3\) and a matched Credit Card comparison at
\(p\in\{3,5\}\).
The 32--64-qubit bilinear-CG fixed-angle runs compare decoded exact-budget mass and low-energy
feasible samples against chance-feasibility, random-feasible, and
strict-feasible classical references. Warm-start controls are run at 36, 52,
and 64 decision qubits.

The fractional path has an angle-domain constraint because native \(R_{ZZ}\)
gates are restricted to the backend-supported angle interval. The tied and
bilinear-CG angles used in the hardware runs fit this interval
directly. Fully grouped unrestricted logical angles can leave the interval. For
those schedules, exact modulo-\(\pi\) \(R_{ZZ}\) folding with local \(R_Z(\pi)\)
corrections preserves the logical phase separator while keeping the compact
fractional circuit.

In the \HWCompBlockxyQubits-qubit routing check, the penalty-based X~mixer has
much lower decoded exact-budget mass under the same transpilation setting
(Supplementary Table~\ref{tab:penalty_x_vs_strict_feasible_qaoa}), while the
constraint-preserving XY circuits recover their strict-feasible sector after
layout decoding.

\paragraph{Fixed-threshold CG-XY-QAOA hardware runs.}

The primary CG-XY-QAOA hardware metric is best-known-threshold hit rate.
Supplementary Table~\ref{tab:cgxy_fully_grouped_qpu_replay_p3} gives the
per-instance planted-panel simulation and matched Heron hardware runs for that metric.
The Heron executions use the selected fully grouped \(p=3\) schedule with the same
decision register, fractional-gate implementation, basis initialization,
decoder, and shot budget. Threshold hits occur on
\CGXYFullyGroupedQPUHitInstances{} instances, with
\CGXYFullyGroupedQPUThresholdHits{} hits in
\CGXYFullyGroupedQPUTotalShots{} submitted shots. The mean decoded exact-budget
mass is \CGXYFullyGroupedQPUMeanFeasPct{}.

On the simulation-selected benchmark instance,
Table~\ref{tab:creditcard_bilinear_cg_best_known_replay} compares tied,
bilinear-CG, and folded fully grouped hardware execution under the same BK-threshold
metric with fixed angles at \(p=3\) and \(p=5\). The \(p=5\) condition applies
five alternating phase-separator and mixer layers, two more than the \(p=3\)
condition. At this depth, fully grouped angles raise the pooled BK hit rate
\CreditTwentyQPFiveFullyGroupedBKRateFold{} from
\CreditTwentyQPFiveTiedBKRatePct{} to
\CreditTwentyQPFiveFullyGroupedBKRatePct{}, while retaining
\CreditTwentyQPFiveFullyGroupedFeasPct{} decoded exact-budget mass. At both
tested depths, the observed repeat-level BK-rate range for fully grouped
CG-XY-QAOA is disjoint from that of tied XY-QAOA.

\begin{table}[!htbp]
\centering
\small
\setlength{\tabcolsep}{5pt}
\renewcommand{\arraystretch}{1.1}
\caption[Matched 20-qubit CG-XY-QAOA hardware comparison]{\emph{Matched 20-qubit CG-XY-QAOA hardware comparison.} Fully grouped has the highest BK hit rate at both depths, while its mean exact-budget mass lies within the observed ranges of the tied and bilinear-CG conditions. IBM Heron R3 runs use the same Credit Card instance, sparse support, patch, exact-feasible basis initialization, and decoder at each depth. Each variant and depth uses ten predeclared 4096-shot repetitions. Exact-budget entries are repetition means [minimum, maximum] of the fraction of decoded shots satisfying both cardinality constraints. BK hit rates pool all submitted shots at each variant and depth and use the predeclared best-known threshold. Exact pooled hit counts and repeat-level BK rates are reported in Supplementary Section~\ref{supp:three_tier_cross_validation}.}
\label{tab:creditcard_bilinear_cg_best_known_replay}
\TableRowsSetup
\begin{tabular}{lrrrr}
\TableHideRows
\toprule
& \multicolumn{2}{c}{$p=3$} & \multicolumn{2}{c}{$p=5$} \\
\cmidrule(lr){2-3}\cmidrule(lr){4-5}
Variant & Exact-budget mass & BK hit rate & Exact-budget mass & BK hit rate \\
\midrule
\TableBodyRows
XY-QAOA & 66.6\% [65.7\%, 67.3\%] & 0.046\% & 65.1\% [64.3\%, 66.5\%] & 0.14\% \\
Bilinear-CG & 66.0\% [65.2\%, 67.1\%] & 0.054\% & 64.1\% [62.9\%, 65.3\%] & 0.26\% \\
Fully grouped & 66.4\% [64.4\%, 68.4\%] & 0.39\% & 64.7\% [62.8\%, 65.6\%] & 3.7\% \\
\bottomrule
\end{tabular}
\end{table}

\FloatBarrier
\needspace{8\baselineskip}

\paragraph{32--64-qubit sparse-cost hardware runs.}

The fixed-angle 32--64-qubit runs measure decoded exact-budget mass relative to
chance feasibility and best dense-rescored energy relative to matched
random-feasible and strict-feasible classical references on the same dense QUBO.
Across the five fixed-angle bilinear-CG conditions, decoded exact-budget mass
ranges from \(\CreditWideFixedChanceRatioMin{}\times\) to
\(\CreditWideFixedChanceRatioMax{}\times\) the uniform-bitstring chance
probability, including 64 decision qubits at \(p=2\) and 36 at \(p=3\). Because
these circuits use QUBO-independent
basis initialization and constant tied angles, they isolate feasibility
retention under fixed sparse-cost circuits.

Classical warm-start hardware runs test energy transfer from dense-objective
seeds under the same sparse-cost execution setting. They use the matched-control
protocols and tensor-network angle selection defined in
Section~\ref{sec:implementation} on the 36-, 52-, and 64-qubit QUBOs.

Energy gaps are reported in QUBO objective units, with lower values better. If
\(E_{\rm fixed}\), \(E_{\rm warm}\), and \(E_{\rm class}\) denote the best
dense-rescored fixed-angle, warm-start, and strict-feasible classical energies,
respectively, the closed-gap fraction is
\[
  g_{\rm closed}=
  \frac{E_{\rm fixed}-E_{\rm warm}}
       {E_{\rm fixed}-E_{\rm class}}.
\]

At 36, 52, and 64 qubits, the warm starts in
Table~\ref{tab:hardware_run_summary} close
\CreditWarmStartGapClosedRangePct{} of the dense-rescored energy gap between the
fixed-angle runs and the strict-feasible classical reference. Matched
random-feasible initialization controls do not close that gap. Strict-feasible
classical references remain the lower-energy target at every width.

The band-1 procedure greedily flips bits to restore \((k,m)\), choosing each
flip by its immediate
sparse-objective energy change. Supplementary
Table~\ref{tab:cgxy_fixed_angle_condition_details} lists the strict and
repaired best-energy gaps for the 32--64-qubit conditions. Repair of the 36-qubit
warm-start condition matches the sparse classical reference. Primary hardware
metrics are computed on strict decoded exact-budget samples before repair.

\begin{table}[!htbp]
\centering
\footnotesize
\setlength{\tabcolsep}{3.5pt}
\renewcommand{\arraystretch}{1.08}
\caption[Sparse-cost hardware scaling and controls]{\emph{Sparse-cost hardware scaling and controls.} The fixed-angle conditions retain exact-budget mass from \(161\times\) to \(5{,}238\times\) the uniform-bitstring chance probability, and warm starts reduce the dense-rescored gap at all three widths where tested. Fixed-angle IBM Heron R3 runs are summarized by width and depth. For repeated conditions, exact-budget entries are repetition means [minimum, maximum]. Single-run entries show the observed fraction without brackets. Exact/chance compares the reported exact-budget mass with uniform-bitstring chance feasibility. Dense and band-1 gaps are reported in the corresponding QUBO objective's energy units, with lower values better. Dense gap is the best decoded fixed-angle energy minus the strict-feasible dense classical reference. Warm gap closed is \((E_{\rm fixed}-E_{\rm warm})/(E_{\rm fixed}-E_{\rm class})\). Band-1 gap is the repaired fixed-angle energy minus the sparse classical reference on the executed sparse objective.}
\label{tab:hardware_run_summary}
\TableRowsSetup
\begin{tabular}{rrrrrrr}
\TableHideRows
\toprule
Qubits & $p$ & Exact-budget mass & Exact/chance & Dense gap & Warm gap closed & Band-1 gap \\
\midrule
\TableBodyRows
32 & 3 & 56.3\% & 161$\times$ & +6.3 & -- & +0.2 \\
36 & 3 & 47.4\% [47.2\%, 47.5\%] & 218$\times$ & +19.9 & 93.3\% & +1.8 \\
52 & 2 & 31.2\% [30.9\%, 31.5\%] & 876$\times$ & +103.8 & 82.7\% & +12.7 \\
56 & 2 & 42.2\% & 5{,}238$\times$ & +113.4 & -- & +3.2 \\
64 & 2 & 25.5\% [24.9\%, 25.9\%] & 1{,}875$\times$ & +84.9 & 78.1\% & +15.2 \\
\bottomrule
\end{tabular}
\end{table}

\FloatBarrier

\section{Discussion}
\label{sec:discussion}

\subsection{Main Findings}

The planted-noise analysis formalizes one mechanism for
joint-selection advantage: signal concentrated on a jointly selected
sample--feature block is preserved by joint scoring, while feature-first
aggregation accumulates noise over the candidate samples. Synthetic results
show corresponding gains when the task selects multiple samples and features
under feature-specific noise. The public-benchmark
gains are strongest in balanced candidate pools.

Grouping the diagonal phase separator allows more
expressive schedules than a single shared angle. Fully grouped cost-plus-mixer gives same-depth lift on the
noiseless planted-bicluster sweep by changing both diagonal phases and mixer
transport (Figure~\ref{fig:cgxy_noiseless_evidence}). Bilinear-CG keeps tied
mixer transport and frees only the sample--feature-component phase. It is the
lower-dimensional selected-angle schedule used for sparse hardware circuits.
Its hardware interpretation depends on the sample--feature couplings retained in
the sparse surrogate, as summarized in
Figure~\ref{fig:submitted_coupling_diagnostic} and
Table~\ref{tab:cgxy_fixed_angle_condition_details}.

The main hardware bottleneck is the resource and noise cost of
implementing diagonal phases that preserve low-energy structure on sparse
connectivity. The factorial implementation study separates target choice from
the complete Block~XY path. The fractional target lowers native two-qubit count
and part of the submitted depth, while fusion, late binding, and edge-colored
transport provide further depth and duration reductions. Block~XY preserves the
exact-budget sector in the ideal circuit. After transpilation and hardware
execution, routing overhead and device noise determine the decoded exact-budget
mass. The fully grouped
20-qubit conditions give the highest measured Credit Card BK hit rate at both
tested depths. The 32--64-qubit circuits retain decoded exact-budget mass above
chance, and warm starts reduce the dense-rescored gap to the strict-feasible
classical reference.

\subsection{Calibration Choices Shape the Objective}

Calibration is a coefficient-level design choice for the sample--feature weight
matrix $\matW$. Both the magnitude and
structure of calibration error matter. Feature-aligned drift is amplified by
feature-first aggregation, whereas sample-aligned drift is much less damaging under
the analyzed aggregation rules.

The calibration-robustness result supports a concrete calibration-evaluation
procedure. Before running the exact-budget optimizer, candidate calibration maps
can be compared on older reference windows by their feature-aligned error while
keeping the candidate pool, budgets, and selection method fixed. The Credit Card
calibration-map sensitivity analysis in Supplementary Figures~\ref{fig:real_creditcard_calibration_data_quality}
and~\ref{fig:real_creditcard_calibration_mechanism} illustrates this evaluation.
Among nine calibration maps, Gaussian-tail $z$ has the lowest feature-aligned
sup-norm and highest oldest-block joint F1 at all three candidate-pool sizes.
The broader implication is that the
joint formulation makes calibration robustness measurable and addressable
within the same constrained objective, while preserving the exact-budget
selection semantics.

\subsection{Conditions for Joint-Selection Advantage}

The feature-first failure mode is the early marginalization step. A sequential
baseline compresses each feature to a column-level score before knowing which
samples will survive the fixed top-$k$ review budget. If feature value is
concentrated on a case-dependent subset, a large column aggregate can reflect
irrelevant samples while a relevant feature can look weak after aggregation.
The joint objective delays that decision and rewards signal on the same
selected samples and features through the coupled term
$\sum_{i\in S,j\in F} W_{ij}$.

In the synthetic budget sweeps, the joint advantage is negligible at \(k=1\) and grows
once the stage-2 task selects multiple anomalies and multiple features under
feature-specific noise. In the calibration-stable mechanism analyses, the
feature-selection gain peaks near the intermediate feature-budget ratio
(Supplementary Table~\ref{tab:synthetic_mechanism_stable}). The advantage is
therefore largest when the stage-2 task is a coupled subset-selection problem
rather than a feature-ranking problem with one selected case.

\paragraph{Limits of the joint-selection advantage.}
The same parameter sweep also maps the boundary of the joint-selection
advantage. With zero feature noise, every feature is informative, the
feature-first rule already recovers the planted feature set, and the
joint-minus-sequential combined-F1 difference is exactly zero across all tested
$(N,D,k,m)$ settings. With a
single-anomaly budget ($k{=}1$), no feature can aggregate coordinated evidence
across multiple selected samples, and joint selection can fall below sequential
once feature noise is high. These boundary cases clarify what the joint objective adds. It exploits
case-dependent feature utility under multi-anomaly review. When every feature is
informative, or when the budget selects only one anomaly, the coupled term has
little additional structure to offer, and the simpler sequential rule is
sufficient and can be preferable.

\subsection{Limitations and Future Work}

The study isolates a stage-2 constrained-selection setting with fixed candidate
pools and budgets. The hardware experiments evaluate execution and sampling for
a constraint-preserving ansatz on an IBM Heron R3 processor.
On the 32--64-qubit dense QUBOs,
strict-feasible classical baselines still find lower energies.

On sparse connectivity, implementing the full dense $N\times D$ phase layer
would require routing that adds two-qubit operations and accumulated noise. The
hardware-aligned masks used for the 32--64-qubit runs retain
\CreditWideRetainedCrossMassRangePct{} of the dense cross-register $L_1$
coupling mass. At this retained mass, every exact sparse optimum differs from
diagonal-only selection, and dense rescoring attributes
\CreditWideSparseDenseGainRetainedRangePct{} of the dense-reference gain to the
sparse optimum. The hardware circuits implement this sparse phase separator.
Dense rescoring evaluates their measured selections on the original QUBO.

The formal recovery results have an analogous scope. The main planted theorem
covers the centered bilinear core, and the full-objective extension gives a
sufficient condition under bounded marginal terms. The primary method optimizes
a fixed \((k,m)\) sector. The Credit Card neighboring-budget analysis in
Supplementary Section~\ref{supp:creditcard_budget_diagnostics} measures the
tradeoff between retrospective capture and review capacity across nearby
sectors and does not define an automatic budget-selection rule.

Future work falls into five connected directions: sparse
surrogate design, deployment calibration, domain transfer, richer constraints,
and algorithmic scaling.
Prior QAOA work indicates that sparse phase operators should preserve the
low-energy structure of the dense objective. Retaining many large coefficients alone is
insufficient~\cite{liu2022sparsifiedphase}. For the present bipartite QUBO
setting, this suggests selecting and, when appropriate, reweighting phase
operators by validation on the dense objective. Pre-hardware checks on that
objective should include CVaR,
tail summaries~\cite{barkoutsos2020cvar}, threshold-hit probability,
low-energy overlap, and exact optimum agreement where enumeration is possible.
Exact execution of dense
non-native cost layers on sparse superconducting devices is itself
resource-intensive~\cite{harrigan2021nonplanar,sack2024nonplanar}.

For deployment, the calibration results make the weight matrix part of the
deployed system, not just an optimizer input.
The theory studies planted joint structure with sub-Gaussian noise,
Bernstein-tail residuals, and bounded entrywise perturbations. The
public-benchmark analyses identify feature-aligned drift as the harmful mode to
monitor. A deployment-style study should update reference windows, choose the
scoring map, sparsification rule, and hardware layout jointly, and track which
perturbation modes change the constrained optimum. Connecting stage-2 selection
to upstream candidate generation and analyst feedback would evaluate calibration
maintenance across the full anomaly-detection system.

Beyond anomaly-feature selection, the same optimization primitive should be
tested in other bipartite selection domains with problem-specific calibration
requirements. The weight matrix would encode domain-specific quantities such as
binding affinity, coverage probability, or service utility. A cross-domain study
should distinguish the transferable parts of the method from domain-specific
modeling. The
exact-budget feasible sector, Block~XY mixer, and sparse-surrogate metrics
transfer from the present construction. The bipartite
weight matrix must be redesigned and validated for each domain.

Richer models can add anomaly-structure constraints to the current $(k,m)$
exact-budget sector while preserving the exact-budget invariant enforced by the
mixer. Representative extensions include temporal or graph coherence and
structured feature constraints such as group sparsity or mutual exclusion.
Penalty encodings can represent some of these constraints. The quantum design
question is which constraints should instead be implemented by feasible mixers
that enforce the constraint directly through the feasible subspace and mixer
dynamics~\cite{hadfield2019quantum,wang2020xy,larose2022mixerphaser}.

The main open algorithmic question is how the constrained optimizer scales.
\qaoa{} for general \qubo{} is best understood as a
heuristic without general performance guarantees, and the potential for quantum
advantage in discrete optimization remains an open
question~\cite{abbas2024quantumoptimization}. The same constrained-selection
problem can also be studied with oracle-based quantum search or optimization
primitives, such as amplitude amplification or quantum minimum
finding~\cite{brassard2002amplitude,durr1996minimum}. In that
setting, the reversible comparison circuit, cost loading, thresholding, and data
access are part of the oracle cost. A complementary empirical direction is to
optimize angles at the larger decision-register widths, push CG-XY-\qaoa{} to
deeper circuits, and test mitigation methods that preserve decoded exact-budget mass,
with this constrained-selection problem as the scaling benchmark.

\section{Conclusion}
\label{sec:conclusion}

This work formulates exact-budget bipartite selection as a joint
\qubo{}/Ising objective for choosing anomalous samples and explanatory features
under fixed budgets. Calibrated residual weights define the dense objective,
while the exact-budget sector keeps feasibility explicit. The perturbation
analysis explains why treating those weights as estimates matters: in the
planted model, joint selection is less sensitive to entrywise coefficient error
than the feature-first rule analyzed here. The experiments show gains in balanced
public-benchmark pools and synthetic settings with multiple selected anomalies
and features.

The quantum contribution combines a coupling-grouped XY-\qaoa{}
parameterization with a hardware-aware realization using register-preserving
variable placement, sparse-surrogate phases, Block~XY fusion, fractional gates,
edge-colored scheduling, and layout-aware decoding. The complete implementation
path reduces submitted depth and target-based duration beyond exposing
fractional gates alone. Noiseless simulations show
same-depth gains from fully grouped schedules. At 20 decision qubits and
\(p=5\), fully grouped angles raise the measured BK hit rate
\CreditTwentyQPFiveFullyGroupedBKRateFold{} to
\CreditTwentyQPFiveFullyGroupedBKRatePct{}, with
\CreditTwentyQPFiveFullyGroupedFeasPct{} decoded exact-budget mass. Bilinear-CG
provides the lower-dimensional fixed-transport schedule used in the
32--64-qubit sparse circuits.
On IBM Heron R3, the
resulting sparse-surrogate circuits retain feasible-sector mass far above chance
at 64 decision qubits \((p=2)\) and 36 \((p=3)\). Measured samples are decoded
and rescored on the dense objective against strict-feasible classical baselines,
which remain the lower-energy target.

\paragraph{Acknowledgments.}
The author thanks IBM Quantum for access to IBM Quantum hardware used in the
experiments, VTT Technical Research Centre of Finland for access to VTT Q50
used in early exploratory phases of this work, Ilmo Salmenperä for valuable feedback and Vlad Stirbu
for support during manuscript preparation.

\paragraph{Data availability.}
The Credit Card fraud dataset and the public IBM IT-AML benchmark used in the experiments
are available from their cited public
releases~\cite{kaggle_creditcardfraud,altman2023realistic,kaggle_ibm_itaml}.
Additional data supporting the hardware experiments are
available from the author upon reasonable request.

\bibliographystyle{unsrtnat}
\bibliography{references}

\clearpage
\appendix
\setcounter{section}{0}
\setcounter{equation}{0}
\setcounter{figure}{0}
\setcounter{table}{0}
\setcounter{algorithm}{0}
\setcounter{theorem}{0}
\renewcommand{\thesection}{S\arabic{section}}
\renewcommand{\theequation}{S\arabic{equation}}
\renewcommand{\thefigure}{S\arabic{figure}}
\renewcommand{\thetable}{S\arabic{table}}
\renewcommand{\thealgorithm}{S\arabic{algorithm}}
\renewcommand{\thetheorem}{S\arabic{theorem}}
\renewcommand{\theHsection}{supp.\arabic{section}}
\renewcommand{\theHequation}{supp.\arabic{equation}}
\renewcommand{\theHfigure}{supp.\arabic{figure}}
\renewcommand{\theHtable}{supp.\arabic{table}}
\renewcommand{\theHalgorithm}{supp.\arabic{algorithm}}
\renewcommand{\theHtheorem}{supp.\arabic{theorem}}

\section*{Supplementary Material}
This supplement provides proofs, robustness analyses, and extended hardware results
for the constrained-selection formulation and Coupling-Grouped XY Quantum
Approximate Optimization Algorithm (CG-XY-QAOA).

\bigskip
\noindent\textbf{\large Contents}\par\medskip
\begingroup
\setlength{\parskip}{2pt}
\noindent\hyperref[supp:notation]{\ref*{supp:notation}.\quad Notation}\dotfill \pageref*{supp:notation}\par
\noindent\hyperref[supp:calibration_validity]{\ref*{supp:calibration_validity}.\quad Calibration Validity and Coefficient Bounds}\dotfill \pageref*{supp:calibration_validity}\par
\noindent\hyperref[supp:elementary_formulation_lemmas]{\ref*{supp:elementary_formulation_lemmas}.\quad Formulation and CG-XY-QAOA Properties}\dotfill \pageref*{supp:elementary_formulation_lemmas}\par
\noindent\hyperref[supp:noise-proof]{\ref*{supp:noise-proof}.\quad Joint-Recovery Theory and Calibration Robustness}\dotfill \pageref*{supp:noise-proof}\par
\noindent\hyperref[supp:repro]{\ref*{supp:repro}.\quad Experimental Protocols and Reproducibility}\dotfill \pageref*{supp:repro}\par
\noindent\hyperref[supp:swap_sa]{\ref*{supp:swap_sa}.\quad Classical and Penalty-Based Baselines}\dotfill \pageref*{supp:swap_sa}\par
\noindent\hyperref[supp:synthetic_sensitivity]{\ref*{supp:synthetic_sensitivity}.\quad Synthetic Sensitivity and Mechanism Analyses}\dotfill \pageref*{supp:synthetic_sensitivity}\par
\noindent\hyperref[supp:creditcard_checks]{\ref*{supp:creditcard_checks}.\quad Public-Benchmark Comparisons and Protocol Sensitivity}\dotfill \pageref*{supp:creditcard_checks}\par
\noindent\hyperref[supp:qaoa_depth]{\ref*{supp:qaoa_depth}.\quad CG-XY-QAOA Grouping and Sampling Results}\dotfill \pageref*{supp:qaoa_depth}\par
\noindent\hyperref[supp:hardware_detail]{\ref*{supp:hardware_detail}.\quad Hardware Implementation, Decoding, and IBM Heron Execution}\dotfill \pageref*{supp:hardware_detail}\par
\endgroup
\medskip

\newpage
\section{Notation}
\label{supp:notation}

The depth-frontier and mechanism sweeps use two perturbation modes:
\textsc{residual} for planted-mean residual scoring and \textsc{mixed} for a
mixed-rank-one residual. Together with $\mu$, the mode defines the
planted-bicluster operating point.

\begingroup
\footnotesize
\setlength{\tabcolsep}{3pt}
\renewcommand{\arraystretch}{0.96}
\TableRowsSetup
\begin{longtable}{>{\raggedright\arraybackslash}p{0.25\linewidth}
                  >{\raggedright\arraybackslash}p{0.68\linewidth}}
\TableHideRows
\toprule
Symbol & Meaning \\
\midrule
\endfirsthead
\TableHideRows
\toprule
Symbol & Meaning \\
\midrule
\endhead
\bottomrule
\endfoot
\TableGroupRow
\multicolumn{2}{l}{\emph{Dimensions, budgets, and selection sets}} \\
\TableBodyRows
$N$ & Number of samples (rows of $\matZ$). \\
$D$ & Number of features (columns of $\matZ$). \\
$k$ & Stage-2 anomaly budget, equal to the number of samples selected. \\
$m$ & Stage-2 feature budget, equal to the number of features selected. \\
$[N]$, $[D]$ & Index sets $\{1,\ldots,N\}$ and $\{1,\ldots,D\}$. \\
$S \subseteq [N]$, $F \subseteq [D]$ & Selected sample/feature index sets, with $|S|=k$ and $|F|=m$. \\
$\vecs,\vecf$ & Binary indicator vectors for $S$ and $F$. \\
$\Omega$ & Exact-budget feasible sector
  \(\{(\vecs,\vecf):|\vecs|=k,\ |\vecf|=m\}\). \\
$S^\star,F^\star$ & Ground-truth (planted) sample/feature sets used to compute synthetic F1. \\
\midrule
\TableHideRows
\TableGroupRow
\multicolumn{2}{l}{\emph{Scoring matrices and bounds}} \\
\TableBodyRows
$\matX \in \mathbb{R}^{N\times D}$ & Unprocessed observation matrix. \\
$\matZ \in \mathbb{R}^{N\times D}$ & Standardized residual/score matrix derived from $\matX$. \\
$\matW$ & Sample--feature score field. The empirical objective uses a calibrated
  nonnegative field in $[0,w_{\max}]^{N\times D}$. The planted theory uses a
  real-valued signal-plus-noise field. \\
$M$, $\widetilde W=M\odot W$ & Binary support mask and the resulting sparse surrogate weight field. \\
$E_\times$ & Executable cross-register position pairs in a selected hardware layout. \\
$(\pi_{\mathcal S},\pi_{\mathcal F})$ & Within-register sample and feature permutations used for variable placement. \\
$R_{\mathcal S\mathcal F}$, $\rho_{\mathcal S\mathcal F}$ & Retained cross-coupling $L_1$ mass and its ratio to the dense cross-coupling mass. \\
$g_{\rm dense}$ & Dense-reference gain retained by the exact sparse optimum after dense rescoring. \\
$\matW^\star$ & Ideal coefficient field before calibration perturbation. \\
$\widehat{\matW}$ & Observed coefficient field satisfying the generic bound $\|\widehat{\matW}-\matW^\star\|_\infty\le\eta$. Clipping defines one stress-test construction. \\
$\matW^\tau$ & Clipped score field used in the clipped-residual recovery corollary. \\
$w_{\max}$ & Per-entry weight cap from feature-wise calibration. \\
$\matQ$ & Quadratic Unconstrained Binary Optimization (QUBO) cost matrix,
  $E(x)=x^\top \matQ x$. \\
$a_i,b_j$ & Sample and feature marginal scores in the calibrated objective. \\
$\lambda$ & Positive weight on the bilinear sample--feature coupling term. \\
\midrule
\TableHideRows
\TableGroupRow
\multicolumn{2}{l}{\emph{Perturbation and calibration parameters}} \\
\TableBodyRows
$\tau$ & Clipping level in the clipped-residual recovery bounds. \\
$\eta$ & Calibration-perturbation radius for the generic bound $\|\widehat{\matW}-\matW^\star\|_\infty\le\eta$. In the uniform synthetic stress tests, $\Delta_{ij}\sim\mathrm{Unif}[-\eta,\eta]$. \\
$\rho$ & Tradeoff parameter in Eq.~\mainref{eq:objective}. \\
$\rho_{\mathrm{extra}}$ & Retained fraction of the noiseless $p\to2p$ score increment in the depth-efficiency comparison. \\
$\rho_\tau$ & Deterministic clipping-bias budget in the clipped-residual recovery corollary. \\
$\nu$ & Synthetic noise fraction, equal to the proportion of uninformative features in synthetic data and distinct from $\eta$. \\
\midrule
\TableHideRows
\TableGroupRow
\multicolumn{2}{l}{\emph{Planted-bicluster signal model}} \\
\TableBodyRows
$\mu$ & Planted-bicluster signal strength, defined as the entrywise amplitude of $\matW_{ij}$
        on the planted support $S^\star\times F^\star$ in
        Theorem~\ref{thm:joint-noise-superiority-supp} and the
        depth-frontier sweep. \\
$\sigma$ & Sub-Gaussian scale of the centered noise term
        $\varepsilon_{ij}$ in the planted-bicluster model. \\
$\mathcal M$ & Number of feasible sample--feature pairs, $\binom Nk\binom Dm$. \\
$\Gamma$ & Positive-semidefinite proxy covariance matrix in the correlated-residual extension. \\
\midrule
\TableHideRows
\TableGroupRow
\multicolumn{2}{l}{\emph{QAOA and hardware-execution parameters}} \\
\TableBodyRows
$p$ & Quantum Approximate Optimization Algorithm (QAOA) depth (number of
  cost+mixer layers). \\
$\bm{\beta}$, $\bm{\gamma}$ & QAOA mixer/cost angle parameter vectors under the
   layerwise convention. Layer-shared controls use global tied angles instead.
   Bilinear-CG uses a tied marginal angle and separate layerwise bilinear
   angles while keeping the mixer tied. The fully grouped schedule uses
   three per-component cost angles $\gamma_{\mathcal S,\ell},
   \gamma_{\mathcal F,\ell}, \gamma_{\mathcal{SF},\ell}$ and a layerwise
   mixer angle. \\
$H_{\mathrm{C}}$, $H_{\mathrm{mix}}$ & QAOA cost and mixer Hamiltonians. $H_{\mathrm{mix}}$ is the Block~XY mixer. \\
$\alpha$ & Normalized expected feasible-sector energy,
  $(\bar C_\Omega-\mathbb E[C])/(\bar C_\Omega-C^\star)$, with $\alpha=1$ at
  the exact feasible optimum and $\bar C_\Omega$ the uniform-feasible
  reference. \\
$E_{\mathrm{BK}}$ & Best-known (BK) energy threshold fixed before the evaluated run. \\
\midrule
\TableHideRows
\TableGroupRow
\multicolumn{2}{l}{\emph{Metrics}} \\
\TableBodyRows
$p_{\rm feas}$ & Decoded exact-budget mass, equal to the fraction of layout-decoded samples in $\Omega$. \\
$p_{\rm chance}$ & Uniform-bitstring feasibility probability
  \(\binom{N}{k}\binom{D}{m}/2^{N+D}\). \\
$p_{\rm BK}$ & Fixed-threshold hit rate \(P\{x\in\Omega,\ E(x)\le E_{\rm BK}\}\). \\
$z_{\rm rand}$ & \((\bar E_{\rm feas}-\mu_{\rm RF})/\sigma_{\rm RF}\), where
  \(\mu_{\rm RF}\) and \(\sigma_{\rm RF}\) are the matched random-feasible
  mean and per-shot standard deviation. Negative values favor hardware samples. \\
\(\mathrm{CVaR}_5(E)\) & Mean energy of the lowest-energy 5\% tail of a sampled distribution. \\
$\mathrm{F1@k}_{\text{samples}}$ & $2|\hat S\cap S^+|/(|\hat S|+|S^+|)$ for labeled benchmark data, reducing to $|\hat S\cap S^\star|/k$ in balanced or planted fixed-size settings. \\
$\mathrm{F1@m}_{\text{features}}$ & $|\hat F\cap F^\star|/m$ under fixed-size selection. \\
Combined F1 & $\tfrac{1}{2}(\mathrm{F1@k}_{\text{samples}}+\mathrm{F1@m}_{\text{features}})$ (arithmetic mean). \\
$\Delta\mathrm{F1}$ & Joint minus sequential F1 gap under identical inputs and budgets. \\
\end{longtable}
\endgroup

\section{Calibration Validity and Coefficient Bounds}
\label{supp:calibration_validity}

The finite-reference tail map, bounded weight construction, and conformal
ablation rely on the following calibration bounds.

\begin{proposition}[Finite-sample validity of empirical cumulative-distribution tail p-values]
\label{prop:pit-validity}
Let $\{Z_{ij}: i\in\mathcal{R}\}$ be iid draws from a continuous cumulative
distribution function (CDF) $F_j$, and
let $\widehat{F}_j$ be the empirical CDF on $\lvert\mathcal{R}\rvert = n$
reference samples. The scoring map that produces the $Z_{ij}$ values is assumed
fixed independently of this empirical CDF (ECDF) calibration sample. A
split-reference protocol with separate fit and calibration windows satisfies
this condition. For an independent test draw $Z\sim F_j$ define the plug-in
probability-integral-transform (PIT) variable
$\widehat{U}=\widehat{F}_j(Z)$.
Then for any $\varepsilon>0$,
\[
  \Pr\!\left(\sup_z \left\lvert \widehat{F}_j(z)-F_j(z)\right\rvert > \varepsilon\right)
  \le 2e^{-2n\varepsilon^2}\,,
\]
and, conditional on a reference sample satisfying
$\sup_z \lvert \widehat{F}_j(z)-F_j(z)\rvert \le \varepsilon$, we have
$\max(0,t-\varepsilon) \le \Pr(\widehat{U}\le t)\le \min(1,t+\varepsilon)$ for all $t\in[0,1]$.
In particular, with probability at least $1-\delta$ over the reference sample,
the Kolmogorov distance between $\widehat{U}$ and $\mathrm{Unif}(0,1)$ is at most
$\varepsilon=\sqrt{\frac{1}{2n}\log\frac{2}{\delta}}$.
The smoothed clipped rank estimate used in the experiments differs from
$\widehat F_j$ by at most $1/(n+1)$ uniformly, so the same bound holds with
$\varepsilon$ replaced by $\varepsilon+1/(n+1)$.
\end{proposition}

This is the Dvoretzky--Kiefer--Wolfowitz bound in Massart's tight
form~\cite{massart1990tight} combined with the monotone PIT sandwiching
argument. Rank smoothing changes the empirical CDF by at most \(1/(n+1)\).
For the two-sided map \(g(U)=2\min(U,1-U)\),
\(|g(u)-g(v)|\le 2|u-v|\). The plug-in two-sided tail probability therefore
differs from its ideal value by at most
\(2\{\varepsilon+1/(n+1)\}\), which also bounds its Kolmogorov distance from
the uniform distribution.

\begin{lemma}[Clipping bounds Hamiltonian coefficients and gate angles]
\label{lem:bounded-weights}
Suppose $0\le W_{ij}\le w_{\max}$. Under the convention $x=(I-Z)/2$, the
bilinear term $-\lambda W_{ij}s_if_j$ contributes
$J_{ij}=-\lambda W_{ij}/4$ to the Pauli \(ZZ\) coefficient and therefore has magnitude
$\lvert J_{ij}\rvert \le \lambda w_{\max}/4$. Consequently, each cost-layer interaction
$e^{-i\gamma J_{ij}Z_iZ_j}$ can be implemented with an $\mathrm{RZZ}(\theta_{ij})$
rotation with $\lvert \theta_{ij}\rvert \le 2\lvert\gamma J_{ij}\rvert \le
\lvert\gamma\rvert\lambda w_{\max}/2$. Thus $w_{\max}$ directly controls coefficient
dynamic range and the required phase-angle range for transpilation.
\end{lemma}

\begin{proposition}[Split conformal p-values are super-uniform]
\label{prop:conformal-validity}
Suppose the nonconformity score is fixed independently of the calibration and
test observations, for example by fitting it on a disjoint subset. Let the
observations in the calibration window $\mathcal{R}$ and the additional test
draw $Z$ be jointly exchangeable. The usual split-conformal rank \(p\)-value,
computed by comparing the test score with the calibration scores and using
conservative tie handling, satisfies
\[
  \Pr(p \le \alpha)\le \alpha \qquad \text{for all } \alpha\in[0,1]\,,
\]
independently of the underlying distribution.
\end{proposition}

The conformal ablation is therefore a distribution-free alternative to
ECDF/PIT~\cite{vovk2022algorithmic,lei2018distributionfree}. The implemented
conformal scorer splits the reference window between fitting its MAD-based
nonconformity map and calibrating ranks. Conditional on a fixed reference
window, the random split introduces variability absent from the ECDF/PIT
map~\cite{lei2018distributionfree}.

\section{Formulation and CG-XY-QAOA Properties}
\label{supp:elementary_formulation_lemmas}

\paragraph{Rank-one coupling separability check.}
For the core bilinear objective
\[
  B(S,F)=\sum_{i\in S}\sum_{j\in F} W_{ij},
  \qquad W=\bm u\bm v^\top,\quad \bm u,\bm v\ge 0,
\]
we have
\[
  B(S,F)=\left(\sum_{i\in S}u_i\right)
         \left(\sum_{j\in F}v_j\right).
\]
Both factors are nonnegative, so a maximizer is obtained by independently
choosing the $k$ largest entries of $\bm u$ and the $m$ largest entries of
$\bm v$. Ties or a zero factor may produce other equivalent maximizers. If
$W=\bm u\bm v^\top+\bm E$, then for any feasible $(S,F)$,
\[
  \left|B_W(S,F)-B_{\bm u\bm v^\top}(S,F)\right|
  =
  \left|\sum_{i\in S}\sum_{j\in F}E_{ij}\right|
  \le km\|\bm E\|_\infty.
\]
Let $(S_0,F_0)$ be a rank-one maximizer and let $(\widehat S,\widehat F)$ be a
maximizer for $B_W$. Put $\epsilon=km\|\bm E\|_\infty$. The pointwise bound gives
\[
  B_W(\widehat S,\widehat F)
  \le B_{\bm u\bm v^\top}(\widehat S,\widehat F)+\epsilon
  \le B_{\bm u\bm v^\top}(S_0,F_0)+\epsilon
  \le B_W(S_0,F_0)+2\epsilon .
\]
Therefore the perturbed optimum can improve over the rank-one independent
solution by at most $2km\|\bm E\|_\infty$ under the perturbed objective. The
factor of two arises because the pointwise perturbation bound is applied
independently at the two feasible sets being compared.
This proves Proposition~\mainref{prop:rank-one}.

\paragraph{Block~XY feasibility.}
Let
\[
  \hat N_{\mathcal S}=\sum_{i\in\mathcal S}\frac{I-Z_i}{2},
  \qquad
  \hat N_{\mathcal F}=\sum_{j\in\mathcal F}\frac{I-Z_j}{2}
\]
be the sample and feature excitation-number operators. Each pairwise XY term
\(H_{\mathrm{XY}}^{(i,j)}\) can be written as
\[
  H_{\mathrm{XY}}^{(i,j)}
  =
  \tfrac12(X_iX_j+Y_iY_j)=\sigma_i^+\sigma_j^-+\sigma_i^-\sigma_j^+,
\]
where \(\sigma^\pm=(X\pm iY)/2\). This interaction moves one excitation
between the two qubits and therefore commutes with
the total excitation number of that register. Because the Block~XY mixer contains
only sample--sample and feature--feature terms,
\[
  [H_{\mathrm{mix}},\hat N_{\mathcal S}]
  =
  [H_{\mathrm{mix}},\hat N_{\mathcal F}]
  =0.
\]
The cost Hamiltonian is diagonal in the computational basis, so it also
preserves both Hamming weights. Therefore every cost--mixer layer maps the
simultaneous eigenspace
$(\hat N_{\mathcal S},\hat N_{\mathcal F})=(k,m)$ to itself. The conclusion
holds for tensor Dicke initialization, exact-feasible computational-basis
initialization, and any other initial state supported on the feasible subspace.
This proves Proposition~\mainref{prop:feasibility}.

\paragraph{Computational complexity.}
\label{supp:complexity}

The hardness proof reduces from Balanced Complete Bipartite Subgraph, problem
GT24 in Garey and Johnson~\cite{garey1979computers}.

\begin{proposition}[Complexity of constrained bipartite selection]
\label{prop:cbs_complexity}
The decision form of constrained bipartite selection is \scaps{NP}-complete.
Consequently, the optimization form is \scaps{NP}-hard.
\end{proposition}

\begin{proof}
Membership in \scaps{NP} is immediate. For hardness, reduce from Balanced
Complete Bipartite Subgraph. Given \(G=(U,V,E)\) and target size \(K\), set
\(N=|U|\), \(D=|V|\), \(k=m=K\), \(a=b=0\), \(\rho=0\), and
\(W_{ij}=\mathbf{1}\{(i,j)\in E\}\). Then
\[
  B(S,F)=\sum_{i\in S}\sum_{j\in F} W_{ij}=|E(S,F)|.
\]
Thresholding the constructed instance at \(K^2\) is therefore equivalent to
asking whether \(G\) contains a complete \(K\)-by-\(K\) bipartite subgraph.
\end{proof}

\paragraph{Coupling-grouped phase separator.}
\label{supp:qaoa_grouped_multiangle}
The diagonal cost Hamiltonian separates into sample-marginal, feature-marginal,
and bilinear coupling components,
\[
  H_{\mathrm{C}} = H_{\mathcal S}+H_{\mathcal F}+H_{\mathcal S\mathcal F}.
\]
For the binary objective in Eq.~\mainref{eq:hamiltonian}, these components
are the separate Ising images of the sample marginal, feature marginal, and
sample--feature product terms:
\[
  H_{\mathcal S}=\frac{1}{2}\sum_i a_iZ_i,\qquad
  H_{\mathcal F}=\frac{1}{2}\sum_j b_jZ_{N+j},\qquad
  H_{\mathcal S\mathcal F}=\frac{\lambda}{4}\sum_{i,j}W_{ij}
  \left(Z_i+Z_{N+j}-Z_iZ_{N+j}\right).
\]
Thus the sample--feature component contains the one-local shifts induced by
the binary products as well as their \(ZZ\) interactions. This definition
matches the component costs used in the grouped simulations.
Bilinear-CG-XY-QAOA uses this decomposition inside the Block~XY ansatz
by replacing the tied phase separator with \(U_{\mathrm C}^{\mathrm{bilCG}}\),
defined by
\[
  U_{\mathrm C}^{\mathrm{bilCG}}(\gamma_{\mathrm{marg}},
  \gamma_{\mathcal S\mathcal F})
  =
  \exp\!\left[-i\left(
    \gamma_{\mathrm{marg}}(H_{\mathcal S}+H_{\mathcal F})
    + \gamma_{\mathcal S\mathcal F} H_{\mathcal S\mathcal F}
  \right)\right].
\]
The fixed-transport grouped separator \(U_{\mathrm C}^{\mathrm{3grp}}\) is
\[
  U_{\mathrm C}^{\mathrm{3grp}}(\gamma_{\mathcal S},\gamma_{\mathcal F},
  \gamma_{\mathcal S\mathcal F})
  =
  \exp\!\left[-i\left(
    \gamma_{\mathcal S} H_{\mathcal S}
    + \gamma_{\mathcal F} H_{\mathcal F}
    + \gamma_{\mathcal S\mathcal F} H_{\mathcal S\mathcal F}
  \right)\right].
\]
Setting all phase angles equal recovers the layerwise XY-\qaoa{} phase
separator exactly. Section~\ref{supp:qaoa_grouped_multiangle_details} gives the
containment, local-response, and depth-efficiency statements together with the
depth and sampling results.
In applications the lower-dimensional bilinear-CG variant is a targeted refinement
of tied XY-QAOA. The extra bilinear phase can improve the ansatz when the
sample--feature coupling mass is non-negligible and the local bilinear phase
response is distinct from the marginal-field response. The fully grouped variant uses
the same CG-XY-QAOA decomposition and additionally gives layerwise control of
transport on the feasible swap graph for the same fixed-threshold sampling
metric.

\subsection{CG-XY-QAOA Formal Properties}
\label{supp:qaoa_grouped_multiangle_details}

The building blocks are multi-angle \qaoa{}~\cite{herrman2022multiangle} and
the feasible-subspace XY mixer~\cite{wang2020xy}. These building blocks define
the finite-dimensional CG-XY-QAOA phase-separator hierarchy for the bipartite
sample--feature cost decomposition.

\begin{proposition}[Grouped phase containment, variational dominance, and small-angle response]
\label{prop:grouped-phase-local-response}
Let $\Omega$ be a finite feasible basis, let $C:\Omega\to\mathbb{R}$ be a
diagonal cost, and let $M$ be a real symmetric mixer Hamiltonian on $\Omega$.
The grouped phase separator contains the tied separator exactly. Setting
$\gamma_{\mathcal S}=\gamma_{\mathcal F}=\gamma_{\mathcal S\mathcal F}=\gamma$
gives
\[
  \gamma H_{\mathrm C}
  =
  \gamma H_{\mathcal S}+\gamma H_{\mathcal F}
  +\gamma H_{\mathcal S\mathcal F}.
\]
Consequently, at any fixed depth $p$, the fixed-transport grouped
CG-XY-QAOA parameter set contains the tied-cost parameter subset obtained with
the same shared mixer-transport schedule. The fully grouped cost-plus-mixer
parameterization contains layerwise XY-\qaoa{} as the tied-angle subset of its
independent mixer angles.
Therefore the globally optimized grouped expected cost
is no larger than the globally optimized tied expected cost. Equivalently, for
any normalized minimization score
$\alpha(\psi)=(\bar C-\langle C\rangle_\psi)/(\bar C-C^\star)$ with
$\bar C>C^\star$, the best grouped score at depth $p$ is at least the best tied
score at depth $p$, relative to this tied parameter subset.

For the local response, write
\(C=C_{\mathcal S}+C_{\mathcal F}+C_{\mathcal S\mathcal F}\) and let
\(D_C=\mathrm{diag}(C)\) and \(D_a=\mathrm{diag}(C_a)\). For
\[
  E(\beta,\gamma_a)=
  \langle u|e^{i\gamma_a D_a}e^{i\beta M}D_Ce^{-i\beta M}
  e^{-i\gamma_a D_a}|u\rangle,
  \qquad
  |u\rangle=|\Omega|^{-1/2}\sum_{x\in\Omega}|x\rangle,
\]
the componentwise mixed response \(r_a\) is
\[
  r_a :=
  \left.
  \frac{\partial^2 E}{\partial\beta\,\partial\gamma_a}
  \right|_{\beta=\gamma_a=0}
  =
  \frac{2}{|\Omega|}
  \sum_{x<y}M_{xy}\bigl(C_x-C_y\bigr)\bigl(C_a(x)-C_a(y)\bigr).
  \label{eq:grouped_component_response}
\]
In the tied direction this reduces to the total-cost Dirichlet form. In the
constrained bipartite application,
\[
  |u\rangle=|D_k^N\rangle\otimes |D_m^D\rangle,
\]
because the tensor Dicke state is the uniform superposition over all
\((S,F)\) with \(|S|=k\) and \(|F|=m\). Thus
Eq.~\eqref{eq:grouped_component_response} gives the mixed derivative at the
origin for the Dicke-initialized feasible-sector ansatz with fixed mixer \(M\).
Optimized finite-angle and multi-layer behavior is evaluated in the
depth-frontier simulations and hardware experiments.
\end{proposition}

\begin{proof}
The containment statement follows by substituting
$\gamma_{\mathcal S}=\gamma_{\mathcal F}=\gamma_{\mathcal S\mathcal F}=\gamma$
and collecting the three commuting diagonal cost components. Applying the same
substitution independently in every layer maps any tied-cost, shared-mixer
depth-$p$ parameter vector to a fixed-transport grouped depth-$p$ parameter
vector that produces the same circuit state. Thus the tied-cost shared-mixer
parameter set is a subset of the fixed-transport grouped parameter set, and
minimizing the same expected cost over the larger
set cannot increase the optimum value. The normalized-score statement follows
because $\alpha$ is an affine decreasing function of $\langle C\rangle_\psi$
when $\bar C>C^\star$.

For the small-angle identity, define \(G=\gamma_aD_a\). Differentiating
\[
  \langle u|e^{iG}e^{i\beta M}D_Ce^{-i\beta M}e^{-iG}|u\rangle
\]
first in \(\beta\) and then in \(\gamma_a\) at zero gives
\(-\langle u|[D_a,[M,D_C]]|u\rangle\). Since
\[
  [D_a,[M,D_C]]_{xy}
  =
  -\bigl(C_a(x)-C_a(y)\bigr)\bigl(C_x-C_y\bigr)M_{xy},
\]
and \(M\) is symmetric, the uniform-state expectation becomes
Eq.~\eqref{eq:grouped_component_response}.
\end{proof}

\begin{proposition}[Fully grouped metric containment and transport freedom]
\label{prop:fully-grouped-endpoint-containment}
The fully grouped cost-plus-mixer parameterization contains the tied,
bilinear-CG, and fixed-transport grouped parameter subsets at the same depth.
Consequently, for any
fixed larger-is-better metric \(Q(\psi)\) that depends only on the final
state or its decoded distribution, including normalized score, exact-optimum
probability, or best-known-threshold (BK) hit rate, the fully grouped search
space contains all three alternatives. Optimizing over that parameterization
therefore gives a maximum at least as large as optimizing over any of these
parameter subsets.
\end{proposition}

\begin{proof}
The tied, bilinear-CG, and fixed-transport schedules are recovered from the
fully grouped schedule by imposing equality constraints on the component phases
or mixer angles. Optimizing the same larger-is-better metric over the resulting
parameter superset cannot decrease its optimum.
\end{proof}

Proposition~\ref{prop:fully-grouped-endpoint-containment} is separate from the
small-angle phase-response identity in Eq.~\eqref{eq:grouped_component_response}.
The mixed-derivative formula
isolates the benefit of diagonal component phases for a fixed mixer generator.
Freeing \(\beta_\ell\) changes transport on the feasible swap graph after each
phase kick, following the alternating-operator view of constrained mixers and
phasers~\cite{hadfield2019quantum,larose2022mixerphaser}.

The diagonal cost layers fail to commute with the XY mixer. Different layerwise
mixer angles therefore define distinct transport steps and cannot in general be
collapsed into one shared transport angle. Let
\(\Pi_{\mathrm{BK}}=\sum_{x\in\Omega:E(x)\le E_{\mathrm{BK}}}|x\rangle\langle x|\)
be the projector onto feasible states at or below the fixed best-known
threshold. The metric gradient with respect
to a mixer angle in the ideal continuous-time mixer model has the commutator
form
\[
  \partial_{\beta_\ell}\langle\Pi_{\mathrm{BK}}\rangle
  =
  i\langle[M,U_{>\ell}^{\dagger}\Pi_{\mathrm{BK}}U_{>\ell}]\rangle_{\ell,+},
\]
where \(U_{>\ell}\) denotes the layers after mixer \(\ell\), and the
expectation is taken at the state immediately after that mixer. For the ordered
product mixer used in circuits, the same expression holds with \(M\) replaced
by the corresponding product-layer derivative generator. Different diagonal
component phases change the state entering this commutator, and different
\(\beta_\ell\) values control how much probability current is applied on the
feasible swap graph after each phase kick. Fully grouped optimization therefore
has metric directions that are absent from fixed-transport CG. These
directions can redistribute probability across the threshold cut. This explains
why BK hit rate can improve even when mean objective, CVaR, and exact-optimum
mass rank the same schedules differently.

\begin{proof}[Proof of Corollary~\mainref{cor:grouped-equal-norm-descent}]
Let \(r=(r_{\mathcal S},r_{\mathcal F},r_{\mathcal S\mathcal F})\) denote the
three component responses in Eq.~\eqref{eq:grouped_component_response}, and
fix \(\beta=\epsilon\). The mixed second-order term is
\(\epsilon\,\gamma^{\top}r\), so descent for minimization uses phase
directions with \(\gamma^{\top}r<0\). The best signed tied vector
\((\pm\epsilon,\pm\epsilon,\pm\epsilon)\) has norm \(\sqrt{3}\epsilon\) and
predicted descent \(\epsilon^2|\mathbf{1}^{\top}r|\). Over the Euclidean
sphere \(\|\gamma\|_2=\sqrt{3}\epsilon\), Cauchy's inequality gives the most
negative inner product at
\(\gamma_{\mathrm{grp}}=-\sqrt{3}\epsilon\,r/\|r\|_2\), with predicted descent
\(\epsilon^2\sqrt{3}\|r\|_2\). Subtracting the tied descent gives
\(\epsilon^2\sqrt{3}\|r\|_2\bigl(1-|\cos\angle(r,\mathbf{1})|\bigr)\), and
equality holds exactly when \(r\) is collinear with \(\mathbf{1}=(1,1,1)\),
up to sign.
\end{proof}

The component formula is the finite-dimensional instance of the standard
small-angle alternating-operator response expansion~\cite{hadfield2022analytical},
specialized to the calibrated cost decomposition used here. The sample field,
feature field, and bilinear signal components can have different
cross-Dirichlet responses against the total cost on the same feasible mixer
graph, giving a fixed-layer rationale for the coupling-grouped phase split that
is used in CG-XY-QAOA.

\paragraph{Depth-efficiency comparison.}
Let $A_p$ be the optimized tied-angle score at depth $p$, let $G_p$ be the
optimized grouped CG-XY-QAOA score under consideration at depth $p$, and let $A_{2p}$ be the
optimized tied-angle score at depth $2p$, all measured by the same
larger-is-better normalized score. Assume $A_{2p}>A_p$ and define
\[
  r_p=\frac{G_p-A_p}{A_{2p}-A_p}.
\]
Suppose a finite-hardware implementation of the added $p$ full tied-angle layers
retains a fraction $\rho_{\mathrm{extra}}\in[0,1]$ of the noiseless $p\to2p$
score increment. Its effective score is then
\(A_p+\rho_{\mathrm{extra}}(A_{2p}-A_p)\), and grouped depth $p$ is at least as
good as this extra-depth tied circuit exactly when
\(r_p\ge\rho_{\mathrm{extra}}\).

\paragraph{Trainability considerations for the constraint-preserving ansatz.}
\label{supp:trainability}
The tied XY-QAOA ansatz and the CG-XY-QAOA variants act inside the
$(k,m)$ feasible sector of dimension $\binom{N}{k}\binom{D}{m}$. Block~XY
commutes with the per-register number operators, so the reachable dynamical Lie
algebra is restricted by both number symmetries. The resulting symmetry
restriction eliminates infeasible leakage without guaranteeing trainability.
The feasible sector remains
exponentially large when the budgets scale with register size, and optimization
still depends on depth, locality, cost structure, and parameterization
\cite{mcclean2018barren,larocca2024review}.

The circuits studied here are shallow and structured: hardware results reach
$p=5$, and the noiseless depth sweep reaches $p=8$. Their construction differs
from the deep random-circuit setting used in one standard approximate-2-design
barren-plateau argument~\cite{cerezo2021costfunction}. The dense
sample--feature cost still couples the registers, so the optimization traces and
depth sweep evaluate the resulting landscape empirically rather than infer
trainability from symmetry or depth alone.

The calibration cap
$W_{ij}=\mathrm{clip}(-\log p_{ij},0,w_{\max})$ bounds the spectral scale of the
bilinear cost on the feasible sector. Together with bounded marginal
coefficients, it also bounds the corresponding cost-angle sensitivity. This is
a coefficient-scale bound, not an optimization guarantee. Over the tested
shallow-depth range, the expected-energy score $\alpha$ increases and the
same-depth fully grouped lift remains positive at displayed precision.

\section{Joint-Recovery Theory and Calibration Robustness}
\label{supp:noise-proof}

The finite-sample recovery analysis combines a planted-noise model with
deterministic calibration-perturbation bounds. The
bilinear-core recovery theorem is a scan-statistic adaptation of elevated
submatrix localization and biclustering
\cite{ariascastro2011anomalouscluster,kolar2011minimax,butucea2015sharp,
cai2017submatrix,hajek2018submatrix}. The additional results track how exact
budgets, feature-first aggregation, and bounded calibration error change the
certificate.
Throughout this section, the recovery statements assume nontrivial feasible
classes and confidence levels, so the displayed logarithms are evaluated with
\(\mathcal M>1\), \(m(D-m)>0\), and \(0<\delta<1\).

\subsection{Recovery Guarantees}

\begin{theorem}[Joint recovery under a planted bicluster with sub-Gaussian noise]
\label{thm:joint-noise-superiority-supp}
Assume planted sets $S^\star\subseteq[N]$, $F^\star\subseteq[D]$ with
$|S^\star|=k$, $|F^\star|=m$, and
\[
W_{ij}=\mu\,\mathbf{1}\{i\in S^\star,\;j\in F^\star\}+\varepsilon_{ij},
\]
where $\{\varepsilon_{ij}\}$ are independent, mean-zero, $\sigma$-sub-Gaussian.
Let
\[
B(S,F)=\sum_{i\in S}\sum_{j\in F}W_{ij},\qquad
\mathcal{M}=\binom{N}{k}\binom{D}{m}.
\]
If
\[
\mu \;\ge\; 2\sigma\sqrt{\frac{\log\!\left((\mathcal{M}-1)/\delta\right)}{\min(k,m)}},
\]
then with probability at least $1-\delta$, $(S^\star,F^\star)$ is the unique
maximizer of $B(S,F)$ over $|S|=k$, $|F|=m$.
\end{theorem}

\begin{proof}
Fix a competitor \((S,F)\neq(S^\star,F^\star)\) and set
\(d=km-|S\cap S^\star|\,|F\cap F^\star|\). Exact cardinalities imply
\(d\ge d_{\min}:=\min(k,m)\). The planted signal gap is \(\mu d\). After
cancelling overlap between \(S^\star\times F^\star\) and \(S\times F\), the
noise gap is a signed sum over at most \(2d\) independent
\(\sigma\)-sub-Gaussian entries, hence is \(\sigma\sqrt{2d}\)-sub-Gaussian.
For this fixed competitor,
\[
  \Pr\!\left(B(S,F)\ge B(S^\star,F^\star)\right)
  \le \exp\!\left(-\frac{\mu^2d}{4\sigma^2}\right)
  \le \exp\!\left(-\frac{\mu^2d_{\min}}{4\sigma^2}\right).
\]
A union bound over the \(\mathcal M-1\) competitors gives the stated
finite-sample sufficient condition.
\end{proof}

\begin{proposition}[Covariance-aware correlated-residual extension]
\label{prop:covariance-aware-joint-recovery}
Replace the independent entrywise residual assumption in
Theorem~\ref{thm:joint-noise-superiority-supp} by the following centered
sub-Gaussian vector MGF condition~\cite{hsu2012tail}. Let
\(\varepsilon=(\varepsilon_{ij})_{i\in[N],j\in[D]}\) satisfy
\[
  \mathbb E\exp\!\left(\sum_{i,j}u_{ij}\varepsilon_{ij}\right)
  \le
  \exp\!\left(\frac{1}{2}u^\top\Gamma u\right)
  \qquad\text{for all }u\in\mathbb R^{ND},
\]
where \(\Gamma\succeq0\) is a proxy covariance matrix. For a competitor
\((S,F)\neq(S^\star,F^\star)\), define
\[
  c_{S,F,ij}
  =
  \mathbf 1\{(i,j)\in S^\star\times F^\star\}
  -
  \mathbf 1\{(i,j)\in S\times F\},
  \qquad
  d(S,F)=km-|S\cap S^\star|\,|F\cap F^\star|.
\]
Let
\[
  K_\Gamma
  =
  \max_{(S,F)\neq(S^\star,F^\star)}
  \frac{\sqrt{c_{S,F}^{\top}\Gamma c_{S,F}}}{d(S,F)}.
\]
If
\[
  \mu
  \ge
  \sqrt{2\log((\mathcal M-1)/\delta)}\,K_\Gamma,
\]
then with probability at least \(1-\delta\), \((S^\star,F^\star)\) is the
unique maximizer of \(B(S,F)\) over \(|S|=k, |F|=m\).
\end{proposition}

\begin{proof}
For a competitor, write \(c=c_{S,F}\) and \(d=d(S,F)\). The planted signal gap
is \(\mu d\), and the noise gap is \(c^\top\varepsilon\). The vector MGF
condition gives the standard Chernoff bound
\[
  \Pr(-c^\top\varepsilon\ge \mu d)
  \le
  \exp\!\left(-\frac{\mu^2d^2}{2c^\top\Gamma c}\right),
\]
with the usual zero-variance convention~\cite{hsu2012tail}. The definition of
\(K_\Gamma\) makes this probability at most \(\delta/(\mathcal M-1)\) for
every competitor. A union bound proves the claim.
\end{proof}

This proposition isolates the role of dependence. Independence is not needed
for the deterministic planted gap. It enters only through the random contrast
scale \(\sqrt{c_{S,F}^{\top}\Gamma c_{S,F}}\). Under the independent model
\(\Gamma=\sigma^2 I\), the contrast vector has at most \(2d(S,F)\) nonzero
entries and the proposition recovers the scale used in
Theorem~\ref{thm:joint-noise-superiority-supp}. If every coefficient is also
perturbed by at most \(\eta\), the deterministic calibration argument in
Proposition~\ref{prop:calibration-perturbation-supp} replaces \(\mu\) by
\(\mu-2\eta\) in the covariance-aware condition.

For intuition, suppose a feature-swap competitor keeps \(S=S^\star\) and
replaces \(f\) planted features by \(f\) non-planted features. Then
\(d=fk\). If different features are independent and residuals within each
feature have equicorrelation
\[
  \operatorname{Cov}(\varepsilon_{ij},\varepsilon_{i'j})
  =
  \sigma^2\bigl((1-r)\mathbf 1\{i=i'\}+r\bigr),
  \qquad 0\le r\le1,
\]
then
\[
  c^\top\Gamma c
  =
  2\sigma^2\{(1-r)d+r f k^2\}.
\]
For the one-feature-swap case \(f=1\), the normalized variance factor is
\[
  \frac{(1-r)k+rk^2}{k^2}
  =
  r+\frac{1-r}{k}.
\]
Positive within-feature correlation therefore weakens the \(1/k\) averaging
available under independent residuals and reaches no averaging at \(r=1\).

\begin{corollary}[Bounded-marginal extension for the full score]
\label{cor:bounded-marginal-full-score}
Under the assumptions of Theorem~\ref{thm:joint-noise-superiority-supp}, define
the full maximization score
\[
  Q(S,F)
  =
  \sum_{i\in S}a_i+\sum_{j\in F}b_j+\lambda B(S,F),
  \qquad \lambda>0.
\]
Assume \(|a_i|\le A_{\max}\) and \(|b_j|\le B_{\max}\). If
\[
  \lambda\mu-\frac{2A_{\max}}{m}-\frac{2B_{\max}}{k}
  \ge
  2\lambda\sigma
  \sqrt{\frac{\log((\mathcal M-1)/\delta)}{\min(k,m)}},
\]
then with probability at least \(1-\delta\), \((S^\star,F^\star)\) is the
unique maximizer of \(Q(S,F)\) over \(|S|=k, |F|=m\).
\end{corollary}

\begin{proof}
Fix a competitor and write \(s=|S^\star\setminus S|\),
\(f=|F^\star\setminus F|\), and
\(d=km-|S\cap S^\star|\,|F\cap F^\star|\). The marginal terms can favor the
competitor by at most \(2A_{\max}s+2B_{\max}f\). Exact cardinalities give \(s\le d/m\) and
\(f\le d/k\), so the expected full-score gap is at least
\[
  d\left(\lambda\mu-\frac{2A_{\max}}{m}-\frac{2B_{\max}}{k}\right).
\]
The random part is the same bilinear contrast multiplied by \(\lambda\).
Applying the theorem's fixed-competitor tail bound and union bound gives the
condition.
\end{proof}

\begin{lemma}[Entrywise Bernstein-MGF residuals imply signed-sum tails]
\label{lem:entrywise-bernstein-signed-sums}
Let \(\varepsilon_1,\ldots,\varepsilon_r\) be independent, centered residuals
satisfying the Bernstein moment-generating-function (MGF) condition
\[
  \mathbb E\exp(\lambda\varepsilon_\ell)
  \le
  \exp\!\left(\frac{\lambda^2 v^2}{2}\right),
  \qquad |\lambda|<1/b,
\]
with common variance proxy \(v^2\) and scale \(b>0\). Then for any signs
\(a_\ell\in\{-1,+1\}\) and \(X=\sum_{\ell=1}^r a_\ell\varepsilon_\ell\),
\[
  \Pr(X\le -t)
  \le
  \exp\!\left[
    -\frac{1}{2}
    \min\!\left(
      \frac{t^2}{r v^2},\frac{t}{b}
    \right)
  \right].
\]
\end{lemma}

\begin{proof}
Signs preserve the MGF condition, and independence gives
\(\mathbb E\exp(\lambda X)\le\exp(\lambda^2rv^2/2)\) for
\(|\lambda|<1/b\). The displayed one-sided tail bound is the standard Chernoff
optimization for this Bernstein MGF convention~\cite[Sec.~2.8]{boucheron2013concentration}.
\end{proof}

\begin{proposition}[Bernstein-tail extension for sub-exponential residuals]
\label{prop:bernstein-tail-joint-recovery}
Replace the sub-Gaussian assumption in
Theorem~\ref{thm:joint-noise-superiority-supp} by independent centered
residuals satisfying the entrywise Bernstein MGF condition in
Lemma~\ref{lem:entrywise-bernstein-signed-sums}.
Let \(d_{\min}:=\min(k,m)\) and
\(\mathcal M=\binom Nk\binom Dm\). If
\[
  \mu \ge
  \max\!\left\{
    2v\sqrt{\frac{\log((\mathcal M-1)/\delta)}{d_{\min}}},
    \frac{2b\log((\mathcal M-1)/\delta)}{d_{\min}}
  \right\},
\]
then \((S^\star,F^\star)\) is the unique maximizer of \(B(S,F)\) with
probability at least \(1-\delta\). With deterministic entrywise calibration
error \(\|\widehat{\matW}-\matW^\star\|_\infty\le\eta\), the same statement holds after
replacing \(\mu\) by \(\mu-2\eta\).

For the sequential feature-first selector in
Corollary~\ref{cor:sequential-snr-supp}, the corresponding sufficient
condition is
\[
  \mu k \ge
  \max\!\left\{
    2v\sqrt{N\log\!\left(\frac{m(D-m)}{\delta}\right)},
    2b\log\!\left(\frac{m(D-m)}{\delta}\right)
  \right\},
\]
and under deterministic entrywise calibration error it is obtained by replacing
\(\mu k\) with \(\mu k-2N\eta\).
\end{proposition}

\begin{proof}
For a fixed joint competitor, the noise contrast is a signed sum over at most
\(2d\) residuals with \(d\ge d_{\min}\). Lemma~\ref{lem:entrywise-bernstein-signed-sums}
with \(r\le2d\) and \(t=\mu d\) makes the fixed-competitor failure probability
at most \(\delta/(\mathcal M-1)\) under the displayed joint condition. A union
bound proves joint recovery. Deterministic calibration error is handled by the
same \( \mu\mapsto \mu-2\eta \) replacement proved in
Proposition~\ref{prop:calibration-perturbation-supp}. For the feature-first
selector, each planted-vs-unplanted column contrast is a signed sum over \(2N\)
entries. Applying the same Bernstein bound with \(r=2N\), followed by a union
bound over \(m(D-m)\) pairs, gives the sequential condition and its calibrated
variant.
\end{proof}

\paragraph{Numerical verification.}
The sub-exponential check uses four parameter combinations
\[
(N,D,k,m)\in
\{(8,6,2,2),(12,8,2,2),(16,8,3,2),(20,6,2,2)\}
\]
at \(\delta=0.1\) with independent \(\mathrm{Laplace}(0,1)\) residuals. The
verification uses \(v=2\) and \(b=2\), which satisfy the stated Bernstein MGF
condition.

Each reported point uses 800 independent trials. For these parameter combinations, the joint selector achieves
\(\ge 1-\delta\) recovery at the theoretical threshold \(\mu^\star\). The
empirical recovery crossover is well below the sufficient threshold. At
\(\mu_{joint}^\star\), the joint selector meets or exceeds the feature-first
sequential selector whenever the two thresholds differ.

\begin{corollary}[Clipped-residual recovery with bias budget]
\label{cor:clipped-residual-recovery}
Let \(\zeta_{ij}\) be independent centered residuals, not necessarily
sub-Gaussian. For a clipping level \(\tau>0\), define
\[
  T_\tau(x):=\max\{-\tau,\min\{x,\tau\}\},
  \qquad
  \bar\zeta_{ij}:=T_\tau(\zeta_{ij})-\mathbb E T_\tau(\zeta_{ij}).
\]
Assume the clipped score field \(\matW^\tau\) can be written as
\[
  \matW^\tau_{ij}
  =
  \mu\,\mathbf 1\{i\in S^\star,\;j\in F^\star\}
  +\bar\zeta_{ij}+\mathcal D^\tau_{ij},
  \qquad
  \|\mathcal D^\tau\|_\infty\le \rho_\tau,
\]
where \(\mathcal D^\tau\) collects clipping bias, signal attenuation, or systematic
reference-window bias. If a calibrated estimate \(\widehat{\matW}\) additionally
satisfies \(\|\widehat{\matW}-\matW^\tau\|_\infty\le\eta\) and
\[
  \mu - 2(\rho_\tau+\eta)
  \ge
  2\tau\sqrt{
    \frac{\log((\mathcal M-1)/\delta)}{d_{\min}}
  },
\]
then \((S^\star,F^\star)\) is the unique maximizer of the calibrated joint
objective with probability at least \(1-\delta\).

For the sequential feature-first selector, a corresponding sufficient condition
is
\[
  \mu k - 2N(\rho_\tau+\eta)
  \ge
  2\tau\sqrt{
    N\log\!\left(\frac{m(D-m)}{\delta}\right)
  }.
\]
\end{corollary}

\begin{proof}
For a fixed joint competitor, deterministic bias and calibration change the gap
by at most \(2d(\rho_\tau+\eta)\). The centered clipped residual contrast is a
signed sum over at most \(2d\) variables in intervals of length \(2\tau\), so
Hoeffding's inequality~\cite{hoeffding1963probability} gives the displayed
joint condition after the same competitor union bound used in
Theorem~\ref{thm:joint-noise-superiority-supp}. For the
sequential selector, each column contrast contains at most \(2N\) clipped
residuals and at most \(2N(\rho_\tau+\eta)\) deterministic perturbation. The
same Hoeffding step and a union bound over \(m(D-m)\) feature pairs give the
sequential condition. This bias--variance structure parallels robust truncation
methods for heavy-tailed empirical risk~\cite{brownlees2015empirical}.
\end{proof}

\paragraph{Numerical verification.}
The heavy-tail check draws residuals from \(\mathrm{Cauchy}(0,1)\), which
has no finite moments. Clipping at \(\tau=1.5\) gives centered clipped
residuals bounded in \([-\tau,\tau]\) with Hoeffding range \(2\tau\). The
parameter combinations are unchanged:
\[
(N,D,k,m)\in
\{(8,6,2,2),(12,8,2,2),(16,8,3,2),(20,6,2,2)\}
\]
at \(\delta=0.1\).

Each reported point uses 800 independent trials. The clipped selector achieves \(\ge 1-\delta\) recovery at the theoretical
threshold \(\mu^\star\) both with bare clipped residuals
(\(\rho_\tau=\eta=0\)) and with an injected deterministic perturbation
\(\rho_\tau=0.3\) plus calibration error \(\eta=0.2\). The empirical recovery
crossover is again well below the sufficient threshold.

Without clipping, \(\mathrm{Cauchy}\) noise at \(\mu^\star\) reduces joint
recovery sharply because single heavy-tail outliers dominate the score field.
With clipping, every trial recovers the planted bicluster.
The experiment illustrates the bias--variance tradeoff formalized by the
corollary. Clipping introduces a deterministic bias budget \(\rho_\tau\) in
exchange for the concentration needed for recovery.

\begin{corollary}[Sequential feature-first sufficient signal-to-noise condition]
\label{cor:sequential-snr-supp}
Under the same model, define column scores
$C_j=\sum_{i=1}^{N}W_{ij}$ and select the top-$m$ columns.
A sufficient condition for exact recovery of $F^\star$ with probability at least
$1-\delta$ is
\[
\mu k \;\ge\; \sigma\sqrt{8N\log\!\left(\frac{m(D-m)}{\delta}\right)}.
\]
\end{corollary}

\begin{proof}
For each \(j\in F^\star\) and \(\ell\notin F^\star\),
\(\Delta_{j\ell}:=C_j-C_\ell\) equals \(\mu k\) plus a centered
\(2\sigma\sqrt N\)-sub-Gaussian noise term. Thus
\[
  \Pr(\Delta_{j\ell}\le0)
  \le \exp\!\left(-\frac{(\mu k)^2}{8N\sigma^2}\right).
\]
Exact feature recovery fails only if some planted-unplanted feature pair is
misordered. A union bound over \(m(D-m)\) pairs gives the condition.
\end{proof}

\begin{proposition}[Calibration-perturbation stability of joint recovery]
\label{prop:calibration-perturbation-supp}
Assume the planted model of
Theorem~\ref{thm:joint-noise-superiority-supp} for an ideal weight field
$\matW^\star$, and let $\widehat{\matW}$ be a calibrated estimate with
\[
\|\widehat{\matW} - \matW^\star\|_\infty \le \eta.
\]
Define $\widehat B(S,F)=\sum_{i\in S}\sum_{j\in F}\widehat{\matW}_{ij}$.
If
\[
  \mu - 2\eta \;\ge\; 2\sigma\sqrt{\frac{\log\!\left((\mathcal{M}-1)/\delta\right)}{\min(k,m)}},
\]
then with probability at least $1-\delta$, $(S^\star,F^\star)$ is the unique
maximizer of $\widehat B(S,F)$ over $|S|=k$, $|F|=m$.
\end{proposition}

\begin{proof}
Fix any competitor $(S,F)\neq(S^\star,F^\star)$ with $|S|=k$, $|F|=m$ and let
$d=km-|S\cap S^\star|\,|F\cap F^\star|$. As in
Theorem~\ref{thm:joint-noise-superiority-supp},
$d\ge d_{\min}:=\min(k,m)$ and the ideal gap is
\[
  B^\star(S^\star,F^\star)-B^\star(S,F)=\mu d+\Xi_{S,F},
\]
where $\Xi_{S,F}$ is centered and $\sigma\sqrt{2d}$-sub-Gaussian.
For calibration error
$\Delta_{ij}=\widehat{\matW}_{ij}-\matW^\star_{ij}$, define
\[
  R_{S,F}=
  \sum_{(i,j)\in S^\star\times F^\star}\Delta_{ij}
  -\sum_{(i,j)\in S\times F}\Delta_{ij}.
\]
Since $|S^\star\times F^\star\setminus S\times F|=|S\times F\setminus S^\star\times F^\star|=d$,
we have $|R_{S,F}|\le 2d\eta$.
Hence
\[
  \widehat B(S^\star,F^\star)-\widehat B(S,F)=\mu d+\Xi_{S,F}+R_{S,F},
\]
so
\[
  \Pr\!\left(\widehat B(S,F)\ge \widehat B(S^\star,F^\star)\right)
  \le
  \Pr\!\left(\Xi_{S,F}\le -d(\mu-2\eta)\right)
  \le
  \exp\!\left(-\frac{(\mu-2\eta)^2 d}{4\sigma^2}\right).
\]
Using $d\ge d_{\min}$ and a union bound over all $\mathcal{M}-1$ competitors gives
\[
  \Pr\!\left(\exists (S,F)\neq(S^\star,F^\star):
  \widehat B(S,F)\ge \widehat B(S^\star,F^\star)\right)
  \le
  (\mathcal{M}-1)\exp\!\left(-\frac{(\mu-2\eta)^2 d_{\min}}{4\sigma^2}\right).
\]
The stated condition makes this probability at most $\delta$.
\end{proof}

\begin{corollary}[Method-agnostic calibration-to-recovery transfer]
\label{cor:calibration-transfer}
Let $\delta=\delta_{\text{noise}}+\delta_{\text{cal}}$. Assume the planted-noise
conditions of Theorem~\ref{thm:joint-noise-superiority-supp} and a calibration
procedure that returns $\widehat{\matW}$ with
\[
\Pr\!\left(\|\widehat{\matW}-\matW^\star\|_\infty>\eta_{n,\delta_{\text{cal}}}\right)
  \le \delta_{\text{cal}},
\]
where $n=|\mathcal R|$ is reference-window size. If
\[
  \mu \;\ge\;
  2\eta_{n,\delta_{\text{cal}}}
  +2\sigma\sqrt{
    \frac{\log\!\left((\mathcal{M}-1)/\delta_{\text{noise}}\right)}{\min(k,m)}
  },
\]
then
\[
  \Pr\!\left((S^\star,F^\star)\text{ is uniquely recovered from }\widehat B\right)
  \ge 1-\delta.
\]
\end{corollary}

Let \(A=\{\|\widehat{\matW}-\matW^\star\|_\infty
\le\eta_{n,\delta_{\text{cal}}}\}\). On \(A\), a recovery failure is contained
in the same planted-noise contrast event bounded by
Proposition~\ref{prop:calibration-perturbation-supp}, so
\(\Pr(\text{failure}\cap A)\le\delta_{\text{noise}}\). Since
\(\Pr(A^c)\le\delta_{\text{cal}}\), direct event inclusion and a union bound
give the stated result without requiring independence between calibration and
the planted residuals.

\begin{corollary}[Sequential calibration sensitivity under entrywise perturbation]
\label{cor:sequential-calibration-sensitivity}
Under the setting of Corollary~\ref{cor:sequential-snr-supp}, let
$\widehat{\matW}$ satisfy $\|\widehat{\matW}-\matW^\star\|_\infty\le\eta$ and define
$\widehat C_j=\sum_{i=1}^N \widehat{\matW}_{ij}$.
A sufficient condition for exact recovery of $F^\star$ by top-$m$ selection on
$\widehat C_j$ with probability at least $1-\delta$ is
\[
  \mu k - 2N\eta
  \;\ge\;
  \sigma\sqrt{8N\log\!\left(\frac{m(D-m)}{\delta}\right)}.
\]
\end{corollary}

\begin{proof}
For any feature $j$,
$|\widehat C_j-C_j|
\le \sum_{i=1}^N |\widehat{\matW}_{ij}-\matW_{ij}^\star|
\le N\eta$.
Hence for each pair $(j,\ell)$ with $j\in F^\star$, $\ell\notin F^\star$,
\[
  \widehat \Delta_{j\ell}
  :=\widehat C_j-\widehat C_\ell
  =\Delta_{j\ell}+r_{j\ell},
  \qquad |r_{j\ell}|\le 2N\eta.
\]
Therefore
\[
  \Pr(\widehat \Delta_{j\ell}\le 0)
  \le
  \Pr(\Delta_{j\ell}\le 2N\eta)
  \le
  \exp\!\left(
    -\frac{(\mu k-2N\eta)^2}{8N\sigma^2}
  \right),
\]
using the same sub-Gaussian tail step as
Corollary~\ref{cor:sequential-snr-supp}. A union bound over $m(D-m)$ pairs
gives the claim.
\end{proof}

\paragraph{Sufficient-condition separation.}
Theorem~\ref{thm:joint-noise-superiority-supp} and
Corollary~\ref{cor:sequential-snr-supp} provide a conditional separation. The
joint certificate uses pairwise bicluster structure and incurs a combinatorial
competitor term through \(\mathcal M\), while feature-first sequential scoring
accumulates \(N^{1/2}\) noise in each column marginal. Under deterministic
entrywise perturbation, Proposition~\ref{prop:calibration-perturbation-supp}
requires
\[
  \mu
  \ge
  2\eta
  +2\sigma\sqrt{\frac{\log((\mathcal M-1)/\delta_{\mathrm{joint}})}{\min(k,m)}},
\]
whereas Corollary~\ref{cor:sequential-calibration-sensitivity} requires
\[
  \mu
  \ge
  \frac{
    2N\eta
    +\sigma\sqrt{8N\log(m(D-m)/\delta_{\mathrm{seq}})}
  }{k}.
\]
Thus the joint certificate loses \(2\eta\) under entrywise calibration error,
while the analyzed sum-column feature-first certificate loses \(2N\eta\) in the
column margin, equivalently \(2N\eta/k\) as a threshold on \(\mu\). For fixed
\(k,m\) and increasing \(N\), the feature-first sufficient threshold grows
faster.
The comparison is between sufficient conditions: if \(\mu\) lies between the
two displayed thresholds, the joint condition certifies recovery while the
analyzed feature-first condition does not certify recovery.

These results give finite perturbation bounds for the calibrated stage-2 weight
field. The baseline theorem assumes independent mean-zero
$\sigma$-sub-Gaussian entry noise and deterministic entrywise perturbation
bounded by $\eta$. The Bernstein-tail extension adds the corresponding
linear-log term for centered sub-exponential residuals. These bounds are most
informative for nearly stationary reference windows with controlled
residual tails. Drift, heteroskedasticity, and structured dependence motivate
the empirical calibration evaluations.

\subsection{Calibration-Robustness Experiments}

\paragraph{Controlled perturbation sweep.}
The controlled sweep perturbs the positive synthetic setting
\((N,D,k,m)=(20,10,3,5)\) under the feature-specific Gaussian generator and
noise fractions 40\%, 50\%, and 60\%. Separate empirical evaluations cover
drift and reference-window effects. For each seed, we first compute a clean
calibrated weight field $\matW^\star$, choose the highest-F1 clean sequential
rule among maximum, sum, and median aggregation rules, and then hold that rule
fixed while perturbing every entry according to
\[
\widehat{\matW} = \mathrm{clip}(\matW^\star + \Delta, 0, w_{\max}),
  \qquad \Delta_{ij}\sim\mathrm{Unif}[-\eta,\eta].
\]
Each draw therefore satisfies $\|\widehat{\matW}-\matW^\star\|_\infty\le \eta$ exactly.
The clip step makes the perturbation distribution element-wise asymmetric for
$\matW^\star_{ij}$ near $0$ or $w_{\max}$. Those entries see a one-sided effective
perturbation while interior entries see the full uniform draw. The achieved
infinity norm shown alongside $\eta$ is therefore at most the nominal
$\eta$ ceiling used on the robustness-curve axis. This empirical stress test
evaluates the theorem's entrywise perturbation premise under one planted
generator and one bounded perturbation law.

Panels~(a) and (b) of
Figure~\ref{fig:synthetic_calibration_robustness} show the resulting
seed-mean combined and feature F1 with bootstrap confidence bands. Joint
selection stays nearly flat through moderate perturbations, while the fixed
clean sequential rule degrades much earlier. At 50\% noise and $\eta=1$,
combined F1 is \csname CalibPerturbJointCombinedNoise50EtaOne\endcsname{} for
joint selection and \csname CalibPerturbSeqCombinedNoise50EtaOne\endcsname{}
for the fixed sequential rule. The corresponding joint-versus-sequential
values are \csname CalibPerturbJointFeatureNoise50EtaOne\endcsname{} versus
\csname CalibPerturbSeqFeatureNoise50EtaOne\endcsname{} for feature F1 and
\csname CalibPerturbJointExactNoise50EtaOnePct\endcsname{} versus
\csname CalibPerturbSeqExactNoise50EtaOnePct\endcsname{} for exact pair
recovery.
\needspace{18\baselineskip}

\begin{proposition}[Structured calibration modes under feature-first aggregation]
\label{prop:structured-calibration-modes}
Consider the theorem-aligned sequential baseline with sum-aggregated feature
scores
\[
  C_j := \sum_{i=1}^{N} W_{ij},
\]
and the exact-budget joint objective
\[
  B(S,F) := \sum_{i\in S}\sum_{j\in F} W_{ij}.
\]
\begin{enumerate}
  \item If the calibrated field has additive row and feature bias
  \[
\widehat{\matW}_{ij} = W_{ij} + \xi_i + \zeta_j,
  \]
  then
  \[
    \widehat C_j = C_j + \Xi + N \zeta_j,
    \qquad
    \Xi := \sum_{i=1}^{N} \xi_i,
  \]
  and
  \[
    \widehat B(S,F) = B(S,F) + m\sum_{i\in S} \xi_i + k\sum_{j\in F} \zeta_j.
  \]
  Thus row-only bias cancels from the sequential feature ranking, while
  feature bias is amplified by the pool size $N$ in the sequential margin and
  enters the joint objective with coefficient $k$, the sample budget.

  \item If the calibrated field has rank-one drift
  \[
\widehat{\matW}_{ij} = W_{ij} + u_i v_j,
  \]
  then
  \[
    \widehat C_j = C_j + U v_j,
    \qquad
    U := \sum_{i=1}^{N} u_i,
  \]
  and
  \[
    \widehat B(S,F) = B(S,F)
    + \left(\sum_{i\in S} u_i\right)\left(\sum_{j\in F} v_j\right).
  \]
  Hence centered shared drift matters to the sequential feature ranking only
  through its nonzero row-sum component $U$. A mean-shifted row factor behaves
  like feature bias after aggregation.
\end{enumerate}
\end{proposition}

\begin{proof}
Both claims follow by summing the perturbation term over $i$ for the sequential
feature score and over $S\times F$ for the joint objective.
\end{proof}

\paragraph{Structured perturbation decomposition.}
Proposition~\ref{prop:structured-calibration-modes} is an unclipped algebraic
statement, so clipping can in principle break the exact cancellation of row
bias. We test the bounded-weight setting by rerunning the
\((N,D,k,m)=(20,10,3,5)\) setting at 50\% noise. The analysis isolates five
structured perturbation modes and holds the theorem-aligned feature-first
sum-aggregation baseline fixed.
\[
  \Delta_{ij}\in
  \left\{
    \text{iid entrywise},\;
    a_i,\;
    b_j,\;
    u_i v_j \text{ with } \sum_i u_i \approx 0,\;
    u_i v_j \text{ with } \sum_i u_i > 0
  \right\}.
\]
Panels~(c) and (d) of
Figure~\ref{fig:synthetic_calibration_robustness} show the resulting
feature-level response. At $\eta=1$, sequential feature F1 remains
\csname StructuredCalibRowBiasSeqFeatureEtaOne\endcsname{} under row bias and
\csname StructuredCalibCenteredRankOneSeqFeatureEtaOne\endcsname{} under
centered rank-one drift, whereas feature bias and mean-shifted rank-one drift
reduce it to \csname StructuredCalibFeatureBiasSeqFeatureEtaOne\endcsname{} and
\csname StructuredCalibMeanShiftedRankOneSeqFeatureEtaOne\endcsname{},
respectively. Joint feature F1 under feature bias stays at
\csname StructuredCalibFeatureBiasJointFeatureEtaOne\endcsname{}. The row-bias
and feature-bias clipping rates are nearly identical,
\csname StructuredCalibRowBiasClipEtaOnePct\endcsname{} versus
\csname StructuredCalibFeatureBiasClipEtaOnePct\endcsname{}, which rules out a
clipping-frequency explanation. The synthetic message is therefore sharper than
the generic $2N\eta$ worst-case bound. The harmful mode is feature-aligned bias,
or shared drift with a nonzero row-mean component, because that is the structure
amplified by feature-first aggregation.
\begin{figure}[!htbp]
  \centering
  \includegraphics[width=\linewidth]{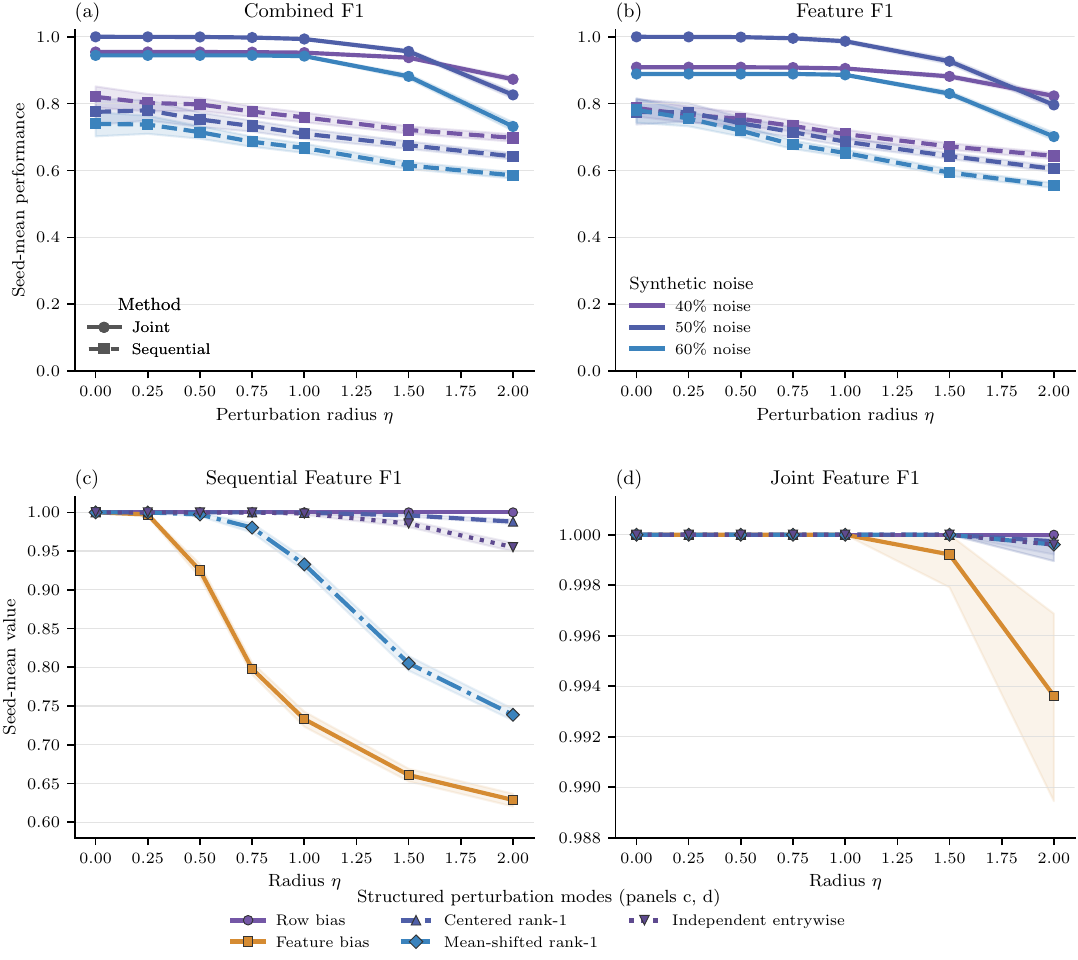}
  \caption{Feature-aligned perturbations degrade feature-first recovery much
  more than row-aligned perturbations, while joint feature F1 stays near one
  across the five structured modes in panel (d). Panels (a, b) sweep the
  perturbation radius at three noise levels. The sequential aggregation rule is
  selected on the clean field and held fixed while bounded entrywise perturbations are injected with
  $\Delta_{ij}\sim\mathrm{Unif}[-\eta,\eta]$. Panels (c, d) hold the 50\% noise
  setting fixed and isolate five structured modes with the theorem-aligned
  sum-aggregation sequential baseline.
  Means and 95\% percentile-bootstrap bands use 64 seeds and $B=1000$. Panel
  (d) expands the vertical scale.}
  \label{fig:synthetic_calibration_robustness}
\end{figure}

\paragraph{Credit Card Calibration Analyses.}
On labeled benchmark data the clean field $\matW^\star$ is unobserved, so we test operational
analogues of calibration error in the balanced Credit Card time-split analyses.
Across these analyses, the test slices, feature subsets, sequential aggregation
rule, and practical joint solver are fixed. Only the reference data or
calibration map used to estimate the stage-2 weight field changes.

First, we vary only the reference window, comparing the full window
against random 50\%, 25\%, 10\%, and 5\% subwindows. In the reference-window
panel of Figure~\ref{fig:real_creditcard_calibration_data_quality}, joint
performance remains at
or above the fixed sequential baseline throughout the sweep. At displayed
precision, the mean joint-minus-sequential anomaly-F1 gap is nonnegative at
every reference-window size. At the smallest reference fraction
(\csname RealCalibCreditLowestRefFracPct\endcsname{}) the mean gap is
\csname RealCalibCreditLowestRefMeanGapPct\endcsname{} percentage points. At
$(N,D,k,m)=(50,12,5,4)$ with a 10\% reference window, anomaly F1 is
\csname RealCalibCreditJointFiftyRefTen\endcsname{} for the joint selector
versus \csname RealCalibCreditSeqFiftyRefTen\endcsname{} for the sequential
baseline.

Second, we isolate temporal drift by refitting the weight field on fixed-size
contiguous pre-test blocks whose end times move farther back from the test window.
The temporal-drift panel of
Figure~\ref{fig:real_creditcard_calibration_data_quality} shows a positive
joint anomaly-F1 advantage after averaging over the displayed reference blocks
within each setting. For
$(N,D,k,m)=(50,12,5,4)$, the oldest block ends
\csname RealCalibDriftCreditLagOldestDays\endcsname{} days before the test
window. Anomaly F1 is \csname RealCalibDriftCreditJointFiftyOldest\endcsname{}
for the joint selector versus
\csname RealCalibDriftCreditSeqFiftyOldest\endcsname{} for the sequential
baseline, giving a joint-minus-sequential gap of
\csname RealCalibDriftCreditFiftyOldestGapPct\endcsname{} percentage points.
\begin{figure}[!htbp]
  \centering
  \includegraphics[width=\linewidth]{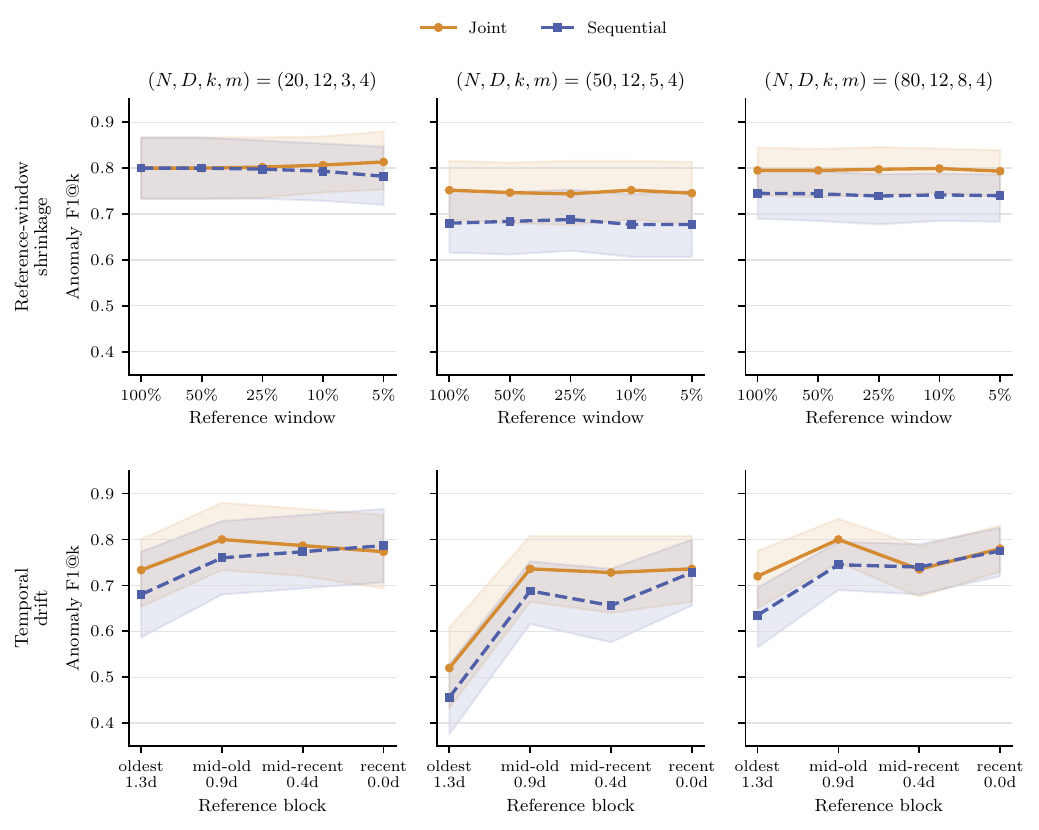}
  \caption{Joint anomaly F1 remains at or above the fixed feature-first baseline
  throughout the reference-window sweep and has a positive average advantage
  across the temporal-drift sweep. The top row varies the reference window over
  the full window and 50\%, 25\%, 10\%, and 5\% pre-test subwindows. The bottom
  row holds its size fixed and varies temporal drift through rolling pre-test
  blocks, with ticks giving end lag in days. Test slices, calibration map, and
  sequential rule are fixed. Means and method-wise 95\%
  percentile-bootstrap bands use 25 seeded slices and $B=1000$. The bands do
  not represent intervals for the paired method difference. Reference-window
  points average five replicates per seed.}
  \label{fig:real_creditcard_calibration_data_quality}
\end{figure}

\needspace{5\baselineskip}
The two degradation analyses induce different perturbation structures. For each
empirical weight-field difference $\Delta$, define
\[
  \bar\Delta=\frac{1}{ND}\sum_{i,j}\Delta_{ij},\qquad
  \delta_i^{\rm row}=\frac{1}{D}\sum_j\Delta_{ij}-\bar\Delta,\qquad
  \delta_j^{\rm feat}=\frac{1}{N}\sum_i\Delta_{ij}-\bar\Delta,
\]
and \(e_{ij}=\Delta_{ij}-\bar\Delta-\delta_i^{\rm row}-\delta_j^{\rm feat}\). The squared-Frobenius
shares are the squared norms of the global-mean, additive-row,
additive-feature, and residual matrices divided by \(\|\Delta\|_F^2\).
Figure~\ref{fig:real_creditcard_calibration_mechanism} shows the mean
decomposition and corresponding sup-norms by candidate-pool size and scenario.
\begin{figure}[!htbp]
  \centering
  \includegraphics[width=\linewidth]{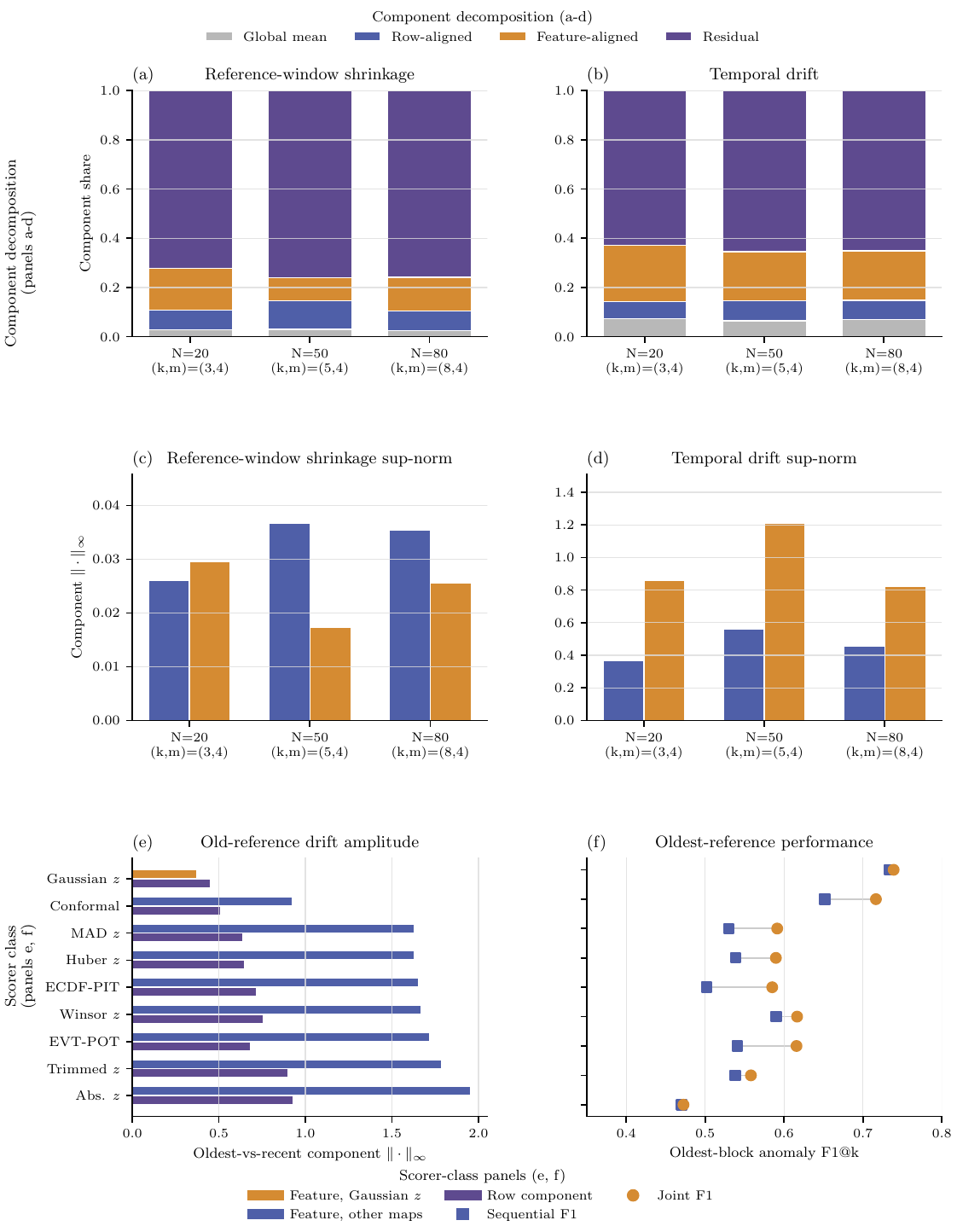}
  \caption{Reference-window shrinkage produces predominantly diffuse residual
  error, whereas temporal lag introduces stronger feature-aligned drift. Among
  the displayed calibration maps, Gaussian $z$ has the lowest feature-aligned
  drift and the highest oldest-block joint F1. Panels (a, b) decompose empirical
  weight-field differences into global-mean, additive-row, additive-feature, and residual components.
  Panels (c, d) report the corresponding row and
  feature sup-norms. Panels (e, f) compare nine calibration maps by pooled
  aligned-perturbation sup-norm and oldest-block anomaly F1, using the same
  scorer order.}
  \label{fig:real_creditcard_calibration_mechanism}
\end{figure}

Reference-window shrinkage is dominated by diffuse residual perturbations.
Pooled over the settings, the residual share is
\csname RealCalibStructRefResidualSharePooled\endcsname{}, versus
\csname RealCalibStructRefFeatureSharePooled\endcsname{} for the feature
component and \csname RealCalibStructRefRowSharePooled\endcsname{} for the row
component. The corresponding feature and row sup-norms are small,
\csname RealCalibStructRefFeatureInfPooled\endcsname{} and
\csname RealCalibStructRefRowInfPooled\endcsname{}.

Temporal reference lag has a more feature-aligned structure. The residual
component remains the largest single share at
\csname RealCalibStructDriftResidualSharePooled\endcsname{}. Among structured
components, feature-aligned drift dominates sample-aligned drift. Pooled over the
drift settings, the feature share is
\csname RealCalibStructDriftFeatureSharePooled\endcsname{} versus
\csname RealCalibStructDriftRowSharePooled\endcsname{} for the row component,
and the feature sup-norm is
\csname RealCalibStructDriftFeatureInfPooled\endcsname{} versus
\csname RealCalibStructDriftRowInfPooled\endcsname{}. Across the drift
settings, the feature component exceeds the row component in
\csname RealCalibStructDriftFeatureShareGtRowPct\endcsname{} of share
comparisons and \csname RealCalibStructDriftFeatureInfGtRowPct\endcsname{} of
sup-norm comparisons, with a pooled feature/row sup-norm ratio of
\csname RealCalibStructDriftFeatureRowInfRatio\endcsname{}. Thus the
feature-aligned mode in the structured synthetic result is a closer operational
model of drift from older reference windows than of finite-sample subsampling noise.

Finally, we compare nine stage-2 calibration maps on the same
oldest-vs-recent temporal-drift design. Panels~(e) and (f) of
Figure~\ref{fig:real_creditcard_calibration_mechanism}
show the resulting perturbation structure and oldest-block anomaly F1. The
Gaussian-tail $z$ calibration map has the smallest feature-aligned sup-norm at
all \csname RealCalibScorerZMinFeatureInfRegimeCount\endcsname{}
candidate-pool sizes and the highest oldest-block joint anomaly F1 at all
\csname RealCalibScorerZBestJointRegimeCount\endcsname{} sizes.

Pooled across candidate-pool sizes, the Gaussian-tail $z$ map's feature-aligned
sup-norm is
\csname RealCalibScorerZFeatureInf\endcsname{}, compared with
\csname RealCalibScorerMainProtocolFeatureInf\endcsname{} for the protocol
calibration map assigned to each candidate-pool size (conformal for $N=20$ and
$N=80$, and empirical-CDF probability-integral-transform (ECDF-PIT) for
$N=50$). The feature-aligned drift decreases by
\csname RealCalibScorerZVsMainFeatureInfReductionPct\endcsname.
Its pooled oldest-block joint anomaly F1 is
\csname RealCalibScorerZJointStaleFOne\endcsname{}, compared with
\csname RealCalibScorerMainProtocolJointStaleFOne\endcsname{} for the same
protocol maps.

Calibration design can therefore reduce feature-aligned drift from older
reference windows, the mode that feature-first aggregation amplifies. ECDF and conformal weights are
based on fewer distributional assumptions, whereas the Gaussian-tail $z$ map is a parametric
approximation after robust prestandardization. Calibration-map choice therefore
belongs to objective construction and must be validated before selecting
samples and features.
\FloatBarrier

\section{Experimental Protocols and Reproducibility}
\label{supp:repro}

\paragraph{Data availability.}
Synthetic experiments are seeded entirely by fixed integer seeds. No
external data is required. The Credit Card fraud dataset used
in our public-benchmark analyses is publicly available from
Kaggle~\cite{dalpozzolo2015calibrating,kaggle_creditcardfraud}. The IBM IT-AML transaction dataset is the
public synthetic anti-money-laundering (AML) benchmark introduced by Altman et
al.~\cite{altman2023realistic}. Our experiments use the public Kaggle-hosted
release of that benchmark~\cite{kaggle_ibm_itaml}.

\paragraph{Randomization budgets.}
The synthetic noise sweep at
$(N,D,k,m)=(\JointSeqNSamples{},\JointSeqNFeatures{},\JointSeqKAnomalies{},\JointSeqMFeatures{})$
uses $\JointSeqNSeeds{}$ seeds. Other synthetic sweeps use 50 seeds per
parameter combination unless stated otherwise. The synthetic calibration-perturbation sweep uses
\CalibPerturbNSeeds{} seeds, the structured calibration-mode sweep uses
\StructuredCalibModesNSeeds{} seeds, and the Credit Card calibration analyses
use \RealCalibCreditNSeeds{} seeds. Bootstrap confidence intervals (CIs) use
$B{=}1000$ resamples.
Captions or rows state slice, seed, or repeat counts where they affect the
comparison.

\paragraph{Synthetic QAOA circuit instances.}
\label{supp:gaussian_generator}
The formulation-level synthetic F1 experiments use the isolation-style weights
\(W_{ij}=|Z_{ij}|\). The Heron seed
panel, Heron execution controls, and fixed-angle hardware controls use a
separate seeded Gaussian QUBO generator. For a requested parameter setting
\((N,D,k,m)\), it draws
\(\min(3k,\lfloor N/3\rfloor)\) anomalous samples and \(N-\min(3k,\lfloor N/3\rfloor)\)
nominal samples. Nominal entries are sampled from \(\mathcal N(0,0.5^2)\). Anomalous
entries are sampled from \(\mathcal N(s_{ij},0.3^2)\), where the signs
\(s_{ij}\in\{-3,3\}\) are independent.

After sample shuffling, the generator estimates the pool mean and covariance,
assigns normalized sample scores
\((x_i-\bar x)^\top \widehat\Sigma^{-1}(x_i-\bar x)\), normalized feature
scores \(\mathrm{Var}_i(X_{ij})\), and normalized cross weights proportional to
\((X_{ij}-\bar X_j)^2(\widehat\Sigma^{-1})_{jj}\). The QUBO uses negative
sample and feature linear terms and symmetric negative sample--feature
couplings scaled by the generator coupling parameter. This construction defines
the energy instances used by the QAOA hardware protocols. Labeled
F1 comparisons use separate synthetic and labeled-benchmark protocols.

\subsection{QAOA Angle Protocols}
\label{supp:qaoa_angle_protocols}

The QAOA experiments use different angle sources depending on the
claim being tested. The common rule is that hardware-facing thresholds, candidate angle
sets, and optimizer budgets are fixed before the hardware execution they
evaluate. Exact-optimum probability is used only as enumerable-instance
validation.

\paragraph{Noiseless depth-frontier optimization.}
Noiseless depth-frontier simulations start directly from the exact tensor
Dicke state, without simulating preparation gates or preparation noise.
The feasible-sector expected energy is optimized with limited-memory
Broyden--Fletcher--Goldfarb--Shanno with bounds (L-BFGS-B). In the
\(p=1,\ldots,8\) frontier, each active angle is bounded to \([-1,1]\).
Each cell uses five deterministic starts plus three seeded random starts in that
interval, and the selected angles define the optimized state.

\paragraph{BK-threshold simulation study.}
The \(p=3\) BK-threshold simulation study uses exact-feasible basis
initialization. It evaluates each untrained angle initialization, then applies
random-local and SciPy Powell searches to the BK-hit objective. Those searches
use active-angle bounds or clipping in \([0,\pi]\), Powell tolerances
\(\mathtt{xtol}=\mathtt{ftol}=10^{-3}\), and a 96-call budget per start.
Selection maximizes BK hit rate.

\paragraph{Fixed-angle CG-XY-QAOA hardware executions.}
The planted-panel and Credit Card BK executions use exact-feasible basis
initialization and fixed angles selected before hardware execution. The planted
panel executes the simulation-selected fully grouped \(p=3\) schedule. The
Credit Card comparison executes tied, bilinear-CG, and fully grouped schedules
at \(p=3\) and \(p=5\) under the same predeclared threshold. The 32--64-qubit
fixed-angle runs use the bilinear-CG circuit form on its tied-phase subspace and
therefore do not test a grouping benefit.

\paragraph{Tensor-network warm-start angle selection.}
For the 36-, 52-, and 64-qubit warm-start hardware executions and matched
random-feasible controls, angle vectors are selected before hardware execution
by tensor-network shot simulations implemented with
cuQuantum~\cite{bayraktar2023cuquantum}. For each width and depth, the study
constructs the same bilinear-CG circuit submitted to hardware, including the
active patch, hardware-aligned sparse phase separator, fused Block~XY transport,
and exact-feasible basis-state preparation class used by the corresponding
hardware execution. The tensor-network simulation samples the final computational-basis
distribution of this circuit without enumerating the full
\(2^{N+D}\)-dimensional statevector.
The simulations contract the circuit tensor network with cuTensorNet without
bond-dimension truncation or matrix-product-state compression.
Numerical contractions use single-precision complex arithmetic.

Candidate angle vectors are generated separately for two classical feasible
starts at each width. For each start, the inherited schedule and five independent
Gaussian perturbations with standard deviation 0.2 per coordinate are evaluated
with 384 simulated shots per candidate. Each simulated shot is decoded in logical register
order, filtered to the exact \((k,m)\) sector, and rescored on the original dense
QUBO. Within each start, candidate selection minimizes \(\mathrm{CVaR}_5(E)\),
computed as the mean dense energy of the lowest-energy \(5\%\) of feasible decoded
samples. The best candidate from each start is reevaluated with 4096 simulated
shots. The final start-angle pair is selected by confirmed \(\mathrm{CVaR}_5(E)\),
followed by mean energy, best energy, and start rank. The selected angle vector is
reused without further optimization for the classical warm-start hardware
execution and the matched random-feasible control, so the two executions differ
only in the prepared basis state.

\paragraph{Penalty-X angle search.}
The penalty-X comparison uses matched candidate-random angle search with
\PenaltyXCandidateEvals{} candidates, \PenaltyXShotsPerEval{} shots per
objective estimate, and \PenaltyXFinalShots{} final shots.

\subsection{Metric and Table Conventions}

The QAOA tables separate optimized simulations, fixed-angle hardware
executions, and transpilation resource projections. Captions state the relevant
angle source, implementation path, shot budget, decoder, and threshold source
for each table when these fields affect the comparison. Resource-projection
rows report submitted circuit resources for fixed backend and implementation
conditions.

\paragraph{Table conventions.}
In hardware tables, Feas. denotes decoded exact-budget mass unless a
column explicitly names physical-order parsing or native strict feasibility.
For repeated hardware conditions, exact-budget entries give the repetition mean
followed by the observed minimum and maximum in square brackets. Single-run
conditions have no bracketed range. Elsewhere, square brackets contain interval
endpoints, with the caption identifying a confidence interval, interquartile
range, or observed range. Compact entries of the form \(x\pm h\) report a mean
and the larger half-width of the stated confidence interval.
The compact resource column \(d/G_{2Q}\) reports submitted circuit depth and
non-\texttt{SWAP} two-qubit gates. Angle-scale tuples are
ordered as sample field, feature field, bilinear coupling, and mixer. Delta
sign conventions are stated in captions. For energy gaps, lower values are better.

For F1 and true positives at \(k\) (TP@k) gaps, positive values favor the named
joint method unless the caption states otherwise. \(\mathrm{CVaR}_5(E)\), a
conditional-value-at-risk-style tail mean, is the mean energy of the lowest-energy
5\% of decoded-feasible shots. Classical method labels use Greedy+LS for greedy
local search, Coord.-ascent for coordinate ascent, Tabu for tabu search, and
Swap-SA for swap-based simulated annealing.

\section{Classical and Penalty-Based Baselines}
\label{supp:swap_sa}

The stage-2 classical baselines and penalty-X quantum control use the
same constrained bipartite objective as the XY-QAOA variants. Swap-SA denotes a
feasible-subspace simulated annealing method with full-neighborhood
budget-preserving swaps~\cite{metropolis1953equation,kirkpatrick1983optimization}.
Ring-local Kawasaki Metropolis and the Block~XY mixer use the same ring-local
graph~\cite{kawasaki1972kinetics,spitzer1970interaction}, with nearest-neighbor
token exchanges governed by fixed-temperature Metropolis dynamics in the former
and coherent \qaoa{} evolution in the latter.
The formulation-level comparisons evaluate the sequential baseline against each
prespecified joint heuristic directly. Additional controls cover
protocol stress, trigger menus, and penalty-X execution.

\paragraph{Matched-compute baseline ranking.}
Table~\ref{tab:solver_matched_compute} gives the ranking of stage-2 classical
baselines under matched wall-clock caps on the same synthetic instances.
\begin{table}[!htbp]
\centering
\caption[Matched wall-clock ranking on the synthetic stage-2 setting with 50 percent noise]{Matched wall-clock ranking on the synthetic stage-2 setting with 50\% noise. Coord.-ascent and Swap-SA reach the exact optimum on every displayed instance at both wall-clock caps. Each method is run on the same 50 seeded instances under the same target wall-clock cap. The sequential entry is a deterministic feature-first baseline. \texttt{max}, \texttt{sum}, and \texttt{median} marginal aggregations are each evaluated once, and the best aggregation by joint objective is retained. Combined F1 and gap to the exact joint optimum are means over the seeded instances. Runtime is the median. Exact-opt. hit is the share of instances with zero gap.}
\label{tab:solver_matched_compute}
\small
\setlength{\tabcolsep}{4pt}
\renewcommand{\arraystretch}{1.08}
\TableRowsSetup
\begin{tabular}{r l r r r r}
\TableHideRows
\toprule
\shortstack{\textbf{Cap}\\\textbf{(ms)}} & \textbf{Method} & \shortstack{\textbf{Combined}\\\textbf{F1}} & \shortstack{\textbf{Exact-opt.}\\\textbf{hit}} & \shortstack{\textbf{Mean gap}\\\textbf{to opt.}} & \shortstack{\textbf{Median}\\\textbf{runtime (ms)}} \\
\midrule
\TableBodyRows
10 & Sequential & 0.763 & 26\% & 7.41 & 0.07 \\
10 & Greedy+LS & 0.485 & 46\% & 8.45 & 0.09 \\
10 & Coord.-ascent & \textbf{1.000} & 100\% & 0.00 & 10.02 \\
10 & Tabu & 0.837 & 82\% & 2.73 & 10.17 \\
10 & Ring-local Kawasaki & 0.830 & 2\% & 6.45 & 10.04 \\
10 & Swap-SA & \textbf{1.000} & 100\% & 0.00 & 10.27 \\
\addlinespace
100 & Sequential & 0.763 & 26\% & 7.41 & 0.07 \\
100 & Greedy+LS & 0.485 & 46\% & 8.45 & 0.09 \\
100 & Coord.-ascent & \textbf{1.000} & 100\% & 0.00 & 100.02 \\
100 & Tabu & 0.837 & 82\% & 2.73 & 100.15 \\
100 & Ring-local Kawasaki & 0.900 & 12\% & 3.61 & 100.03 \\
100 & Swap-SA & \textbf{1.000} & 100\% & 0.00 & 100.28 \\
\bottomrule
\end{tabular}
\end{table}

\paragraph{Penalty-X baseline evaluation protocol.}
\label{supp:penalty_x_baseline}
A penalty formulation replaces exact-budget state preparation and a
feasibility-preserving mixer with a uniform $|+\rangle^{\otimes n}$ start, a
transverse-field $X$ mixer, and squared cardinality penalties in the cost
Hamiltonian. Under the convention $E(x)=x^\top Qx$, let
\[
  c_Q=\max\!\left\{10^{-6},\max_i|Q_{ii}|,
  \max_{i<j}|Q_{ij}+Q_{ji}|\right\}
\]
be the largest absolute binary-polynomial coefficient, including the numerical
floor used by the implementation. Because the penalty strength is an additional
modeling choice, we sweep $\lambda/c_Q$ and report the best setting under the
same evaluation budget. The statevector comparison uses small instances where the original
feasible optimum and the penalty-X distribution can both be evaluated exactly.

Both methods are scored on the same original constrained objective $x^\top Qx$
after restricting to exact-budget states. The penalty-X method is also evaluated
with exact-budget filtering, the standard way to extract feasible decisions from
a sampler over all bitstrings. The resource comparison
uses cost-and-mixer layers after Qiskit transpilation and excludes state
preparation, so Dicke preparation and the penalty-X
$|+\rangle^{\otimes n}$ start can be interpreted separately.
The resulting comparison is shown in
Table~\ref{tab:penalty_x_vs_strict_feasible_qaoa}.
\begin{table}[!htbp]
\centering
\footnotesize
\setlength{\tabcolsep}{4pt}
\renewcommand{\arraystretch}{1.08}
\caption[Penalty-X baseline against the hard-constraint XY-QAOA setting]{Penalty-X baseline against the hard-constraint XY-QAOA setting. Strict-feasible XY-QAOA remains in the feasible sector natively, whereas penalty-X has low native feasibility and higher two-qubit count. Instances are $(N,D,k,m)$. Both methods use the same matched candidate-random angle-search budget and are scored on the original constrained objective after exact-budget filtering. The penalty scale is $c_Q=\max\{10^{-6},\max_i|Q_{ii}|,\max_{i<j}|Q_{ij}+Q_{ji}|\}$ under $E(x)=x^\top Qx$. In all displayed instances, both settings reach zero best-gap after exact-budget filtering, so the table reports native feasibility, exact-optimum probability, selected penalty multiplier, and two-qubit resource overhead.}
\label{tab:penalty_x_vs_strict_feasible_qaoa}
\TableRowsSetup
\begin{tabular}{ccrrrrr}
\TableHideRows
\toprule
\textbf{$(N,D,k,m)$} & \textbf{$p$} & \shortstack{\textbf{Strict}\\\textbf{feas.}} & \shortstack{\textbf{Penalty-X}\\\textbf{feas.}} & \shortstack{\textbf{Opt. prob.}\\\textbf{strict / pen.}} & \shortstack{\textbf{Penalty}\\\textbf{$\lambda/c_Q$}} & \shortstack{\textbf{Penalty-X / strict}\\\textbf{2Q gates}} \\
\midrule
\TableBodyRows
$(4,4,2,2)$ & 1 & 100.0\% & 16.3\% & 4.17\% / 1.66\% & 0.25 & 1.30$\times$ \\
$(4,4,2,2)$ & 2 & 100.0\% & 9.7\% & 6.07\% / 2.51\% & 0.25 & 1.23$\times$ \\
$(5,5,2,2)$ & 1 & 100.0\% & 12.7\% & 1.18\% / 0.17\% & 2 & 1.39$\times$ \\
$(5,5,2,2)$ & 2 & 100.0\% & 16.4\% & 0.67\% / 0.76\% & 4 & 1.28$\times$ \\
$(6,6,3,3)$ & 1 & 100.0\% & 16.5\% & 0.25\% / 0.06\% & 8 & 1.37$\times$ \\
$(6,6,3,3)$ & 2 & 100.0\% & 14.5\% & 0.21\% / 0.08\% & 2 & 1.31$\times$ \\
\bottomrule
\end{tabular}
\end{table}

\FloatBarrier

\paragraph{Swap-SA procedure.}
Swap-SA searches the exact $(k,m)$ sector using sample, feature, and optional
coupled swaps. Biased proposal modes guide the search, and acceptance uses a
Metropolis-style temperature rule without a Hastings proposal-ratio correction.
Reported values are best-found objectives from the search trajectory. Restarts
and optional best-response refinement are part of the method.
Algorithm~\ref{alg:swap-sa} states the full procedure.

On the Credit Card time-split stage-2 slices, annealed and fixed-temperature
Swap-SA match at the displayed precision for all three candidate-pool sizes.

\begin{algorithm}[!htbp]
\caption{Swap-SA feasible-subspace simulated annealing}
\label{alg:swap-sa}
\begin{algorithmic}[1]
  \Require Weights $\matW\in\mathbb{R}^{N\times D}$, marginals $a\in\mathbb{R}^{N}$, $b\in\mathbb{R}^{D}$, budgets $(k,m)$
  \Require Restarts $R$, steps per restart $L$, temperature schedule
  $t\mapsto \mathcal T(t)$ (constant or annealed)
  \Require Move weights for sample/feature/coupled swaps, optional best-response period $P$
  \State $(S^\star,F^\star)\gets(\emptyset,\emptyset)$, $\mathrm{obj}^\star\gets-\infty$
  \For{$r=1$ to $R$}
    \State Initialize feasible $(S,F)$ (greedy or random) and compute $\mathrm{obj}(S,F)$
    \For{$t=1$ to $L$}
      \State Set temperature $T_t\gets \mathcal T(t)$
      \State Choose move type (sample/feature/coupled) by move weights
      \State Propose feasible swap $(S',F')$ (optionally biased toward high-score indices)
      \State $\Delta\gets \mathrm{obj}(S',F')-\mathrm{obj}(S,F)$
      \State Accept $(S,F)\gets(S',F')$ if $\Delta\ge 0$ else w.p.\ $\exp(\Delta/T_t)$
      \If{$P>0$ and $t\bmod P=0$}
        \State Best-response refine $(S,F)$ by alternately updating $S$ given $F$ and $F$ given $S$
      \EndIf
      \If{$\mathrm{obj}(S,F)>\mathrm{obj}^\star$} set $(S^\star,F^\star)\gets(S,F)$ and $\mathrm{obj}^\star\gets \mathrm{obj}(S,F)$ \EndIf
    \EndFor
  \EndFor
  \State \Return best feasible $(S^\star,F^\star)$
\end{algorithmic}
\end{algorithm}
\FloatBarrier

\paragraph{Ring-local Metropolis reference.}
\label{supp:qi-metropolis}
The ring-local Kawasaki Metropolis reference is the neighborhood-matched
classical reference on the same feasible graph used by the Block~XY mixer. It performs a
fixed-temperature Metropolis walk with best-energy tracking and
nearest-neighbor token exchanges. CG-XY-QAOA uses the same register-local
adjacency structure inside a coherent alternating-operator circuit with cost
phases and variational angle optimization. The comparison isolates the local
move graph. The full-neighborhood classical comparisons use Swap-SA, greedy
local search, coordinate ascent, and tabu search.
\FloatBarrier

\section{Synthetic Sensitivity and Mechanism Analyses}
\label{supp:synthetic_sensitivity}

Unless stated otherwise, all reported confidence intervals use the percentile
bootstrap with $B{=}1000$ resamples and a seeded
random-number generator (RNG). Means
are sample means over the slice or seed budget stated for the corresponding
comparison. We use
the percentile method because the F1@k and combined-F1
statistics are bounded and often discrete at the displayed slice counts, making
jackknife acceleration estimates unstable with the smallest slice counts.

These intervals are conditional on the selected protocol, calibration map,
solver set, split design, and hardware/backend choice. Protocol-selection
uncertainty is outside those
intervals. The confidence intervals are unadjusted across
the synthetic sensitivity grid. Each grid point is a prespecified comparison,
and we record whether its confidence interval for $\Delta\mathrm{F1}$ excludes
zero. For the synthetic results in
Table~\ref{tab:joint_vs_sequential_ci}, the nonzero noise levels
($20$--$60\%$) have intervals strictly above zero, while the zero-noise setting is
the null calibration case. When explicit family-wise tests are used in
sensitivity analyses, Holm--Bonferroni adjustment is applied across the named
comparison set~\cite{holm1979simple}.

\paragraph{Full synthetic comparison.}
The full synthetic comparison includes the feature-first sequential baseline,
all prespecified joint heuristics, the ring-local Kawasaki reference, and the
exact optimum available at this size.
\begin{table}[!htbp]
\centering
\caption[Synthetic feature-specific anomaly benchmark]{Synthetic feature-specific anomaly benchmark. Swap-SA and coordinate ascent attain the exact-optimum mean at every displayed noise level and outperform feature-first sequential whenever $\nu>0$. Columns give noise fraction $\nu$. Entries are mean combined F1, $\tfrac{1}{2}(\mathrm{F1@k}_{\mathrm{samples}}+\mathrm{F1@m}_{\mathrm{features}})$, over the seeded instances. Feature-first sequential reports the best of max, sum, and median aggregation by planted combined F1. Joint methods optimize the same exact-budget objective. The exact optimum is computable at this size. Bold marks the best reported member(s) of the prespecified full-neighborhood joint-heuristic set per column, with uninformative multiway ties left unbolded. Seeds count is n=200 per noise level.}
\label{tab:joint_vs_sequential_practical}
\setlength{\tabcolsep}{7pt}
\renewcommand{\arraystretch}{1.22}
\small
\TableRowsSetup
\begin{tabular}{l r r r r r}
\TableHideRows
\toprule
 & \multicolumn{5}{c}{Feature-specific noise fraction $\nu$} \\
\cmidrule(lr){2-6}
Method & 0\% & 20\% & 40\% & 50\% & 60\% \\
\midrule
\TableBodyRows
Feature-first sequential & 0.833 & 0.854 & 0.821 & 0.791 & 0.752 \\
Swap-SA & 0.833 & 0.885 & \textbf{0.955} & \textbf{1.000} & \textbf{0.944} \\
Greedy+LS & 0.833 & 0.885 & 0.813 & 0.535 & 0.395 \\
Coord.-ascent & 0.833 & 0.885 & \textbf{0.955} & \textbf{1.000} & \textbf{0.944} \\
Tabu & 0.833 & 0.885 & 0.946 & 0.874 & 0.735 \\
Ring-local Kawasaki & 0.833 & 0.884 & 0.935 & 0.946 & 0.907 \\
Exact opt & 0.833 & 0.885 & 0.955 & 1.000 & 0.944 \\
\bottomrule
\end{tabular}
\end{table}

Bootstrap intervals and paired effect sizes in
Table~\ref{tab:joint_vs_sequential_ci} show that the positive mean gaps at
20--60\% noise remain well separated from zero under the fixed seed budget.
\begin{table}[!htbp]
\centering
\caption[Bootstrap 95 percent confidence intervals (CI) and standardized paired effects for the mean...]{Bootstrap 95\% confidence intervals (CI) and standardized paired effects for the mean joint-vs-sequential combined-F1 improvement $\Delta$F1. The paired effect is $d_z=\bar\Delta/s_\Delta$ over the same seeded instances. It is omitted when the paired differences have zero variance.}
\label{tab:joint_vs_sequential_ci}
\small
\TableRowsSetup
\begin{tabular}{lrrrr}
\TableHideRows
\toprule
\textbf{Noise \%} & \textbf{$\Delta$F1 (mean)} & \textbf{95\% CI} & \textbf{$d_z$} & \textbf{\#seeds} \\
\midrule
\TableBodyRows
0\% & +0.000 & [+0.000, +0.000] & -- & 200 \\
20\% & +0.031 & [+0.023, +0.039] & 0.52 & 200 \\
40\% & +0.134 & [+0.115, +0.153] & 0.97 & 200 \\
50\% & +0.209 & [+0.186, +0.231] & 1.27 & 200 \\
60\% & +0.192 & [+0.171, +0.212] & 1.25 & 200 \\
\bottomrule
\end{tabular}
\end{table}

\begin{figure}[p]
\centering
\includegraphics[width=\linewidth,height=0.75\textheight,keepaspectratio]{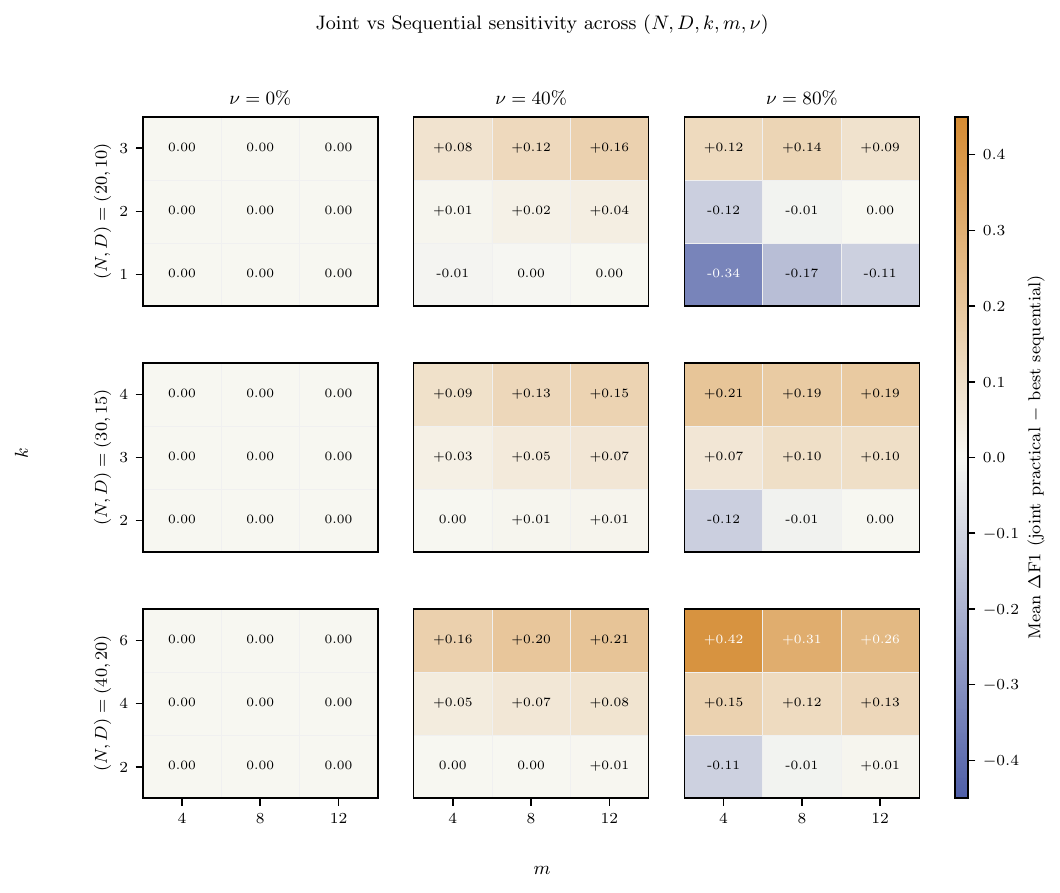}
\caption{Joint-vs-sequential parameter-sweep heatmaps on the synthetic
benchmarks. At $\nu=0$, joint and sequential selection tie in every evaluated
cell. Under feature-specific noise, the largest positive gaps occur at larger
$k$, while $k=1$ can favor sequential selection. Cells show mean
$\Delta\mathrm{F1}$ for the practical joint selector relative to the best
sequential baseline over 500 seeds. Rows vary $(N,D)$, columns vary synthetic
noise fraction $\nu$, and each panel sweeps anomaly budget $k$ against feature
budget $m$. Printed zeros denote evaluated cells with zero mean difference
after rounding.}
\label{fig:joint_vs_sequential_param_sweep_heatmaps}
\end{figure}

\paragraph{Effect sizes alongside the F1 gaps.}
The paired standardized effect \(d_z=\bar\Delta/s_\Delta\) follows the same
pattern across the wider parameter sweep. Positive effects concentrate when
$k>1$ and $m/D$ is intermediate. The largest losses occur at
$k=1$ under high feature noise. These cells show that joint selection helps when
the review budget requires coupled anomaly-feature signal across multiple
samples and features.

\paragraph{Failure-case analysis: $k{=}1$ at high synthetic noise.}
The worst-loss cell is structurally informative. With $k{=}1$, fixing a
candidate sample reduces the bilinear term
$\sum_{i\in S,j\in F}W_{ij}$ to a feature ranking for that sample, whereas the
joint method chooses the sample and feature set together. The feature-first
max-aggregation rule may use a different maximizing sample for each feature and can
therefore be competitive in this setting. Under high feature noise, joint local
search can instead converge to a spurious sample--feature pattern. At lower noise levels with
$k=1$, joint and sequential agree, identifying the
loss as a high-noise corner case.

\paragraph{Synthetic mechanism analysis.}
We isolate parameter combinations whose sign is stable across absolute-\(z\)
and median-absolute-deviation (MAD)-\(z\) calibration. The stable positive results
support the fixed-review-budget explanation. Joint selection helps when multiple
selected samples share informative, nonredundant feature signal. Pure noise and
$k=1$ erase that coupling. Table~\ref{tab:synthetic_mechanism_stable}
reports the calibration-specific lifts and feature-budget ratios.
\begin{table}[!htbp]
\centering
\caption[Calibration-map-stable findings from the synthetic mechanism suite]{Calibration-map-stable findings from the synthetic mechanism suite. All six displayed joint-minus-sequential $\Delta F_1$ values are positive, with MAD $z$ producing the larger lift in each setting. Abs. $z$ uses absolute reference-standardized residuals, while MAD $z$ uses reference median and median-absolute-deviation scaling. Active-feature F1 evaluates recovery on the informative feature subset. In the budget follow-up rows, $m/D$ is the feature-budget ratio with the largest active-feature F1 lift. The mechanism-grid row averages combined-F1 lift with redundant proxy burden fixed at $0.5$.}
\label{tab:synthetic_mechanism_stable}
\small
\setlength{\tabcolsep}{4pt}
\renewcommand{\arraystretch}{1.18}
\TableRowsSetup
\begin{tabular}{l l l r@{\hspace{1.1em}}r r@{\hspace{1.1em}}r}
\TableHideRows
\toprule
\textbf{Setting} & \textbf{Slice} & \textbf{Metric} & \multicolumn{2}{c}{\textbf{Abs. $z$}} & \multicolumn{2}{c}{\textbf{MAD $z$}} \\
\cmidrule(lr){4-5}\cmidrule(l){6-7}
 & & & $m/D$ & $\Delta F_1$ & $m/D$ & $\Delta F_1$ \\
\midrule
\TableBodyRows
Shared support & budget follow-up & active-feature max & $0.27$ & $+0.219$ & $0.27$ & $+0.344$ \\
Redundant decoys & budget follow-up & active-feature max & $0.27$ & $+0.180$ & $0.27$ & $+0.216$ \\
Proxy burden $0.5$ & mechanism grid & combined-F1 mean & -- & $+0.025$ & -- & $+0.060$ \\
\bottomrule
\end{tabular}
\end{table}

\FloatBarrier

\section{Public-Benchmark Comparisons and Protocol Sensitivity}
\label{supp:creditcard_checks}

\subsection{Balanced Public-Benchmark Comparisons}

Tables~\ref{tab:joint_vs_sequential_real_creditcard}
and~\ref{tab:joint_vs_sequential_real_ibm_aml_hi_small} give the balanced
public-benchmark comparisons. They compare the feature-first sequential
baseline with all prespecified strict-feasible joint heuristics under the fixed
stage-2 budgets.
\begin{table}[!htbp]
\centering
\caption[Credit Card fraud balanced time-split stage-2 results]{Credit Card fraud balanced time-split stage-2 results. At each candidate-pool size, the best reported joint heuristic has higher mean anomaly F1@k than feature-first sequential. Entries are mean anomaly F1@k $\pm$ the larger half-width of the 95\% percentile-bootstrap interval over slices. Balanced pools make F1@k equal to precision@k. Feature-first sequential is calibration-selected over max, sum, and median aggregation. Sum aggregation is selected at every candidate-pool size. Columns give $(N,D,k,m)$, the candidate-pool width and exact review budgets. Bold marks the best reported member of the prespecified full-neighborhood joint-heuristic set at each candidate-pool size. Slices count is n=25 per candidate-pool size.}
\label{tab:joint_vs_sequential_real_creditcard}
\small
\setlength{\tabcolsep}{7pt}
\renewcommand{\arraystretch}{1.22}
\TableRowsSetup
\begin{tabular}{l r r r}
\TableHideRows
\toprule
 & \multicolumn{3}{c}{Stage-2 dimensions and budgets $(N,D,k,m)$} \\
\cmidrule(lr){2-4}
Method & $(20,12,3,4)$ & $(50,12,5,4)$ & $(80,12,8,4)$ \\
\midrule
\TableBodyRows
Feature-first sequential & $0.800 \pm 0.067$ & $0.664 \pm 0.072$ & $0.745 \pm 0.055$ \\
Swap-SA & $\mathbf{0.827 \pm 0.067}$ & $0.752 \pm 0.072$ & $0.790 \pm 0.055$ \\
Greedy+LS & $0.800 \pm 0.067$ & $0.728 \pm 0.072$ & $\mathbf{0.795 \pm 0.055}$ \\
Coord.-ascent & $\mathbf{0.827 \pm 0.067}$ & $\mathbf{0.760 \pm 0.064}$ & $0.790 \pm 0.055$ \\
Tabu & $\mathbf{0.827 \pm 0.067}$ & $0.752 \pm 0.072$ & $0.790 \pm 0.055$ \\
Ring-local Kawasaki & $0.347 \pm 0.080$ & $0.304 \pm 0.072$ & $0.285 \pm 0.050$ \\
\bottomrule
\end{tabular}
\end{table}

\begin{table}[!htbp]
\centering
\caption[IBM IT-AML (HI-Small) balanced stage-2 results under chronological splits]{IBM IT-AML (HI-Small) balanced stage-2 results under chronological splits. At both candidate-pool sizes, the best reported joint heuristic has higher mean anomaly F1@k than feature-first sequential, with non-overlapping displayed 95\% bootstrap intervals. Entries are mean anomaly F1@k $\pm$ the larger half-width of the 95\% percentile-bootstrap interval over slices. Balanced pools make F1@k equal to precision@k. Feature-first sequential uses max aggregation with the supervised logistic-regression scorer. Columns give $(N,D,k,m)$, the candidate-pool width and exact review budgets. Bold marks the best reported member of the prespecified full-neighborhood joint-heuristic set at each candidate-pool size. Slices count is n=25 per candidate-pool size.}
\label{tab:joint_vs_sequential_real_ibm_aml_hi_small}
\setlength{\tabcolsep}{4pt}
\renewcommand{\arraystretch}{1.08}
\small
\TableRowsSetup
\begin{tabular}{l r r}
\TableHideRows
\toprule
 & \multicolumn{2}{c}{Stage-2 dimensions and budgets $(N,D,k,m)$} \\
\cmidrule(lr){2-3}
Method & $(20,38,3,4)$ & $(50,38,5,4)$ \\
\midrule
\TableBodyRows
Feature-first sequential & $0.467 \pm 0.093$ & $0.120 \pm 0.056$ \\
Swap-SA & $0.680 \pm 0.120$ & $0.344 \pm 0.128$ \\
Greedy+LS & $0.627 \pm 0.133$ & $\mathbf{0.400 \pm 0.112}$ \\
Coord.-ascent & $0.653 \pm 0.120$ & $0.328 \pm 0.120$ \\
Tabu & $\mathbf{0.720 \pm 0.120}$ & $0.328 \pm 0.120$ \\
\bottomrule
\end{tabular}
\end{table}

\FloatBarrier

\subsection{Credit Card Protocol Analyses}

\paragraph{Protocol sensitivity.}
Table~\ref{tab:joint_vs_sequential_real_creditcard_protocol_stress} compares two
sensitivity analyses of the balanced Credit Card protocol. Unbalanced pools relax the
balanced-slice requirement while retaining the same fixed-$k$ review budget,
and the two methods move much closer. The balanced analysis varies the
reference/buffer/test fractions while holding reference-only label-free
calibration and test-time sequential aggregation fixed. Its joint-minus-sequential
F1 differences remain positive across the tested splits.
\begin{table}[!htbp]
\centering
\caption[Credit Card protocol sensitivity]{Credit Card protocol sensitivity. Panel (a) relaxes the balanced-pool requirement, where F1@k differs from precision@k. Panel (b) varies the reference/buffer/test split while retaining balanced candidate pools. Entries are paired joint-minus-feature-first anomaly-F1@k differences, reported as the mean $\pm$ the larger half-width of the 95\% percentile-bootstrap interval. Positive values favor joint selection. Each entry uses the best prespecified joint heuristic for that setting.}
\label{tab:joint_vs_sequential_real_creditcard_protocol_stress}
\footnotesize
\setlength{\tabcolsep}{5pt}
\renewcommand{\arraystretch}{1.18}
\textit{(a) Unbalanced candidate pools}\par\smallskip
\TableRowsSetup
\begin{tabular}{lclr}
\TableHideRows
\toprule
$(N,D,k,m)$ & Mean prevalence & Joint method & $\Delta F_1$ \\
\midrule
\TableBodyRows
$(20,12,3,4)$ & 5.3\% & Swap-SA & $-0.025 \pm 0.050$ \\
$(50,12,5,4)$ & 2.0\% & Swap-SA & $0.000 \pm 0.000$ \\
\bottomrule
\end{tabular}

\medskip
\textit{(b) Balanced time-split sensitivity}\par\smallskip
\TableRowsSetup
\begin{tabularx}{\linewidth}{l *{3}{>{\centering\arraybackslash}X}}
\TableHideRows
\toprule
\textbf{$(N,D,k,m)$} & \shortstack{$P_{70/10/20}$\\($n_{ref}\!\approx\!199k$, $n=10$)} & \shortstack{$P_{60/20/20}$\\($n_{ref}\!\approx\!171k$, $n=25$)} & \shortstack{$P_{40/20/40}$\\($n_{ref}\!\approx\!114k$, $n=10$)} \\
\midrule
\TableBodyRows
(20,12,3,4) & $0.067 \pm 0.133$ & $0.027 \pm 0.053$ & $0.167 \pm 0.133$ \\
(50,12,5,4) & $0.080 \pm 0.140$ & $0.096 \pm 0.064$ & $0.160 \pm 0.100$ \\
(80,12,8,4) & -- & $0.050 \pm 0.040$ & -- \\
\bottomrule
\end{tabularx}
\end{table}

\FloatBarrier

\subsection{Credit Card Neighboring-Budget Analysis}
\label{supp:creditcard_budget_diagnostics}

Across 960 Credit Card candidate pools, we evaluated neighboring sectors
\(k=1,\ldots,16\) and
\(m\in\{2,3,5,7,10,14,20,28\}\) with the same greedy local-search and
coordinate-ascent selectors used in the public-benchmark comparisons. Labels
were used only after each label-free sector recommendation to compute
retrospective F1.

In empirical-CDF/PIT reservoir pools, the strongest label-free recommendations
reached retrospective F1 \(0.155\) while choosing a sample budget \(k\) that
exceeded the labeled-positive count by about \(13.5\) on average. The best
count-aligned rule reported in the
analysis reached F1 \(0.063\). The sweep therefore measures the tradeoff between
retrospective capture and review capacity and does not validate automatic
budget discovery. The primary optimization problem remains selection within the
\((k,m)\) sector fixed by the application budgets.

\subsection{IBM IT-AML Staged Trigger-Based Evaluation}
\label{supp:ibm_trigger_menu_realism}

We evaluate an IBM IT-AML trigger-threshold setting under chronological splits
using the full engineered \texttt{graph\_causal} feature set. These
graph-derived transaction summaries define the feature representation. The
protocol mirrors a common anti-money-laundering (AML)
alert-review setting. A stage-1 trigger scores chronological transaction
windows, a fixed threshold rule retains only medium- or high-risk windows, and
stage-2 selects a constrained $(k,m)$ subset within a candidate pool drawn from
the retained window. The sequential baseline is the feature-first two-stage
baseline.

A prerequisite for the stage-2 comparison is that stage-1 raises the positive
rate in the retained candidate pools. We fit a logistic trigger on historical
non-overlapping 20,000-transaction windows and hold it fixed on the test split.
For candidate-pool sizes \(N\in\{500,1000,2000\}\), recall@\(N\) is the pooled
fraction of labeled positives in these windows retained among the top \(N\)
trigger scores. The supervised trigger reaches recall@\(N\) of
\StageOneTurnonLogregRecallRangePct{} with precision/base-rate lift
\StageOneTurnonLogregEnrichRange{}$\times$. A label-free robust-z trigger,
computed from rolling 7-day past references and the sum of the four largest one-sided
z-scores, stays near baseline on the same budgets. The supervised retained
pools therefore contain enough labeled positives for the stage-2 comparison.

Table~\ref{tab:joint_vs_sequential_real_ibm_aml_hi_small_trigger_menu_stage1_baseline}
compares two baselines inside the retained windows. Greedy+LS has positive mean
true-positive-at-\(k\) (TP@k) differences from the generic feature-first
baseline in all four menu/budget conditions. Its mean differences from the
budget-aware top-\(m\) trigger are negative in all four. The stage-2 pools draw
\(N=500\) candidates uniformly from the top \(5N\) upstream
trigger scores. The budget-aware upstream ranker is therefore stronger under
this broad queue protocol.
\begin{table}[!htbp]
\centering
\caption[IBM IT-AML trigger-menu stage-2 TP@k differences]{IBM IT-AML trigger-menu stage-2 TP@k differences. Pools use sampling $N$ candidates uniformly from the top-$(5N)$ by the trigger score ($r=5$) with $N=500$, $D=38$, and $m=4$. Medium/high menus retain 10/29 and 6/29 test windows. $\Delta_{\mathrm{seq}}$ compares the joint heuristic with generic feature-first selection, and $\Delta_{\mathrm{budget}}$ compares it with the budget-aware top-$m$ trigger. Entries are paired means $\pm$ the larger half-width of the 95\% percentile-bootstrap interval. The displayed joint heuristic is Greedy+LS in every displayed row.}
\label{tab:joint_vs_sequential_real_ibm_aml_hi_small_trigger_menu_stage1_baseline}
\small
\setlength{\tabcolsep}{6pt}
\renewcommand{\arraystretch}{1.18}
\TableRowsSetup
\begin{tabular}{ l r r r }
\TableHideRows
\toprule
Menu & $k$ & $\Delta_{\mathrm{seq}}$ TP@k & $\Delta_{\mathrm{budget}}$ TP@k \\
\midrule
\TableBodyRows
Med & 20 & 0.48\,$\pm$\,0.72 & -0.64\,$\pm$\,0.64 \\
Med & 50 & 0.44\,$\pm$\,0.84 & -1.72\,$\pm$\,1.36 \\
High & 20 & 1.00\,$\pm$\,0.80 & -2.64\,$\pm$\,1.16 \\
High & 50 & 1.52\,$\pm$\,1.20 & -7.04\,$\pm$\,2.64 \\
\bottomrule
\end{tabular}
\end{table}

\section{CG-XY-QAOA Grouping and Sampling Results}
\label{supp:qaoa_depth}

\subsection{Noiseless Depth and Readout Results}
\label{supp:cgxy_depth_resource_evidence}

The primary simulator-side quantity is the optimized state's expected
feasible-sector energy, normalized to $\alpha\in[0,1]$ with
$\alpha=1$ at the exact feasible optimum and $\alpha=0$ at the
uniform-feasible reference
\(\bar C_\Omega=|\Omega|^{-1}\sum_{x\in\Omega}C(x)\). The depth sweep and
classical feasible-sector reference use the same normalization.

All optimized trajectories from the prespecified \(p=1,\ldots,8\) grid enter
these summaries. In the small-instance sweep, 53 of 54 tied trajectories and 53
of 54 fully grouped trajectories increase at every successive depth. The mean
is monotone for every perturbation-mode and signal-strength combination. All 12
tied and all 12 fully grouped trajectories in the \((14,10,3,3)\) series are
monotone.

\paragraph{Matched-resource comparison.}
The depth sweep compares tied and fully grouped parameterizations at the same
depth \(p\). It reports \(\alpha\), the grouped-minus-tied lift, and finite-shot
readout probabilities for the corresponding optimized states. A tied
depth-\(2p\) circuit provides a parameter-count-matched comparison while using
twice as many logical layers.

For $(N,D,k,m)=(14,10,3,3)$, the fully grouped cost-plus-mixer
CG-XY-QAOA mean increases monotonically across the \(p=1,\ldots,8\) ladder, and
the grouped-minus-tied mean is larger at every depth than at \(p=1\). The
lower-\(\mu\) setting yields the higher $\alpha$ at $p=8$.
On these instances, the matched exact-budget classical baseline reaches the
exact optimum at every depth. It takes the best result from Greedy+LS (greedy
local search), Coord.-ascent, Tabu, and Swap-SA under the same constrained
objective and $\alpha$ normalization. The CG-XY-QAOA ladder therefore reports
solution quality at fixed QAOA depth within the exact-budget sector.

Table~\ref{tab:qaoa_sparse_exact_14_10_readout} separates expected-energy
quality from finite-shot readout quality for the same optimized states. The
optimized distributions remain diffuse enough that exact-optimum probability
stays below saturation, even though best-of-4096 readout quality is much closer to
the optimum than the mean expected-energy score.

\begin{table}[!htbp]
\centering
\footnotesize
\setlength{\tabcolsep}{2.3pt}
\renewcommand{\arraystretch}{1.08}
\caption[Finite-shot readout metrics for the optimized (14,10,3,3) QAOA depth sweep over 12 planted...]{Finite-shot readout metrics for the optimized $(14,10,3,3)$ QAOA depth sweep over 12 planted instances. The fully grouped variant's higher mean \(\alpha\) at every depth does not translate into uniform improvement in finite-shot exact-optimum probability relative to tied XY-QAOA. Entries are means over the instances, with paired subcolumns for tied XY-QAOA and fully grouped CG-XY-QAOA. \(\alpha\) is one at the exact feasible optimum and zero at the uniform-feasible reference. \(P_\star(4096)\) is the probability of at least one exact-optimum sample in 4096 independent draws, Best-4096 is the expected best \(\alpha\), and \(P(\alpha\ge0.9)\) is a one-shot probability.}
\label{tab:qaoa_sparse_exact_14_10_readout}
\TableRowsSetup
\begin{tabular}{rrrrrrrrr}
\TableHideRows
\toprule
$p$ & \multicolumn{2}{c}{$\mathbb{E}\alpha$} & \multicolumn{2}{c}{$P_\star(4096)$} & \multicolumn{2}{c}{Best-4096 $\alpha$} & \multicolumn{2}{c}{$P(\alpha\ge0.9)$} \\
\cmidrule(lr){2-3}\cmidrule(lr){4-5}\cmidrule(lr){6-7}\cmidrule(lr){8-9}
& Tied & Grouped & Tied & Grouped & Tied & Grouped & Tied & Grouped \\
\midrule
\TableBodyRows
1 & 0.183 & 0.190 & 28.1\% & 34.0\% & 0.900 & 0.910 & $<0.05\%$ & $<0.05\%$ \\
2 & 0.254 & 0.267 & 56.3\% & 45.9\% & 0.947 & 0.932 & 0.1\% & 0.1\% \\
3 & 0.305 & 0.324 & 53.9\% & 27.1\% & 0.942 & 0.903 & 0.1\% & 0.1\% \\
4 & 0.344 & 0.364 & 53.0\% & 17.7\% & 0.945 & 0.883 & 0.1\% & 0.1\% \\
5 & 0.376 & 0.397 & 31.0\% & 21.4\% & 0.908 & 0.893 & 0.1\% & 0.1\% \\
6 & 0.403 & 0.427 & 12.3\% & 30.1\% & 0.871 & 0.903 & 0.2\% & 0.4\% \\
7 & 0.427 & 0.452 & 16.4\% & 33.6\% & 0.889 & 0.914 & 0.4\% & 0.7\% \\
8 & 0.448 & 0.473 & 23.8\% & 34.8\% & 0.897 & 0.915 & 0.8\% & 0.9\% \\
\bottomrule
\end{tabular}
\end{table}

\FloatBarrier
\subsection{Grouping Controls and Fixed-Threshold Results}
\label{supp:cgxy_grouping_controls}

Numerically, the mixed-derivative identity is checked across the tested
perturbation grid, and the equal-norm descent comparison is evaluated on an
imbalanced planted instance and on a controlled collinear instance. The former
shows the predicted grouped descent advantage, while the controlled collinear
case gives equality, matching
Corollary~\mainref{cor:grouped-equal-norm-descent}.

A same-parameter randomization control further separates the grouped response
from phase-parameter count alone. The type-preserving control keeps the
sample-field, feature-field, and bilinear taxonomy fixed and randomizes
assignments within those classes. The same-size control preserves only the three
group sizes and allows arbitrary term partitions.

\begin{table}[!htbp]
\centering
\footnotesize
\setlength{\tabcolsep}{3pt}
\renewcommand{\arraystretch}{1.08}
\caption{\emph{Randomization control for grouped phases.}
The semantic split lies near the upper tail of type-preserving regroupings and
is unexceptional under unrestricted same-size regrouping. The sample-field,
feature-field, and bilinear split is compared with random three-way
decompositions over \CGXYGroupingNullCaseCount{} exactly enumerable
feasible-sector instances. Median pct. is the percentile rank of the semantic
split in the corresponding random regrouping distribution,
\(P(>\!q_{95})\) is the fraction of cases above that distribution's 95th
percentile, and Gain/null mean is the ratio of semantic gain to random-null
mean gain.}
\label{tab:cgxy_grouping_randomization_control}
\TableRowsSetup
\begin{tabular}{lrrr}
\TableHideRows
\toprule
Control & Median pct. & \(P(>\!q_{95})\) & Gain/null mean \\
\midrule
\TableBodyRows
Type-preserving & \CGXYGroupingNullTypePreservingMedianPercentilePct{} &
\CGXYGroupingNullTypePreservingAboveQNinetyFivePct{} &
\(\CGXYGroupingNullTypePreservingGainOverMean{}\times\) \\
Same-size & \CGXYGroupingNullSameSizeMedianPercentilePct{} &
\CGXYGroupingNullSameSizeAboveQNinetyFivePct{} &
\(\CGXYGroupingNullSameSizeGainOverMean{}\times\) \\
\bottomrule
\end{tabular}
\end{table}

Together with the separate optimized depth sweep, the randomization control
supports the grouping as a structure-aware refinement of the bipartite
Hamiltonian. Exact-budget classical references
remain the lower-energy benchmark on enumerable instances.

\begin{table}[!htbp]
\centering
\small
\setlength{\tabcolsep}{4.2pt}
\renewcommand{\arraystretch}{1.12}
\caption[Grouped-versus-tied depth comparison]{\emph{Grouped-versus-tied depth comparison.} Fully grouped CG-XY-QAOA beats tied XY-QAOA in all 264 same-depth comparisons and exceeds the tied depth-\(2p\) reference in 3 of 264 comparisons. Mean lift is the mean \(G_p-A_p\) across cases, where \(G_p\) and \(A_p\) are the optimized fully grouped and layerwise XY-QAOA gains at depth \(p\). The parameter-matched reference is layerwise XY-QAOA at depth \(2p\). The two win columns report the number of cases where \(G_p\) exceeds \(A_p\) and \(A_{2p}\). Recovery is \((G_p-A_p)/(A_{2p}-A_p)\) and is reported as the median and interquartile range [Q1, Q3] across cases.}
\label{tab:qaoa_calibration_grouped_multiangle}
\TableRowsSetup
\begin{tabular}{lrrrrrr}
\TableHideRows
\toprule
Depth & Cases & \multicolumn{2}{c}{Wins} & Mean lift & Median recovery & Recovery IQR \\
\cmidrule(lr){3-4}
& & vs tied $p$ & vs tied $2p$ & & & \\
\midrule
\TableBodyRows
$p=1$ & 96 & 96/96 & 1/96 & 0.374 & 30.3\% & [17.7\%, 46.2\%] \\
$p=2$ & 96 & 96/96 & 1/96 & 0.487 & 36.4\% & [19.3\%, 61.4\%] \\
$p=3$ & 48 & 48/48 & 1/48 & 0.455 & 33.8\% & [23.6\%, 50.6\%] \\
$p=4$ & 24 & 24/24 & 0/24 & 0.370 & 29.6\% & [24.3\%, 42.2\%] \\
\bottomrule
\end{tabular}
\end{table}

Because the cost components commute and layerwise mixer angles leave mixer
support unchanged, the same-depth comparison changes parameter freedom without
adding logical gates or layers. Table~\ref{tab:qaoa_calibration_grouped_multiangle}
reports the recovery fraction \(r_p=(G_p-A_p)/(A_{2p}-A_p)\). The grouped
depth-$p$ variant has at least the same effective score whenever the added depth
in the tied $2p$ circuit preserves no more than that fraction of its noiseless
gain.

The planted-panel threshold is fixed from a prior tied-baseline hardware trace.
The hardware columns assess the preselected fully grouped schedule, whereas
same-device variant comparisons use the simulation-selected benchmark instance
in Table~\mainref{tab:creditcard_bilinear_cg_best_known_replay}.

\begin{table}[!htbp]
\centering
\footnotesize
\setlength{\tabcolsep}{2.5pt}
\renewcommand{\arraystretch}{1.1}
\caption[Planted-panel fixed-threshold results at p=3]{\emph{Planted-panel fixed-threshold results at $p=3$.} The selected fully grouped schedule produces threshold hits on all eight displayed IBM Heron R3 instances across ten predeclared 4096-shot repetitions per instance. The noiseless columns compare four parameterizations under the same fixed per-instance thresholds. Rates are $10^4p_{\rm BK}$, Hits is the pooled number of hardware threshold-hit shots, and Exact-budget mass is the mean decoded feasible-sector yield.}
\label{tab:cgxy_fully_grouped_qpu_replay_p3}
\TableRowsSetup
\begin{tabular}{rrrrrrrr}
\TableHideRows
\toprule
& \multicolumn{4}{c}{Noiseless $10^4p_{\rm BK}$} & \multicolumn{3}{c}{IBM Heron R3} \\
\cmidrule(lr){2-5}\cmidrule(lr){6-8}
Instance & XY-QAOA & Bilinear-CG & \shortstack{Fixed-transport\\grouped} & Fully grouped & $10^4p_{\rm BK}$ & Hits & Exact-budget mass \\
\midrule
\TableBodyRows
1 & 13.06 & 13.06 & 132.99 & 131.40 & 85.21 & 349 & 67.7\% \\
2 & 0.89 & 0.95 & 3.02 & 12.07 & 9.52 & 39 & 68.8\% \\
3 & 0.76 & 1.52 & 2.65 & 8.14 & 4.15 & 17 & 65.4\% \\
4 & 7.71 & 7.71 & 33.31 & 13.36 & 7.57 & 31 & 67.4\% \\
5 & 1.31 & 3.22 & 6.18 & 12.35 & 5.86 & 24 & 67.6\% \\
6 & 30.69 & 30.69 & 30.69 & 45.00 & 30.52 & 125 & 66.7\% \\
7 & 1.22 & 1.74 & 2.47 & 53.96 & 34.67 & 142 & 67.2\% \\
8 & 8.24 & 9.91 & 18.72 & 40.38 & 23.44 & 96 & 68.7\% \\
\bottomrule
\end{tabular}
\end{table}

\FloatBarrier
\section{Hardware Implementation, Decoding, and IBM Heron Execution}
\label{supp:hardware_detail}

\subsection{Transpilation Methodology and Decoding}
\label{supp:hardware_compilation_methodology}

\paragraph{Transpilation and resource accounting.}
For the simulated transpilation analyses, we build a logical circuit, choose
the platform coupling graph and native basis, and run Qiskit's transpiler with
the specified coupling map, initial-layout policy, and optimization level~3. The
transpilation resource tables and figures use Qiskit~2.1.2.
The square-lattice transpilation studies use generated sparse coupling maps, while the IBM
hardware analyses use the corresponding backend targets or hardware snapshots. For Heron runs with a fixed hardware patch, that patch is passed as the
initial layout. The Qiskit final layout is stored, and the same patch, shot
count, and decoding rule are held fixed inside each implementation-path comparison.

Submitted circuit depth is the standard circuit depth after parameter
binding, native-gate realization, routing, measurement insertion, and Runtime
instruction-set-architecture (ISA) validation. Runtime ISA validation confirms
that the final circuit uses only instructions, physical couplings, and parameter
values accepted by the selected target. Submitted two-qubit depth counts layers
containing operations on at least two qubits, and native two-qubit instruction
count excludes explicit \texttt{SWAP} operations. Routing \texttt{SWAP}s are
tracked separately where they affect resource accounting.

The target-based circuit-duration estimate is the longest circuit path
weighted by the instruction durations stored in the backend target. It includes
final measurement duration and excludes qubit reset, inter-shot repetition
delay, Runtime service overhead, and queue time. It is therefore a target-based
estimate rather than device-side timing telemetry. Feasibility is evaluated in
two ways. Physical-order parsing treats physical classical-bit order as logical
order. Layout-decoded parsing first normalizes Qiskit count keys to physical
little-endian bitstrings and then maps them through the returned final layout
before checking $(|\vecs|,|\vecf|)=(k,m)$.

For a symmetric QUBO matrix \(Q\) with \(E(x)=x^\top Qx\), conversion with
\(x=(I-Z)/2\) gives
\[
  h_i=-\frac{1}{2}\left(Q_{ii}+\sum_{j\ne i}Q_{ij}\right),\qquad
  J_{ij}=\frac{Q_{ij}}{2}.
\]
The circuit therefore applies \(R_Z(2\gamma h_i)\) and
\(R_{ZZ}(2\gamma J_{ij})\). For grouped phases, the same conversion is applied
to each binary cost component separately before its angle is assigned.

\paragraph{Register-preserving variable placement.}
For a fixed hardware patch, let \(E_\times\subseteq[N]\times[D]\) denote the
cross-register position pairs connected by executable hardware edges. Before
applying the hardware-edge mask, an exact mixed-integer linear program selects
within-register permutations by solving
\[
  \max_{\pi_{\mathcal S}\in\mathfrak S_N,
        \pi_{\mathcal F}\in\mathfrak S_D}
  \sum_{(u,v)\in E_\times}
  \left|W_{\pi_{\mathcal S}(u),\pi_{\mathcal F}(v)}\right|.
\]
If \(P\) is the resulting permutation matrix from circuit-logical order to
original variable order, the circuit uses \(Q_{\rm placed}=P^\mathsf{T}QP\).
Thus \(y^\mathsf{T}Q_{\rm placed}y=(Py)^\mathsf{T}Q(Py)\), and the two
cardinality constraints are unchanged because \(P\) never moves a variable
between registers. The solver reached zero mixed-integer optimality gap for
every submitted placement problem. Placement optimization uses the dense QUBO and fixed
patch support before simulator samples or hardware outcomes are available.
After layout-aware decoding, \(P\) restores original sample and feature labels
before dense or sparse scoring.

The fractional-gate implementation used for hardware execution constrains native
\(R_{ZZ}\) angles to the backend-supported interval. We use Qiskit's convention
\[
  R_{ZZ}(\theta)=\exp[-i\theta Z\otimes Z/2].
\]
For any integer \(q\),
\[
  R_{ZZ}(\theta+q\pi)=(-iZ\otimes Z)^q R_{ZZ}(\theta)
\]
up to global phase. For odd \(q\), the residual \(Z\otimes Z\) factor is
implemented by local \(R_Z(\pi)\) corrections, since
\(R_Z(\pi)\otimes R_Z(\pi)=-Z\otimes Z\) under this convention and therefore
implements \(Z\otimes Z\) up to global phase. A
negative folded angle is converted to a positive native angle by the exact echo
identity
\[
  (X\otimes I)R_{ZZ}(\theta)(X\otimes I)=R_{ZZ}(-\theta).
\]
Thus an unrestricted logical \(R_{ZZ}(\theta)\) can be represented exactly by a
native-range positive \(R_{ZZ}\) gate together with local \(R_Z(\pi)\)
corrections and, when needed, an \(X\)-echo on one endpoint.

The Block~XY mixer fusion uses Qiskit's \(\mathrm{XXPlusYY}\) convention
\[
  \mathrm{XXPlusYY}(\phi,0)
  =
  \exp[-i\phi(XX+YY)/4].
\]
Because \(XX\) and \(YY\) commute,
\[
  R_{XX}(\theta)R_{YY}(\theta)
  =
  \exp[-i\theta(XX+YY)/2]
  =
  \mathrm{XXPlusYY}(2\theta,0).
\]
The implementation applies this rewrite only to adjacent equal-angle
\(R_{XX}(\theta)R_{YY}(\theta)\) pairs on the same edge. Edge-colored scheduling then groups qubit-disjoint
\(\mathrm{XXPlusYY}\) gates into parallel transport layers.

These identities give the pre-routing resource accounting for the active
cost-and-mixer layers. Let \(E_{\rm sparse}\) be the retained sparse cost-term
set and \(E_{{\rm XY},\ell}\) the Block~XY transport edges in layer \(\ell\).
Before backend resynthesis,
\[
  G_{2Q}^{\rm cost}=p|E_{\rm sparse}|,\qquad
  G_{2Q}^{\rm mixer,unfused}=2\sum_{\ell=1}^{p}|E_{{\rm XY},\ell}|,\qquad
  G_{2Q}^{\rm mixer,fused}=\sum_{\ell=1}^{p}|E_{{\rm XY},\ell}|.
\]
The \(R_{ZZ}\) folding keeps the cost-layer two-qubit count fixed at this
logical accounting level. The costs of its local corrections are included in
every post-transpilation resource metric.

The factorial implementation study uses 13 decision-register widths at
\(p\in\{2,3\}\), giving \CompilePathMatchedPointCount{} matched width--depth
points. Every condition uses the same retained cost terms, physical patch,
initial layout, basis initialization, bilinear-CG parameterization, deterministic
angles \(\gamma=0.9\), \(\gamma_{\mathcal S\mathcal F}=2.4\), and
\(\beta=1.15\), backend snapshot, and
\CompilePathTranspilerSeedCount{} transpiler seeds. These resource-only angles
exercise native-range folding and are not selected from simulator or hardware
outcomes.

The 16 conditions cross the standard CZ-basis and fractional-gate targets,
unfused and fused Block~XY mixers, recorded and edge-colored transport orders,
and parameter binding before and after transpilation. Opt-3 CZ and Opt-3
fractional are the unfused, recorded-order, bind-before controls under Qiskit's
optimization-level~3 preset pass manager. The complete CZ and fractional
Block~XY paths use fusion, edge-colored transport, and late parameter binding.
Every fractional condition applies exact native-range \(R_{ZZ}\) folding before
Runtime ISA validation. The factors interact, so Table~\ref{tab:implementation_path_factorial}
reports paired path contrasts rather than additive factor shares.

\begin{table}[!htbp]
\centering
\footnotesize
\caption{\emph{Controlled implementation-path reductions.} Entries give the
median and observed minimum--maximum percentage reduction across
\CompilePathMatchedPointCount{} paired width--depth points. Negative values
denote increases, and the observed ranges are descriptive rather than confidence
intervals. The complete fractional path reduces each reported metric at every
point. Switching targets alone reduces circuit depth without uniformly reducing
the target-based circuit-duration estimate. Native two-qubit instruction count
falls by
\CompilePathNativeTwoQCountReductionRangePct{} under the fractional target and
is unchanged by fusion or transport ordering.}
\label{tab:implementation_path_factorial}
\setlength{\tabcolsep}{3pt}
\TableRowsSetup
\begin{tabularx}{\linewidth}{>{\raggedright\arraybackslash}Xlr@{\;}rr@{\;}rr@{\;}r}
\TableHideRows
\toprule
Implementation path & Reference path & \multicolumn{6}{c}{Reduction across matched points (\%)} \\
\cmidrule(l){3-8}
& & \multicolumn{2}{c}{Circuit depth} & \multicolumn{2}{c}{2Q depth} & \multicolumn{2}{c}{Duration estimate} \\
\cmidrule(lr){3-4}\cmidrule(lr){5-6}\cmidrule(l){7-8}
& & \multicolumn{1}{r}{Median} & Range & \multicolumn{1}{r}{Median} & Range & \multicolumn{1}{r}{Median} & Range \\
\midrule
\TableBodyRows
Opt-3 fractional & Opt-3 CZ & \(31.2\) & \([25.8,35.8]\) & \(0.0\) & \([0.0,8.3]\) & \(4.1\) & \([-1.2,9.1]\) \\
Complete CZ Block-XY & Opt-3 CZ & \(25.9\) & \([10.8,36.3]\) & \(23.1\) & \([0.0,36.4]\) & \(15.3\) & \([7.0,24.4]\) \\
\addlinespace[3pt]
Complete fractional Block-XY & Opt-3 fractional & \(23.5\) & \([6.6,36.4]\) & \(23.1\) & \([8.3,33.3]\) & \(17.5\) & \([8.8,25.1]\) \\
Complete fractional Block-XY & Opt-3 CZ & \(47.4\) & \([35.1,54.8]\) & \(23.1\) & \([8.3,36.4]\) & \(20.0\) & \([12.5,29.2]\) \\
\bottomrule
\end{tabularx}

\end{table}

Algorithm~\ref{alg:layout_aware_decoding} gives the layout-aware decoding map
used after routing permutes physical wires. The map restores logical sample and
feature registers before testing the strict \((k,m)\) budget window.

\begin{algorithm}[!htbp]
\caption{Layout-aware strict-feasibility decoding}
\label{alg:layout_aware_decoding}
\begin{algorithmic}[1]
  \Require Qiskit measurement count key \(c\), final layout \(L\), logical budgets \((k,m)\)
  \State \(y\gets\) physical little-endian bitstring normalized from \(c\)
  \For{logical qubit \(q=0,\ldots,N+D-1\)}
    \State \(p\gets L(q)\) \Comment{physical qubit holding logical qubit \(q\) at measurement}
    \State \(x_q\gets y_p\)
  \EndFor
  \State \(\vecs\gets (x_0,\ldots,x_{N-1})\), \(\vecf\gets (x_N,\ldots,x_{N+D-1})\)
  \State \Return decoded \((\vecs,\vecf)\) and feasibility flag \(\sum_i s_i=k,\ \sum_j f_j=m\)
\end{algorithmic}
\end{algorithm}

\begin{lemma}[Routing requires layout-aware decoding]
\label{lem:routing-feasibility}
Consider an ideal logical circuit that preserves the pair of Hamming weights
$(|\vecs|,|\vecf|)=(k,m)$, for example feasible initialization plus Block~XY
mixing. Let a transpiler insert SWAP routing and return measurement outcomes in
physical wire order. In noiseless simulation, the transpiled circuit is
logically equivalent to the original circuit up to a permutation of wires
described by the final layout. Applying the induced layout-aware decode map
restores the original logical variable order. Therefore decoded feasibility
remains $100\%$ under the ideal transpiled circuit. Physical-order parsing
measures the routed bit partition and can therefore report smaller feasibility.
\end{lemma}

\subsection{Sparse-Surrogate Integrity and Sparse-Objective Controls}
\label{supp:sparse_surrogate_controls}

Depth attribution with circuit barriers separates two dominant resource
sources: deterministic Dicke preparation when used, and routing plus synthesis
for dense \(N\times D\) cost interactions. The hardware runs therefore use
feasible basis-state initialization and hardware-aligned cost sparsification.
In the dense-vs-sparsified resource-scaling sweep, dense and sparsified
circuits use the same Qiskit optimization-level-3 settings on each target's
native basis. That comparison isolates the resource effect of retaining only
hardware-edge cost terms across the sparse-connectivity topologies tested here.
The Heron implementation-path study separately compares the Opt-3 CZ and
Opt-3 fractional controls with complete Block~XY paths on both targets. This
factorial comparison separates target exposure from the additional effects of
fusion, transport order, and parameter-binding order.

The resource-scaling sweep uses the joint bipartite cost model of
Eq.~\mainref{eq:hamiltonian}. A
seeded standard-normal data matrix is formed, the first \(k\) samples are nudged
by \(+2\) on the first \(\min(m,D)\) features, and the resulting objective is
converted to Ising form without cardinality penalties because the Block~XY
mixer enforces the budgets. The sweep covers five decision-register widths and
20 transpiler seeds per cell, with QAOA depths \(p=1,\ldots,4\), the IQM Emerald
square-lattice topology and IBM Heron R3 heavy-hex topology, and each target's
native basis. It also reports retained sample--feature coupling mass on
submitted 32--64-qubit fixed-angle layouts and the submitted Heron resources
used for the runs through 64 decision qubits.

For hardware-aligned coupling-grouped execution, the diagonal fields are always
executable as single-qubit rotations, while the bilinear-CG angle
$\gamma_{\mathcal S\mathcal F}$ acts only through sample--feature couplings
retained on the selected hardware patch. Their retained mass is
\[
  R_{\mathcal S\mathcal F}
  =
  \sum_{\substack{i\in[N],\,j\in[D]\\(i,N+j)\in E_{\mathrm{hw}}}}
    \left|W_{ij}\right|,
\]
reported together with its ratio to the dense cross-coupling mass. This quantity
depends on the calibrated instance and selected hardware patch. When
$R_{\mathcal S\mathcal F}$ is negligible, a coupling-grouped hardware run has
little bilinear signal to test and mainly measures feasible sparse-cost sampling
and diagonal phase controls.

For each 32--64-qubit condition, the exact sparse optimum is obtained by
enumerating selections on the endpoints of the retained cross edges and filling
the remaining budget positions by diagonal cost. The selected samples and
features are then rescored on the dense objective to compute
\(g_{\rm dense}\) in Eq.~\mainref{eq:dense_reference_gain_retained}.

The hardware-aligned Block~XY implementation has lower submitted two-qubit
count and depth than the dense-cost implementation at every tested width and
depth on both sparse-connectivity targets. The largest absolute savings occur
on the IBM Heron R3 heavy-hex target,
where dense sample--feature costs require more routing. Across the native-basis
comparisons, the median aligned-to-dense two-qubit-count ratio is below one for both
targets.
The resource controls separate two effects: retaining only hardware-edge cost
terms lowers sparse-connectivity routing pressure, and fractional scheduling
changes the implementation path for the same logical problem.

Table~\ref{tab:cgxy_fixed_angle_condition_details} lists the retained fraction
of dense sample--feature coupling mass and the strict and repaired best-energy
gaps for the 32--64-qubit Credit Card conditions. These quantities connect the
sparse Hamiltonian executed on hardware to its dense-objective rescoring and
sparse classical reference.

\subsection{IBM Heron R3 Run Protocols}
\label{supp:hardware_heron_execution}

The IBM Heron R3 study contains 169 fixed circuits in five submitted groups.
The planted panel uses eight instances, ten 4096-shot repeats per instance, and
the simulation-selected fully grouped \(p=3\) schedule. The Credit Card
comparison uses tied, bilinear-CG, and fully grouped schedules at \(p=3\) and
\(p=5\), with ten 4096-shot repeats per condition. Fixed-angle scaling uses one
8192-shot run at 32 and 56 qubits and three 8192-shot repeats at 36, 52, and 64
qubits. Warm-start and matched random-feasible controls use three 8192-shot
repeats at 36, 52, and 64 qubits.

The six Credit Card 20-qubit conditions were interleaved by repeat within one
Runtime job. The execution-calibration record passed the predeclared patch
acceptance checks.

All circuits use preselected hardware patches, register-preserving variable
placement, hardware-edge sparsification, the fractional-gate implementation,
Block~XY fusion, edge-colored transport ordering, and layout-aware decoding.
The 20-qubit conditions use exact-feasible basis initialization. The
32--64-qubit fixed-angle conditions use objective-blind feasible basis states. Warm-start
and random-feasible states and their tensor-network-selected angles are fixed
before hardware execution as described in
Section~\ref{supp:qaoa_angle_protocols}. Applying the final-layout decode map
and the recorded within-register permutation restores the logical register
assignment and original variable labels before budget checks and scoring.

\subsection{IBM Heron R3 Sampling Results}
\label{supp:three_tier_cross_validation}

At 20 qubits, BK hit rate is the primary sampling-tail metric. Supplementary
Table~\ref{tab:cgxy_fully_grouped_qpu_replay_p3} gives the probability of
sampling a decoded feasible solution at or below the predeclared threshold on
each planted instance. Threshold hits occur on
\CGXYFullyGroupedQPUHitInstances{} instances.
Table~\mainref{tab:creditcard_bilinear_cg_best_known_replay} gives the matched
Credit Card comparison. Fully grouped angles have the highest measured BK rate
at both depths, reaching \CreditTwentyQPFiveFullyGroupedBKRatePct{} at \(p=5\)
versus \CreditTwentyQPFiveTiedBKRatePct{} for tied XY-QAOA.

\paragraph{Repeat-level BK variation.} For tied XY-QAOA, bilinear-CG, and fully grouped CG-XY-QAOA, respectively, the pooled BK hit counts at $p=3$ are 19, 22, and 161 of 40{,}960 shots. Across the ten predeclared 4096-shot repetitions, the corresponding median rates [minimum, maximum] are 0.037\% [0\%, 0.098\%], 0.049\% [0\%, 0.12\%], and 0.42\% [0.22\%, 0.56\%]. At $p=5$, the pooled counts are 59, 107, and 1{,}503 of 40{,}960 shots, and the repeat-level rates are 0.15\% [0.098\%, 0.22\%], 0.24\% [0.073\%, 0.42\%], and 3.7\% [3.3\%, 4.3\%].

For the 32--64-qubit conditions, exact-budget mass measures feasible-sector
retention and best-energy gaps measure the observed low-energy tail.
Table~\mainref{tab:hardware_run_summary} shows that every fixed-angle condition
remains far above chance feasibility. At 36, 52, and 64 qubits, warm starts
close \CreditWarmStartGapClosedRangePct{} of the fixed-angle best-energy gap to
the strict-feasible dense classical reference. Matched random-feasible controls
do not close that gap, and the strict-feasible classical reference remains
lower at every width.

Table~\ref{tab:cgxy_fixed_angle_condition_details} lists retained coupling,
exact-budget mass, and strict and repaired best-energy gaps for the 32--64-qubit
conditions. Band-1 repair accepts samples within one count of each budget and
greedily flips bits to restore \((k,m)\), choosing each flip by its immediate
sparse-objective energy change. Repair of the 36-qubit warm-start
condition matches the sparse classical reference. Repair is postprocessing and
does not enter the strict exact-budget mass or BK-rate metrics.

\begin{table}[!htbp]
\centering
\footnotesize
\setlength{\tabcolsep}{3.5pt}
\renewcommand{\arraystretch}{1.12}
\caption[Condition-level Credit Card hardware results]{\emph{Condition-level Credit Card hardware results.} Warm starts reduce the dense gap relative to fixed-angle and random-feasible initialization at 36, 52, and 64 qubits. Band-1 repair matches the sparse classical reference for the 36-qubit warm-start condition. The table reports the hardware-supported fraction of dense sample--feature coupling mass and decoded exact-budget mass. For repeated conditions, exact-budget entries are repetition means [minimum, maximum]. Single-run entries show the observed fraction without brackets. Dense and band-1 gaps are reported in the corresponding QUBO objective's energy units, with lower values better. Dense gap is hardware best minus the strict-feasible dense classical reference. Band-1 gap uses the repaired energy and the sparse classical reference.}
\label{tab:cgxy_fixed_angle_condition_details}
\TableRowsSetup
\begin{tabular}{rrlrrrr}
\TableHideRows
\toprule
Qubits & $p$ & Condition & Retained cross mass & Exact-budget mass & Dense gap & Band-1 gap \\
\midrule
\TableBodyRows
32 & 3 & Fixed-angle & 4.6\% & 56.3\% & +6.3 & +0.2 \\
\midrule
36 & 3 & Fixed-angle & 3.6\% & 47.4\% [47.2\%, 47.5\%] & +19.9 & +1.8 \\
36 & 3 & Warm-start & 3.6\% & 49.1\% [48.7\%, 49.5\%] & +1.3 & +0.0 \\
36 & 3 & Random-feasible & 3.6\% & 50.9\% [49.9\%, 51.5\%] & +38.7 & +8.6 \\
\midrule
52 & 2 & Fixed-angle & 2.3\% & 31.2\% [30.9\%, 31.5\%] & +103.8 & +12.7 \\
52 & 2 & Warm-start & 2.3\% & 35.1\% [34.6\%, 35.4\%] & +18.0 & +4.0 \\
52 & 2 & Random-feasible & 2.3\% & 32.2\% [32.1\%, 32.2\%] & +105.4 & +9.1 \\
\midrule
56 & 2 & Fixed-angle & 2.6\% & 42.2\% & +113.4 & +3.2 \\
\midrule
64 & 2 & Fixed-angle & 2.1\% & 25.5\% [24.9\%, 25.9\%] & +84.9 & +15.2 \\
64 & 2 & Warm-start & 2.1\% & 33.5\% [33.3\%, 33.9\%] & +18.6 & +13.4 \\
64 & 2 & Random-feasible & 2.1\% & 32.3\% [31.9\%, 33.1\%] & +117.9 & +11.3 \\
\bottomrule
\end{tabular}
\end{table}

\FloatBarrier

\end{document}